\documentclass{ieeeaccess}
\usepackage[utf8]{inputenc} 
\usepackage[T1]{fontenc}    
\usepackage{times}
\usepackage{amsmath,amssymb,amsfonts}
\usepackage{algorithmic}
\usepackage{graphicx}
\usepackage[font={sf,scriptsize},           
            labelfont={bf,color=accessblue},
            caption=false]                  
           {subfig} 
\newcommand{\subcaptionbox}[2]{\subfloat[#1]{#2}} 
\usepackage{textcomp}
\usepackage[usenames,dvipsnames]{xcolor} 
\usepackage{tikz}
\NewSpotColorSpace{PANTONE}
\AddSpotColor{PANTONE} {PANTONE3015C} {PANTONE\SpotSpace 3015\SpotSpace C} {1 0.3 0 0.2}
\SetPageColorSpace{PANTONE}
\usepackage{booktabs}       
\usepackage{multirow}       
\usepackage[inline]{enumitem}
\usepackage{orcidlink} 
\usepackage{amssymb} 
\usepackage{amsmath} 
\usepackage{mathrsfs} 
\usepackage{amsthm}
\newtheorem{theorem}{Theorem}[section]

\theoremstyle{definition} 
\newtheorem{definition}{Definition}[section]
\usepackage[ruled,vlined,linesnumbered,algo2e]{algorithm2e}
\SetKwInput{KwInput}{Input}
\SetKwInput{KwOutput}{Output}
\SetKwInput{KwRequire}{Require}
\SetKwComment{Comment}{$\triangleright$~}{}
\SetKwFor{Repeat}{repeat}{times}{end}
\SetNlSty{}{}{:} 
\usepackage{bm}
\usepackage{mathtools}
\DeclarePairedDelimiter{\bra}{\langle}{\rvert} 
\DeclarePairedDelimiter{\ket}{\lvert}{\rangle} 
\DeclarePairedDelimiterX{\braket}[2]{\langle}{\rangle}{#1 \vert #2}
\DeclarePairedDelimiterX{\ketbra}[2]{\lvert}{\rvert}{#1 \rangle\langle #2}
\DeclarePairedDelimiter{\abs}{\lvert}{\rvert} 
\DeclarePairedDelimiter{\paren}{\lparen}{\rparen} 
\DeclarePairedDelimiter{\bparen}{[}{]} 
\DeclarePairedDelimiter{\set}{\{}{\}} 
\DeclarePairedDelimiter{\ceil}{\lceil}{\rceil}

\DeclarePairedDelimiter{\tuple}{\langle}{\rangle} 
\DeclareMathOperator*{\argmax}{arg\,max}

\newcommand{\pfunc}[2]{#1\paren*{#2}} 
\newcommand{\bfunc}[2]{#1\bparen*{#2}} 
\newcommand{\pfuncgiven}[3]{#1\paren*{#2 \;\middle|\; #3}} 
\newcommand{\bfuncgiven}[3]{#1\bparen*{#2 \;\middle|\; #3}} 
\newcommand{\given}[1][]{\:#1\vert\:}
\newcommand{\E}{\mathbb{E}} 
\newcommand{\R}{\mathbb{R}} 
\newcommand{\C}{\mathbb{C}} 
\newcommand{\expect}[1]{\bfunc{\E}{#1}}
\newcommand{\expectgiven}[2]{\bfuncgiven{\E}{#1}{#2}}
\newcommand{\hilbert}{\mathcal{H}} 
\newcommand{\probspace}[1]{\pfunc{\mathscr{P}}{#1}} 
\newcommand{\setaction}{\mathcal{A}} 
\newcommand{\setstate}{\mathcal{X}} 
\newcommand{\setreturn}{\mathcal{Z}} 
\usepackage[nomain,acronym,shortcuts,toc,nonumberlist,section,numberedsection=autolabel,style=super,nogroupskip]{glossaries}
\glsdisablehyper 
\makeglossaries
\usepackage{xspace}

\makeatletter
\newcommand{\acx}{\protect\@acx}%
\newcommand{\@acx}[1]{%
  \ifAC@dua
   \acl{#1}%
  \else
   \expandafter\ifx\csname ac@#1\endcsname\AC@used
      \acs{#1}%
   \else
      \acl{#1}%
   \fi
  \fi
}
\makeatother

\newacronym{qnrl}{QnRL}{quantum-native reinforcement learning}
\newacronym{quak}{QuAK}{quantum amplitude kickback}

\newacronym{drl}{DRL}{distributional reinforcement learning}
\newacronym{qdrl}{QDRL}{quantum distributional reinforcement learning}
\newacronym{qrl}{QRL}{quantum reinforcement learning}
\newacronym{rl}{RL}{reinforcement learning}
\newacronym{mdp}{MDP}{Markov decision process}
\newacronym{nn}{NN}{neural network}
\newacronym{qft}{QFT}{quantum Fourier transform}
\newacronym{vqc}{VQC}{variational quantum circuit}
\newacronym{dft}{DFT}{discrete Fourier transform}
\newacronym{ml}{ML}{machine learning}
\newacronym{dqn}{DQN}{deep Q-network}
\newacronym{qgan}{QGAN}{quantum generative adversarial network}
\newacronym{qml}{QML}{quantum machine learning}
\newacronym{vae}{VAE}{variational autoencoder}
\newacronym{qevae}{QeVAE}{quantum-enhanced variational autoencoder}
\newacronym{qcbm}{QCBM}{quantum circuit Born machine}
\newacronym{cnn}{CNN}{convolutional neural network}
\newacronym{ale}{ALE}{Arcade Learning Environment}
\newacronym{ai}{AI}{artificial intelligence}
\newacronym{qnn}{QNN}{quantum neural network}
\usepackage{hyperref} 
\usepackage{zref-clever}
\zcsetup{
    cap,
    nameinlink=false,
} 
\renewcommand{\label}[1]{\zlabel{#1}}
\newcommand{\cref}[1]{\zcref{#1}}
\newcommand{\Cref}[1]{\zcref[S]{#1}}

\zcRefTypeSetup{figure}{
    abbrev,
}
\zcRefTypeSetup{equation}{
    abbrev,
    Name-sg-ab=, 
    name-sg-ab=, 
    Name-pl-ab=, 
    name-pl-ab=, 
    refbounds={(,,,)}, 
    rangesep={--},
}
\zcRefTypeSetup{algocf}{
    Name-sg={Algorithm},
    name-sg={algorithm},
    Name-pl={Algorithms},
    name-pl={algorithms},
    Name-sg-ab={Alg.},
    name-sg-ab={alg.},
    Name-pl-ab={Algs.},
    name-pl-ab={algs.},
}
\zcRefTypeSetup{appendix}{
    Name-sg={Appendix},
    Name-pl={Appendices},
}
\usepackage[%
backend=biber,
style=ieee,
citestyle=numeric-comp,
sorting=none,
natbib=true,
mincitenames=1,
maxcitenames=1,
url=false,
isbn=false]{biblatex}
\DeclareEmphSequence{\upshape,\it,\bfseries} 
\usepackage{todonotes}

\def\BibTeX{{\rm B\kern-.05em{\sc i\kern-.025em b}\kern-.08em
    T\kern-.1667em\lower.7ex\hbox{E}\kern-.125emX}}
\begin{document}
\history{Date of publication xxxx 00, 0000, date of current version xxxx 00, 0000.}
\doi{10.1109/TQE.2020.DOI}

\newcommand{\titletext}{\acrshort{qnrl}: Quantum-Native Reinforcement Learning}
\title{\titletext}
\author{\uppercase{Alexander DeRieux}\,\orcidlink{0000-0003-1606-0668} \IEEEmembership{Graduate Student Member, IEEE} and \uppercase{Walid Saad}\,\orcidlink{0000-0003-2247-2458} \IEEEmembership{Fellow, IEEE}}

\address{Bradley Department of Electrical and Computer Engineering, Virginia Tech Institute for Advanced Computing, Alexandria, VA 22305 USA}


\markboth
{DeRieux and Saad: \titletext}
{DeRieux and Saad: \titletext}

\corresp{Corresponding author: Alexander DeRieux (email: derieux@vt.edu).}

\begin{abstract}

Quantum reinforcement learning (QRL) is a promising approach to learn effective decision strategies across several applications with stochastic environments. Instead of directly modeling the random variables that govern these environments, existing QRL architectures indirectly approximate environment behavior by estimating expected outcomes, which limits their expressive power and adaptive potential. Overcoming such challenges requires a novel QRL approach that exploits the distributional nature of quantum computers to directly model environment random variables as quantum state distributions. Hence, in this paper, a novel framework dubbed \emph{quantum-native reinforcement learning (QnRL)} is proposed. QnRL is a distributional RL framework that learns conditional distributions naturally in Hilbert space via superimposed and entangled quantum states. Thus, QnRL can directly model the behavior of stochastic learning environments via the natural properties of quantum systems. QnRL accomplishes this via a novel, proposed \emph{quantum amplitude kickback (QuAK)} algorithm that enables comparing the $n$-th power of the $m$-th moment of multiple superimposed distributions. It is theoretically proven that a conditional action policy distribution is distilled from the moments of a quantum generative model entirely within Hilbert space via QuAK, and optimized via QnRL. This complex distribution composition is also shown to provide extra dimensions for expressing environment correlations that are unknown to purely classical and classically-sampled quantum distributional models. Experimental results across diverse environments show that QnRL achieves up to $82.9\%$ higher evaluation scores, with up to $94.3\%$ fewer parameters on average, more accurately estimates the expected return for unseen observations, and better adapts to varying stochastic conditions compared to the baseline. 
\end{abstract}


\titlepgskip=-15pt

\maketitle

\section{Introduction}

\Ac{qrl} is a nascent field exploiting the unique properties of quantum mechanics to develop intelligent \ac{qml} agents that learn models of practical environments through experimentation with potentially greater computational efficiency and reduced time complexity, compared to purely classical, i.e., non-quantum, \ac{rl} \cite{Sutton2018ReinforcementLearningIntroduction} paradigms.
As such, \ac{qrl} admits diverse use cases such as financial trading and stock prediction \cite{Chen2025QuantumEnhancedForecastingDeep, Liu2025QuantumEnhancedReinforcementLearning}, control and navigation \cite{Pozza2022QuantumReinforcementLearning, Chen2023EfficientQuantumRecurrent, Liu2024QTRLPracticalQuantum}, satellite communication and coordination \cite{Park2024QuantumMultiAgentReinforcement, Kim2025QuantumReinforcementLearning}, games \cite{Park2025ItsAMeQuantumMario, Cai2025QuantumenhancedHybridDeep}, and smart factories \cite{Yun2023QuantumMultiAgentActorCritic}.
A common thread that links these applications is the need to learn a stochastic environment, which exhibits random behavior that is governed by probability distributions.
Canonically, this stochastic nature is conveyed to \ac{rl} agents through random rewards received as feedback upon taking an action from a learned decision-making strategy, i.e., a \emph{policy}.
Fundamentally, the agents in \ac{qrl} applications must develop policies that can handle the random nature of their stochastic environments.
Many existing \ac{qrl} architectures \cite{Chen2025QuantumEnhancedForecastingDeep, Liu2025QuantumEnhancedReinforcementLearning, Pozza2022QuantumReinforcementLearning, Chen2023EfficientQuantumRecurrent, Liu2024QTRLPracticalQuantum, Park2024QuantumMultiAgentReinforcement, Kim2025QuantumReinforcementLearning, Park2025ItsAMeQuantumMario, Cai2025QuantumenhancedHybridDeep, Yun2023QuantumMultiAgentActorCritic}, however, are designed using the same paradigms as classical \ac{rl} whereby they rely on modeling the \emph{expected outcome of a random environment}, e.g., the canonical Q-value \cite{Sutton2018ReinforcementLearningIntroduction}, instead of modeling the probability distributions which describe the random environment directly.
These probability distributions contain vital information about the random environment that conveys the full spectrum of possible outcomes, such as its variance and shape, which can serve as auxiliary signals to inform the agent's decision-making strategy.
Drawing upon the analogy of learning the behavior of a hidden, or ``black box'', stochastic process the majority of today's \ac{qrl} approaches \cite{Chen2025QuantumEnhancedForecastingDeep, Liu2025QuantumEnhancedReinforcementLearning, Pozza2022QuantumReinforcementLearning, Chen2023EfficientQuantumRecurrent, Liu2024QTRLPracticalQuantum, Park2024QuantumMultiAgentReinforcement, Kim2025QuantumReinforcementLearning, Park2025ItsAMeQuantumMario, Cai2025QuantumenhancedHybridDeep, Yun2023QuantumMultiAgentActorCritic} develop policies with respect to the average output of the black box instead of the underlying probability distributions that govern its behavior.
As such, existing \ac{qrl} architectures have a fundamentally limited perception of the hidden inner behavior of their stochastic environment because they are not capable of tapping into the auxiliary signals that a full distribution would provide.
This limits the capabilities of current \ac{qrl} agents, such as their sensitivity to risk, propensity for exploration, and adaptive potential, and presents a challenge to design quantum policies that model the probability distributions that govern the stochastic environments they learn.

In the classical domain, the concept of \ac{drl}~\cite{Bellemare2017DistributionalPerspectiveReinforcement, Bellemare2023DistributionalReinforcementLearning} was proposed in order to allow agents to gain a better understanding of their environment. In \ac{drl}, agents can model the probability distributions of expected rewards, called \emph{return distributions}, associated with taking actions based on given environment observations, and thus create more expressive decision strategies. However, \ac{drl} uses \acp{nn} to learn the return distributions and, hence, it is only learning a \emph{hidden representation} of the distribution that requires further downstream processing to interpret. Classical \acp{nn} conceptualize probability distributions as arbitrary neuron activations across the entire space of real numbers, because internally the network architecture does not restrict the learnable range of its output. 
This means that the generated distributions are only projections, or \emph{approximations}, of the random variables they emulate, which fundamentally limits the policy's ability to represent the underlying environment behavior. In essence, \acp{nn} are good function approximators, but they are not random variables. Hence, even when agents use \ac{drl}, there will always be inherent uncertainty in their approximation of unobserved events, and thus some characteristics of the stochastic environment remain obfuscated as a black box.
These limitations of \ac{drl} raise a key question: \emph{How can we fundamentally construct distributional models that more closely represent the physical stochastic processes, or random variables, which govern the environment?}

Answering this question can be more naturally done in the quantum domain, by observing that \emph{quantum states represent probability distributions that exist in the space of complex numbers}. 
In other words, quantum systems fundamentally operate as distributionally-native machines that are physically bound by normalization requirements.
This opens the door for constructing \emph{quantum-native} models that can physically represent and manipulate the random variables that govern the learning environment as a natural process of their operation. 
Furthermore, quantum-native models operate over the complex domain. Hence, they naturally possess an extra degree of freedom, compared to classical \acp{nn}, along the phase dimension of a quantum state to express these learned distributions.
Remarkably, there have been few prior works that explored the concept of \ac{qdrl} \cite{Hu2019DistributionalReinforcementLearning, Cherrat2023QuantumDeepHedging}, and more broadly distribution generation via quantum systems \cite{Zoufal2019QuantumGenerativeAdversarial, Du2020ExpressivePowerParameterized, Sweke2021QuantumClassicalLearnability, Zhu2022GenerativeQuantumLearning, Rao2023LearningHardDistributions, Barthe2025ParameterizedQuantumCircuits, Liu2018DifferentiableLearningQuantum, Iaconis2024QuantumStatePreparation, Sano2025QuantumStatePreparation}. However, all prior distributional analysis in \cite{Hu2019DistributionalReinforcementLearning, Cherrat2023QuantumDeepHedging, Zoufal2019QuantumGenerativeAdversarial, Du2020ExpressivePowerParameterized, Sweke2021QuantumClassicalLearnability, Zhu2022GenerativeQuantumLearning, Rao2023LearningHardDistributions, Barthe2025ParameterizedQuantumCircuits, Liu2018DifferentiableLearningQuantum, Iaconis2024QuantumStatePreparation, Sano2025QuantumStatePreparation} relies on sampling learned quantum states into the classical domain, which, due to quantum state collapse, loses information stored along the complex amplitude and phase dimensions that exist exclusively in the quantum domain. Hence, the quantum approaches in \cite{Hu2019DistributionalReinforcementLearning, Cherrat2023QuantumDeepHedging, Zoufal2019QuantumGenerativeAdversarial, Du2020ExpressivePowerParameterized, Sweke2021QuantumClassicalLearnability, Zhu2022GenerativeQuantumLearning, Rao2023LearningHardDistributions, Barthe2025ParameterizedQuantumCircuits, Liu2018DifferentiableLearningQuantum, Iaconis2024QuantumStatePreparation, Sano2025QuantumStatePreparation} act as interchangeable substitutes for classical \acp{nn}, and therefore face the same limitations as classical \ac{drl}.
Further, the shapes of the output distributions in this prior art are often tied to the topology of their quantum circuit architecture. Therefore, some characteristics of the environment random variables must be known a priori to generate distributions that accurately describe them, which is often neither practical nor feasible in many real-world applications.
From a quantum perspective, these limitations raise two key questions: \begin{enumerate*}
    \item How can we construct quantum distributional algorithms that do not lose information through classical sampling and maintain learned quantum domain features?
    \item How can we design flexible quantum generative model architectures such that they can produce many different probability distribution shapes without a priori environment knowledge?
\end{enumerate*}
Answering these questions presents a challenge that requires exploiting the synergy of quantum computers as distributionally-native machines to solve.
Hence, there is a need for a novel quantum-native \ac{drl} approach that more closely emulates the random variables that govern the stochastic environment, maintains quantum features for distributional analysis, and generates a comprehensive spectrum of distribution shapes, which have not been previously been explored.

\subsection{Related Works}\label{sec:qnrl:related_work}

Few prior works \cite{Hu2019DistributionalReinforcementLearning, Cherrat2023QuantumDeepHedging, Zoufal2019QuantumGenerativeAdversarial, Du2020ExpressivePowerParameterized, Sweke2021QuantumClassicalLearnability, Zhu2022GenerativeQuantumLearning, Rao2023LearningHardDistributions, Barthe2025ParameterizedQuantumCircuits, Liu2018DifferentiableLearningQuantum, Iaconis2024QuantumStatePreparation, Sano2025QuantumStatePreparation} have investigated the generation of probability distributions in quantum systems and its intersection with constructing \ac{qml} models that learn distributions as their target goal.
The prior works can be grouped into two broad categories: \begin{enumerate*}
    \item \ac{qdrl} \cite{Hu2019DistributionalReinforcementLearning,Cherrat2023QuantumDeepHedging}, which is directly related to the \ac{drl} setting; and
    \item distribution generation via quantum circuits \cite{Zoufal2019QuantumGenerativeAdversarial, Du2020ExpressivePowerParameterized, Sweke2021QuantumClassicalLearnability, Zhu2022GenerativeQuantumLearning, Rao2023LearningHardDistributions, Barthe2025ParameterizedQuantumCircuits, Liu2018DifferentiableLearningQuantum, Iaconis2024QuantumStatePreparation, Sano2025QuantumStatePreparation}, which spans a broad range of distributionally-based \ac{qml} applications.
\end{enumerate*}

Regarding the first category of \ac{qdrl}, in \cite{Hu2019DistributionalReinforcementLearning}, the authors propose a \ac{qdrl} approach using a classical \ac{nn} paired with a \ac{vqc} using sufficient qubits to represent the quantiles of the return distribution of all actions and Fock space measurement output (i.e., photon particle counts). Their \ac{nn} encodes a classical environment observation into a matrix of parameters, which are a feature map that of Pauli rotations onto the quantile qubits. The output of the \ac{vqc} are Fock measurements sampled from the quantile qubits, which produces a distribution of quantiles for all actions. Importantly, their approach samples the quantile distributions by measurement into the classical domain.
In \cite{Cherrat2023QuantumDeepHedging}, the authors develop a distributional actor-critic approach for deep hedging. They employ two independent \acp{qnn}, an ``actor'' to select actions, and a ``critic'' to evaluate the actor's decision, that are trained in tandem. In particular, their critic is designed as a \ac{qdrl} network to learn return distributions for environment state transitions over a finite time horizon. Importantly, their critic computes the expected return via an expected value measurement, in the range $[-1, 1]$, with respect to an observable constructed from the return distribution support. The total cumulative return along the trajectory is tracked separately and computed classically, and used to construct an exponentiated return function which serves as the critic loss target.
In both \cite{Hu2019DistributionalReinforcementLearning} and \cite{Cherrat2023QuantumDeepHedging}, information stored in the imaginary phase component of the quantum state is lost during the classical measurement sampling process, which means that any classical downstream analysis is not leveraging learned relationships along this extra dimension in complex space. This downsampling projection into classical space misses the broader gains from architecting a quantum-based distributional solution as they map to commensurate classical approaches.
As such, the solutions of \cite{Hu2019DistributionalReinforcementLearning} and \cite{Cherrat2023QuantumDeepHedging} cannot leverage the complex amplitude and phase components of the quantum distribution for comparison or analysis. 
Moreover, the quantum circuit of \cite{Hu2019DistributionalReinforcementLearning} is designed to replicate the behavior of its purely classical counterpart by outputting the quantiles of all actions in the system at the same time, requiring many qubits for repeated quantiles. Thus, the work in \cite{Hu2019DistributionalReinforcementLearning} also does not leverage the superposition nature of the quantum distributions by overlaying them to reduce the system size and resulting simulation computational overhead.
Additionally, the purpose of the distributional critic in \cite{Cherrat2023QuantumDeepHedging} is to aid in tuning the actor policy network during training, and is removed during deployed operation. Hence, the learned quantum return distributions are never used in deployment to influence the policy decision process. 
This, in turn, means their policy cannot analyze the risk of choosing actions based on observations from which it has had no prior exposure. Thus, their policy cannot adapt to unseen stochastic environment behavior.
Moreover, the target cumulative return is fixed as a exponentiated function of the return rather than the return itself, which forces the critic to tune toward a distribution that is not completely representative of the environment behavior.
Finally, the \ac{vqc} design in \cite{Cherrat2023QuantumDeepHedging} requires many qubits dependent on the time horizon length and the return distribution support in addition to dedicated ancillas, which can significantly increase the quantum system size in environments where long trajectories or large distribution supports are required.

The second category of quantum generative approaches can be further split into three groups: \begin{enumerate*}
    \item quantum state preparation using generative models \cite{Zoufal2019QuantumGenerativeAdversarial, Du2020ExpressivePowerParameterized, Sweke2021QuantumClassicalLearnability, Zhu2022GenerativeQuantumLearning, Rao2023LearningHardDistributions, Barthe2025ParameterizedQuantumCircuits},
    \item quantum Born machines \cite{Liu2018DifferentiableLearningQuantum}, and 
    \item matrix product state creation \cite{Iaconis2024QuantumStatePreparation, Sano2025QuantumStatePreparation}.
\end{enumerate*}
In \cite{Zoufal2019QuantumGenerativeAdversarial}, the authors propose a method for learning and loading random probability distributions into quantum states from classical data samples using a \ac{qgan} \cite{Dallaire-Demers2018QuantumGenerativeAdversarial, Lloyd2018QuantumGenerativeAdversarial}. 
%
In \cite{Du2020ExpressivePowerParameterized}, the authors prove that \acp{vqc} with a simple architecture can generate more expressive classical distributions than classical \ac{nn} models.
In \cite{Sweke2021QuantumClassicalLearnability}, the authors propose a quantum generative model for learning quantum states that closely approximate unknown classical distributions.
%
In \cite{Zhu2022GenerativeQuantumLearning}, the authors propose a method for learning classical joint probability distribution functions using a \ac{vqc} ansatz based on the family of multivariate distributions with uniform marginals.
%
The work in \cite{Rao2023LearningHardDistributions} develops a method for modeling the measurement distributions from unknown quantum systems using a hybrid \ac{qevae}, which enhances the generation of classical distributions via hidden quantum correlations.
%
Meanwhile, the work in \cite{Barthe2025ParameterizedQuantumCircuits} shows that \acp{vqc} are universal generative models for multivariate distributions and develops a method for expectation value sampling via quantum measurements. 
%
In \cite{Liu2018DifferentiableLearningQuantum}, the authors propose a gradient-based learning framework for \acp{qcbm} \cite{Cheng2018InformationPerspectiveProbabilistic, Benedetti2019GenerativeModelingApproach} that learn to approximate known classical distributions via \acp{vqc}. 
%
%
%
%
%
In \cite{Iaconis2024QuantumStatePreparation}, the authors propose a method using matrix product states for encoding Gaussian distributions as the amplitude of quantum states. 
%
Finally, in \cite{Sano2025QuantumStatePreparation}, the authors propose a matrix product state technique for loading classical probability distributions into quantum space using reflection symmetry. 
%
%
In terms of quantum state preparation, all prior works \cite{Zoufal2019QuantumGenerativeAdversarial, Du2020ExpressivePowerParameterized, Sweke2021QuantumClassicalLearnability, Zhu2022GenerativeQuantumLearning, Rao2023LearningHardDistributions, Barthe2025ParameterizedQuantumCircuits, Liu2018DifferentiableLearningQuantum, Iaconis2024QuantumStatePreparation, Sano2025QuantumStatePreparation} focus on the approximation of a classical distribution with the goal of achieving high replication accuracy in the classical space through quantum measurement sampling. Further, methods like those in \cite{Iaconis2024QuantumStatePreparation} and \cite{Sano2025QuantumStatePreparation} require a priori knowledge of the random variable to construct an appropriate \ac{vqc} architecture that generates the distribution.
Moreover, none of the prior works \cite{Zoufal2019QuantumGenerativeAdversarial, Du2020ExpressivePowerParameterized, Sweke2021QuantumClassicalLearnability, Zhu2022GenerativeQuantumLearning, Rao2023LearningHardDistributions, Barthe2025ParameterizedQuantumCircuits, Liu2018DifferentiableLearningQuantum, Iaconis2024QuantumStatePreparation, Sano2025QuantumStatePreparation} use the distributions learned by their quantum models in quantum space for analysis, and instead requires they be sampled into a classical representation for downstream applications. 

In summary, no prior work: \begin{enumerate*}
    \item investigates the preparation of random (i.e,. unknown) probability distributions as quantum states in non-generative settings, that is, where the goal is to achieve high replication accuracy,
    \item considers performing analysis on learned quantum distributions in quantum space to preserve both complex amplitude and phase components,
    \item considers the application of superimposing multiple random distributions in quantum space, nor
    \item designs generalized (i.e., non-bespoke) quantum circuit architectures that can generate quantum probability distributions of unknown shapes.
\end{enumerate*}
In other words, the relationship between quantum states, probability distributions, and their applications, particularly in the context of \ac{drl}, are largely unexplored.
These limitations, if not overcome, would restrict the applicability and potential benefits of quantum-based \ac{rl} architectures, requiring environment pre-interaction phases to tailor \ac{vqc} design to distribution shapes, and incurring information loss in classical downstream analysis.

\subsection{Contributions}

In contrast to the prior art, the main contribution of this paper is a novel framework dubbed \emph{\ac{qnrl}}, that brings a quantum advantage to the \ac{drl} paradigm by exploiting the amplitude normalization constraint of quantum systems to learn and analyze quantum distributions in their native complex space.
Our design allows a \ac{qrl} agent to build quantum-native policy models that generate conditional return distributions for all actions simultaneously as superimposed quantum states, analyze the moments of these distributions in their native complex space, and distill the comparison of these moments to measurement sampling of the quantum circuit. This reduces the distribution generation and action selection processes, which canonically were prone to approximation loss and not easily interpretable, down to choosing the action with the highest sampling frequency. In other words, in \ac{qnrl}, the agent inputs an environment observation into a quantum distributional model and measures the action qubits to retrieve the optimal action. 
The key benefit is that, with \ac{qnrl}, we learn these conditional return distributions as quantum states, which express the full range of the complex amplitude domain, and we derive the action distribution in situ from the quantum distribution of returns, which maintains their complex form and does not incur classical approximation loss.
Because the distribution of returns is learned as a quantum state, this also allows agents the flexibility to intercept and analyze the state of the system, which facilitates the design of more interpretable distributional models.
As will be evident from our analysis, \ac{qnrl} will be shown to learn a more optimal quantum policy model, achieve higher evaluation scores on average, and with significantly fewer parameters compared to baselines. 
%
In summary, our key contributions include:
\begin{itemize}
    \item We propose a novel \emph{\ac{qnrl}} framework which, to the best of our knowledge, is the first quantum-based \ac{drl} approach leveraging the \emph{fundamental nature of quantum superposition and quantum entanglement} to both generate quantum return distributions conditioned on input observations, and form an action policy using these quantum return distributions in situ.
    \item We propose a novel algorithm dubbed \emph{\ac{quak}} which, to the best of our knowledge, is the first quantum algorithm for simultaneously computing and comparing the $n$-th power of the $m$-th moments of multiple complex distributions entirely in quantum Hilbert space, making use of \emph{nonlinear trace non-increasing} Kraus operators. The ability to perform such comparisons unlocks the ability to perform statistical analysis directly in quantum systems without the need to sample into the classical domain, in both \ac{drl} and more generalized statistical analysis task regimes.
    \item We theoretically prove, and empirically demonstrate, that a conditional action policy distribution can be distilled from the moments of a quantum generative distributional model entirely within Hilbert space, thereby preserving learned features along the amplitude and phase dimensions of the quantum states. Moreover, we also prove that through quantum superposition and quantum entanglement \ac{quak} parallelizes both the computation and comparison of distributional moments more efficiently and with fewer computational resources than purely classical or hybrid classical-quantum approaches.
    \item We show that the our quantum distributional model learns observation relationships across the full range of the Hilbert space, generating distributions comprised of both amplitude and phase components that realize complex correlations between multiple random variables, which are unknown to both purely classical and classically-sampled quantum distributional approaches.
    \item We evaluate our proposed \ac{qnrl} framework and \ac{quak} algorithm across many diverse environments that comprise an extensive range of observation spaces and target goals, including classic control, grid-world navigation, and Atari games. We empirically demonstrate that \ac{qnrl} achieves evaluation scores up $82.9\%$ higher, and with up to $94.3\%$ fewer parameters, than the classical baseline. Further, we demonstrate that \ac{qnrl} learns a quantum distributional model that more accurately estimates the expected return for unseen observations.
\end{itemize}
%

The rest of this paper is organized as follows. \Cref{sec:qnrl:system} describes our problem formulation. In \cref{sec:qnrl:method}, we describe our proposed \ac{qnrl} framework and \ac{quak} algorithm. \Cref{sec:qnrl:exp} presents our experiment results and discussion. Finally, \cref{sec:qnrl:conclusion} draws key conclusions.

\section{System Model and Problem Statement}\label{sec:qnrl:system}

In this section, we formulate the design of distributional \ac{rl} agent policies in the classical regime, and then leverage properties of quantum mechanical systems to logically expand this formulation into the quantum domain.

\subsection{Distributional Learning Setting}

We consider a quantum-native \ac{drl} setting that extends upon the classical configuration of \cite{Bellemare2017DistributionalPerspectiveReinforcement, Bellemare2023DistributionalReinforcementLearning, Wiltzer2024DistributionalAnalogueSuccessor}
with a single agent interacting in an environment as described by a \ac{mdp} with 5-tuple $\tuple{\setstate, \setaction, p, r, \eta}_{\textrm{MDP}}$, where $\setstate = \set{x_{i}}_{i=0}^{\abs{\setstate}-1}$ and $\setaction = \set{a_{i}}_{i=0}^{\abs{\setaction}-1}$ are the state and action spaces, $p \colon \setstate \times \setaction \to \probspace{\setstate}$ is the state transition probability kernel with $\probspace{\cdot}$ being the probability space over a given set, $r \colon \setstate \times \setaction \to \R$ is an immediate reward function, and $\eta \in [0,1]$ is a discount factor.
An agent policy $\pi \colon \setstate \to \probspace{\setaction}$ is a map from the state space to a distribution over the action space. 
A trajectory $\tau^{\pi} = \set{\tuple{X_{t}, A_{t}, R_{t}}}_{t=0}^{T-1}$ is the set of random state-action-reward pairs an agent experiences through interacting with the environment over a time horizon of length $T \in [1,\infty)$, where $A_{t} \sim \pfuncgiven{\pi}{\cdot}{X_{t}}$, $R_{t} = \pfunc{r}{X_{t}, A_{t}}$, $X_{0} = x \in \setstate$, and $X_{t+1} \sim \pfuncgiven{p}{\cdot}{X_{t}, A_{t}}$ are random variables.
The \emph{random return} $\pfunc{Z^{\pi}}{x,a} = \sum_{t \in \tau^{\pi}} \eta^{t} R_{t}$ is the discounted sum of immediate rewards over a trajectory $\tau^{\pi}$ as governed by $X_{0} = x \in \setstate$ and $A_{0} = a \in \setaction$ for a given policy $\pi$.

In classical \ac{drl} the goal is typically to optimize a given policy $\pi$ according to $\probspace{\setreturn \given x, a}$, which is the \emph{full distribution of random returns} for an initial state-action pair $x \in \setstate$, $a \in \setaction$, whose support is the fixed set of return-value atoms $\setreturn = \set{z_{i}}_{i=0}^{\abs{\setreturn}-1}$ s.t.\ $Z^{\pi} \in \setreturn$. In other words, \ac{drl} optimizes \emph{distributions} as opposed to expectations of those distributions. This can then be represented according to the Bellman equation \cite{Bellman1957DynamicProgramming} in which the value function
\begin{equation}\label{eq:qnrl:qfunc}
\pfunc{Q^{\pi}}{x,a} \vcentcolon= \expectgiven{Z^{\pi}}{x,a} = \expectgiven{\sum_{t \in \tau} \eta^{t} R_{t}}{x,a},
\end{equation}
%
is learned via the expectation of the first moment, i.e., the mean, of the random variable $Z^{\pi}$ conditioned on an initial state-action pair. The optimal action $a^{*}$ according to policy $\pi$ for some initial state $x \in \setstate$ can therefore be found, in the case of a greedy policy, by selecting the action with the highest expected return according to:
\begin{equation}\label{eq:qnrl:a-star}
    a^{*} = \argmax_{a' \in \setaction} \expectgiven{Z^{\pi}}{x,a'}.
\end{equation}
It is with regard to this \ac{drl} setting from which we frame the problems that we address.

\subsection{Problem Formulation}

The learning process in purely classical \ac{drl} must perform the following four independent steps: \begin{enumerate*}
    \item use an \ac{nn} to generate the return distributions $\probspace{\setreturn \given x, a}~\forall a \in \setaction$ for a given $x \in \setstate$,
    \item compute the expected return via \cref{eq:qnrl:qfunc} for \emph{every action},
    \item perform a comparison between these expectations to find $a^{*}$ via \cref{eq:qnrl:a-star} and then preserve the associated return distribution as the \emph{target distribution} $\probspace{\setreturn \given x, a^{*}}$, and
    \item realign the atoms of the target distribution according to the current trajectory.
\end{enumerate*}
However, there are several limitations of this approach. 
First, the domain of the return distributions in \ac{drl} are real numbers, i.e., $\probspace{\setreturn \given \setstate, \setaction} \subseteq \R$. This formulation cannot model relationships between the observations, actions, and returns that have a complex component, i.e., $\probspace{\setreturn \given \setstate, \setaction} \subseteq \C$, which may be advantageous in certain environments.
Second, the evaluation of the return distributions in \cref{eq:qnrl:a-star} only considers their mean. This is a helpful gauge when the central tendency of the returns is desired, but does not consider other properties of the distributions, such as their variance (or ``spread''), which may be more crucial in different environments. 
Third, generating distributions using classical \acp{nn} equates to optimizing a set of parameters that produce logit activations that only form a probability distribution when followed by a downstream softmax layer. Importantly, the classical \ac{nn} is not the distribution itself, but rather a \emph{hidden representation} of that distribution, which is not easily explainable.
Lastly, the size of the classical \ac{nn} grows exponentially with $\setaction$ and $\setreturn$, which means that if $\setaction$ and $\setreturn$ are large the generation, analysis, comparison, and realignment steps become increasingly resource intensive. Moreover, to learn an optimal policy the model must experience a wide range of environment observations. If $\setstate$ is also large, then this increases the optimization time due to the large \ac{nn}, which also affects the computation complexity of model tuning.

Thus, we summarize the problems that face classical \ac{drl} as \begin{enumerate*}
    \item the real domain cannot capture complex relationships between observations, actions, and returns that exist in imaginary space,
    \item return distribution analysis only considers the mean, or average expected return, which does not account for other attributes such as the variance,
    \item classical \acp{nn} model a hidden representation of the return distributions rather than the distributions themselves, which are not readily interpretable, and
    \item the number of \ac{nn} parameters grows exponentially with the size of the action space and support of the return distribution, thereby becoming increasingly resource intensive when these spaces are large.
\end{enumerate*}
This is where quantum computing comes into play. In particular, we can bring a quantum advantage to the \ac{drl} paradigm by transitioning the generation and analysis processes into quantum space via learning the distribution of returns as superimposed and entangled quantum states.

\subsection{Quantum State Probability Distributions}

In quantum computing, a \emph{qubit} is the quantum mechanical analog to the classical bit, with a $q$-qubit \emph{state} $\ket{\psi}$ and its complex conjugate $\bra{\psi} = (\ket{\psi})^{\dagger}$ represented as the tensor product of $2^{q}$-dimensional unit vectors in complex Hilbert space $\hilbert^{\otimes q} \in \mathbb{C}^{2^{q}}$. All qubit states, here on referred to as \emph{quantum states}, can be expressed as a linear combination of any complete and orthonormal basis $\mathcal{B}\vcentcolon=\set{\ket{b}}_{b\in\mathcal{B}}$ of $\hilbert$ such that $\ket{\psi} = \sum_{i=0}^{\abs{\mathcal{B}}-1} c_{i} \ket{b_{i}}$, where $c_{i} = \braket{b_{i}}{\psi} \in \mathbb{C}$ is the complex amplitude of basis state $\ket{b_{i}}$, $\braket{b_{j}}{b_{i}}_{i=j} = 1$, and $\braket{b_{j}}{b_{i}}_{i \neq j} = 0$.
Quantum states can be more generally represented as density matrices which can either be \emph{pure} or \emph{mixed} (also called \emph{ensembles}) as follows:
\begin{equation}
\rho = \sum_{k} \ketbra{\psi_{k}}{\psi_{k}}, \,\,\, \textrm{s.t.} \,\,\, \rho ~\textrm{is} \begin{cases}
    \textrm{\emph{pure}}~ \textrm{if}~ \bfunc{\textrm{tr}}{\rho \rho^{\dagger}} = 1, \\
    \textrm{\emph{mixed}}~ \textrm{otherwise},
\end{cases}
\end{equation}
where $\bfunc{\textrm{tr}}{\cdot}$ is the trace of a square matrix.
If $\rho$ is pure, it is simply defined as the outer product of the pure state:
\begin{equation}
\rho = \ketbra{\psi}{\psi} = \sum_{i,j=0}^{\abs{\mathcal{B}}-1} c_{i}c_{j}^{*} \ketbra{b_{i}}{b_{j}}.
\end{equation}
The inner product of the quantum pure state is also defined, and like the basis must also satisfy $\braket{\psi}{\psi} = 1$. Hence, we have the following normalization constraint on the complex state amplitudes:
\begin{equation}
\braket{\psi}{\psi}
= \bfunc{\textrm{tr}}{\rho}
= \sum_{j,i=0}^{\abs{\mathcal{B}}-1} c_{j}^{*}c_{i} \braket{b_{j}}{b_{i}}
= \sum_{i=0}^{\abs{\mathcal{B}}-1} \underbrace{\abs{c_{i}}^{2}}_{p(b_{i})}
= 1,
\end{equation}
which is by definition a probability distribution $\probspace{\mathcal{B}}$ over the support of $\mathcal{B}$, where the probability of measuring any basis state is
\begin{equation}
    p(b_{i}) = \bra{b_{i}} \rho \ket{b_{i}} = \abs{c_{i}}^{2}.
\end{equation}
This means that \emph{all quantum states are also inherently probability distributions over the support of a given basis}. 
%
Further, because the quantum state amplitude is complex $c_{i} = \alpha_{i} + j \beta_{i}$ we have flexibility over the composition of the atom probabilities, e.g., choosing $\alpha_{i} = \sqrt{p(b_{i})}$ and $\beta_{i} = 0$, vice versa, or a mixture.
%
Therefore, we can re-frame the amplitude constraint above to express the quantum state as a function of basis probabilities like so:
\begin{equation}
\ket{\psi} = \sum_{i=0}^{\abs{\mathcal{B}}-1} \sqrt{p(b_{i})} \ket{b_{i}}, \,\,\, \textrm{s.t.} \,\,\, \abs{\braket{b_{i}}{\psi}}^{2} = p(b_{i}),
\end{equation}
with density matrix
\begin{equation}
\rho = \sum_{i,j=0}^{\abs{\mathcal{B}}-1} \sqrt{p(b_{i}) p(b_{j})} \ketbra{b_{i}}{b_{j}},
\end{equation} 
which aligns with the quantum state distribution formulation of previous works \cite{Zoufal2019QuantumGenerativeAdversarial, Rao2023LearningHardDistributions, Iaconis2024QuantumStatePreparation, Sano2025QuantumStatePreparation}.
From this we can say that \emph{the magnitude of the amplitude of state $b_{i}$ is the square-root of its measurement probability}.
Moreover, because quantum gates $U \in \set{\mathbb{I}, H, X, Y, Z, \dots} $ are unitary operations (i.e., $U U^{\dagger} = \mathbb{I})$, and thereby preserve the amplitude normalization constraint, \emph{the quantum state as a result of a unitary gate will also be a probability distribution}. This means that quantum gates can be used to manipulate quantum state distributions, and by extension \acp{vqc} can learn to generate these state distributions using parameterized gate operations, such as $\pfunc{R_{U}}{\theta} = \exp{(-j U \theta / 2)}$ $\forall U \in \set{X, Y, Z}$.

These key facts, that \begin{enumerate*}
    \item \emph{quantum states naturally represent probability distributions}, and
    \item \emph{quantum state distributions can be manipulated by parameterized gate operations},
\end{enumerate*}
form the basis our proposed solution for transitioning the classical \ac{drl} process into quantum space by \emph{learning distributions of returns natively as complex quantum states}.
Doing so facilitates learning return distributions that have extra degrees of freedom along the imaginary dimension of the complex space, which allows them to represent correlations with the environment observations that may not be apparent in a purely real classical model. Further, it allows multiple distributions to be superimposed atop one another, which facilitates parallel computation. Lastly, the size of quantum systems are proportional to $\log_{2}(\cdot)$ of their classical counterparts. Specifically, the number of qubits required to represent a basis $\mathcal{B}$ is $q_{\mathcal{B}} = \lceil \log_{2}{\abs{\mathcal{B}}} \rceil \ll \abs{\mathcal{B}}$, which means the number of parameters necessary to prepare a quantum distribution of returns can also be significantly reduced compared to classical approaches.
In summary, there is a need for a novel approach that can leverage the quantum domain as a distributionally-native learning medium, as proposed next.

\section{Proposed Q\MakeLowercase{n}RL framework} \label{sec:qnrl:method}

%
%
%
\Figure[t!](topskip=0pt, botskip=0pt, midskip=0pt)[width=\linewidth]{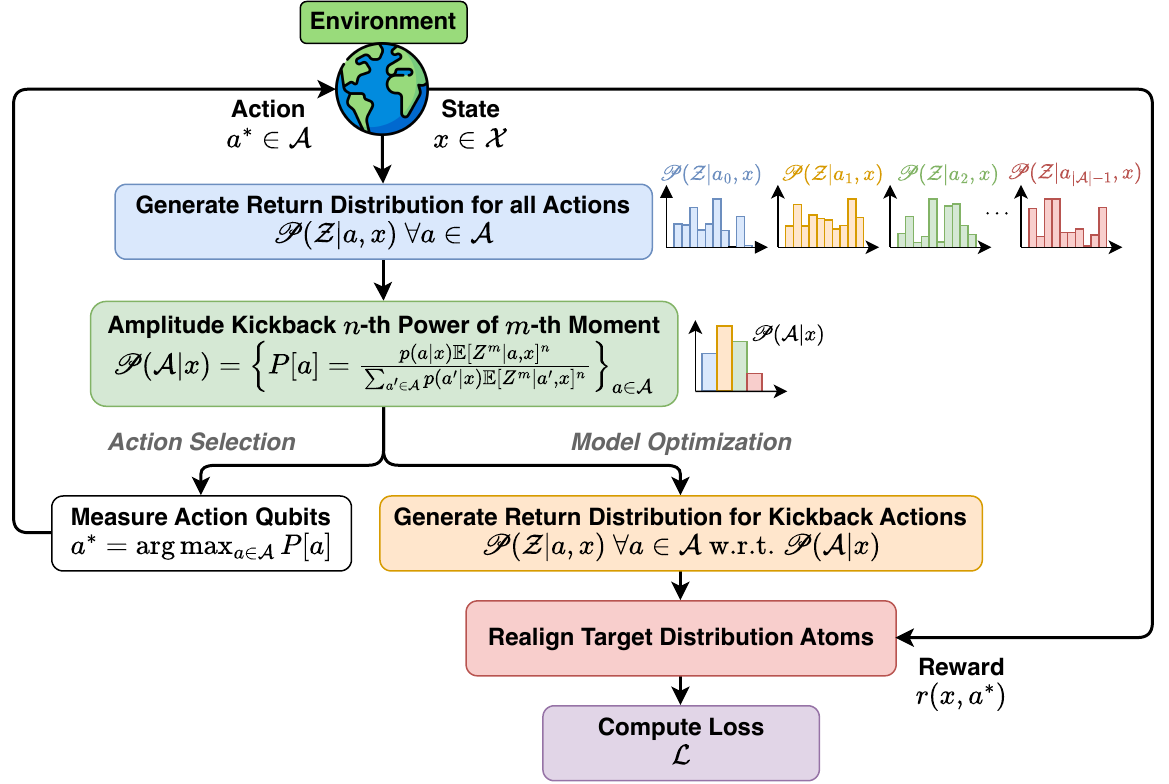}{General design of our \ac{qnrl} framework. A \ac{qnrl} agent generates superimposed quantum conditional return distributions (blue) for all actions given an environment observation, and then runs \ac{quak} algorithm (green) to derive an action distribution relative to the normalized $n$-th power of the $m$-th moment of each conditional return distribution. Action selection (white) is performed by measuring the action qubits. During training, the agent updates the quantum distributional model by generating discounted return distributions for selected target actions (orange), realigns the atoms of this target distribution (red), and then computes policy loss relative to this target distribution (purple).\label{fig:qnrl:framework:overview}}

\begin{figure*}[t!]
\centering
\includegraphics[width=\linewidth]{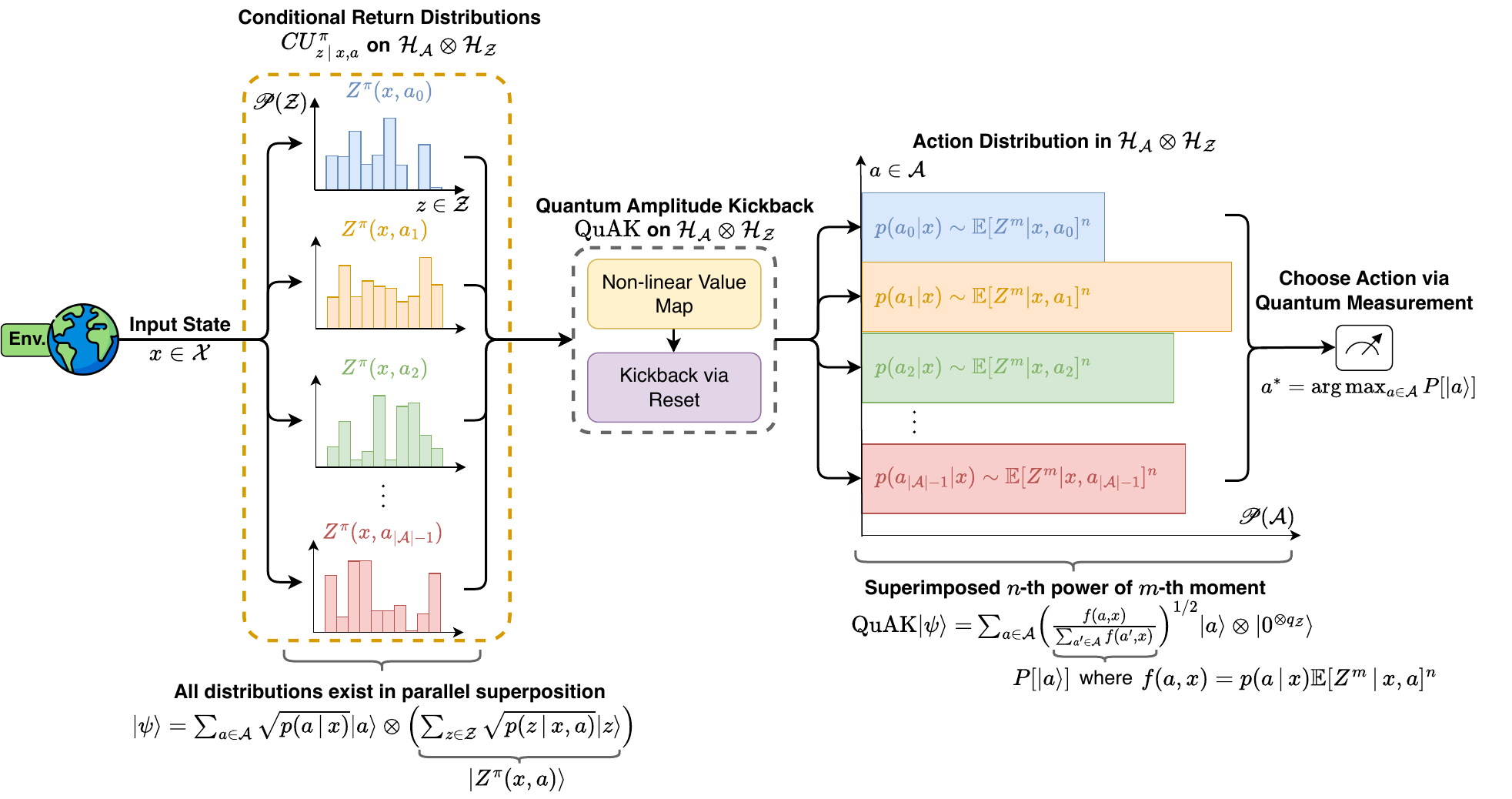}
\caption[Overview of action selection process of \acs{qnrl} using \acs{quak} algorithm.]{High-level view of action selection process in \acs{qnrl} using \acs{quak} algorithm. The agent generates a superimposed quantum state (dashed orange) that is the distribution of returns conditioned on an input observation for every action. The agent then applies a nonlinear value map via Kraus operator (yellow) followed by a quantum reset operation (purple) to map the superimposed distributions of returns into a distribution of actions, where the probability of each action is proportional to the normalized $n$-th power of the $m$-th moment of its return distribution. The agent then measures the action qubits (white) to determine the best action.\label{fig:qnrl:framework:quak_overview}}
\end{figure*}

Our proposed \ac{qnrl} framework is a new approach to \ac{rl} that lies at the intersection of classical \ac{drl} \cite{Bellemare2017DistributionalPerspectiveReinforcement,Bellemare2023DistributionalReinforcementLearning}, and both classical \cite{Mnih2015HumanlevelControlDeep} and quantum value-based \ac{rl} \cite{Skolik2021QuantumAgentsGym}. Inspired by \ac{drl}, we propose to learn a \emph{quantum distributional model of returns} for a set of actions, and then perform statistical analysis on these return distributions to produce a \emph{quantum distribution over the action space}. The keys to our approach, however, are that we learn the distributions of returns as \emph{quantum state probability distributions}, all statistical analysis is performed on these states \emph{natively within quantum space} via nonlinear quantum operators paired with quantum measurements, and that the action policy is formulated by \emph{sampling a quantum distribution of actions derived from our analysis} via quantum circuit measurement.
An overview of our framework is shown in \cref{fig:qnrl:framework:overview}, and a specific view of the action selection phase is shown in \cref{fig:qnrl:framework:quak_overview}.

We consider a combined quantum system of two Hilbert spaces $\hilbert_{\setaction} \otimes \hilbert_{\setreturn} = \C^{2^{q_{\setaction}}} \otimes \C^{2^{q_{\setreturn}}}$ with $q_{\setaction} = \ceil*{\log_{2}{\abs{\setaction}}}$ and $q_{\setreturn} = \ceil*{\log_{2}{\abs{\setreturn}}}$ qubits for the action and return spaces respectively.
From \cref{fig:qnrl:framework:overview}, the two main elements of \ac{qnrl} are the \emph{return distribution generation} phase and the \emph{statistical analysis} phase. All operations for both phases occur within the same quantum system described by $\hilbert_{\setaction} \otimes \hilbert_{\setreturn}$.
First, we prepare quantum return distributions for all actions conditioned on an input observation, i.e., $\forall x \in \setstate: \probspace{\setreturn \given x, a}~\forall a \in \setaction$, via a generalized quantum circuit architecture that we refer to as the \emph{quantum distributional model}. The quantum return distributions for the actions are superimposed, i.e., $\probspace{\probspace{\setreturn}}$, to allow parallel analysis in the next phase.
Next, the expectation of these superimposed quantum return distributions, which we will refer to as their \emph{moments}, are analyzed in situ via a bespoke quantum algorithm using an \emph{amplitude kickback} trick with nonlinear operations, which we refer to as \emph{\ac{quak}}, and transposed into a quantum state distribution of actions conditioned on the input observation, i.e., $\forall x \in \setstate: \probspace{\setaction \given x}$.
This trick enables the \emph{normalized} moments for each $\pfunc{\mathscr{P}}{\mathcal{Z} \given x, a}$ to be encoded as amplitudes within the $\hilbert_{\setaction}$ space. The qubits of $\hilbert_{\setaction}$ can then be sampled to obtain the state-action combination with the \emph{largest moment} relative to others.
Effectively, this reduces the canonical multi-step classical \ac{drl} policy selection procedure to a highly efficient, and distributionally native, quantum circuit algorithm.
In other words, in \ac{qnrl} the policy and return-value distributions are the quantum state. This is very unique compared to classical \ac{drl}, which \emph{models} the return-value distribution through an \ac{nn}, and quantum value-based approaches, which \emph{approximate} the return-value estimate through an expected quantum circuit measurement coupled with a downstream classical \ac{nn}. In \ac{qnrl} there is no intermediary representation to derive the policy, because \emph{the quantum state distribution is the policy}.

During training, our \ac{qnrl} framework is divided into core stages: \begin{enumerate*}
    \item Quantum return distribution generation via the quantum distributional model,
    \item Moment analysis of the quantum return distributions via \ac{quak},
    \item Action selection derived from moment analysis, and
    \item Distributional model optimization via realignment,
\end{enumerate*}
which comprises all of \cref{fig:qnrl:framework:overview}. During execution, only stages 1 through 3 are used, i.e., the 4th stage of model optimization is omitted, which is the right branch of \cref{fig:qnrl:framework:overview}.
\Cref{fig:qnrl:framework:quak_overview} shows a high-level view of the action selection process in \ac{qnrl}. In particular, quantum return distributions are generated via our quantum distributional model in $\hilbert_{\setreturn}$ for individual actions, and then superimposed in $\hilbert_{\setaction} \otimes \hilbert_{\setreturn}$ for all actions, followed by statistical analysis of distribution moments via \ac{quak}, and then action selection via measurement which is derived from this analysis.
To understand how the quantum return distributions are prepared, we next discuss our proposed quantum distributional model.

\subsection{Quantum Distributional Model} \label{sec:qnrl:method:model}

\begin{figure*}[t]
\centering
\includegraphics[width=\linewidth]{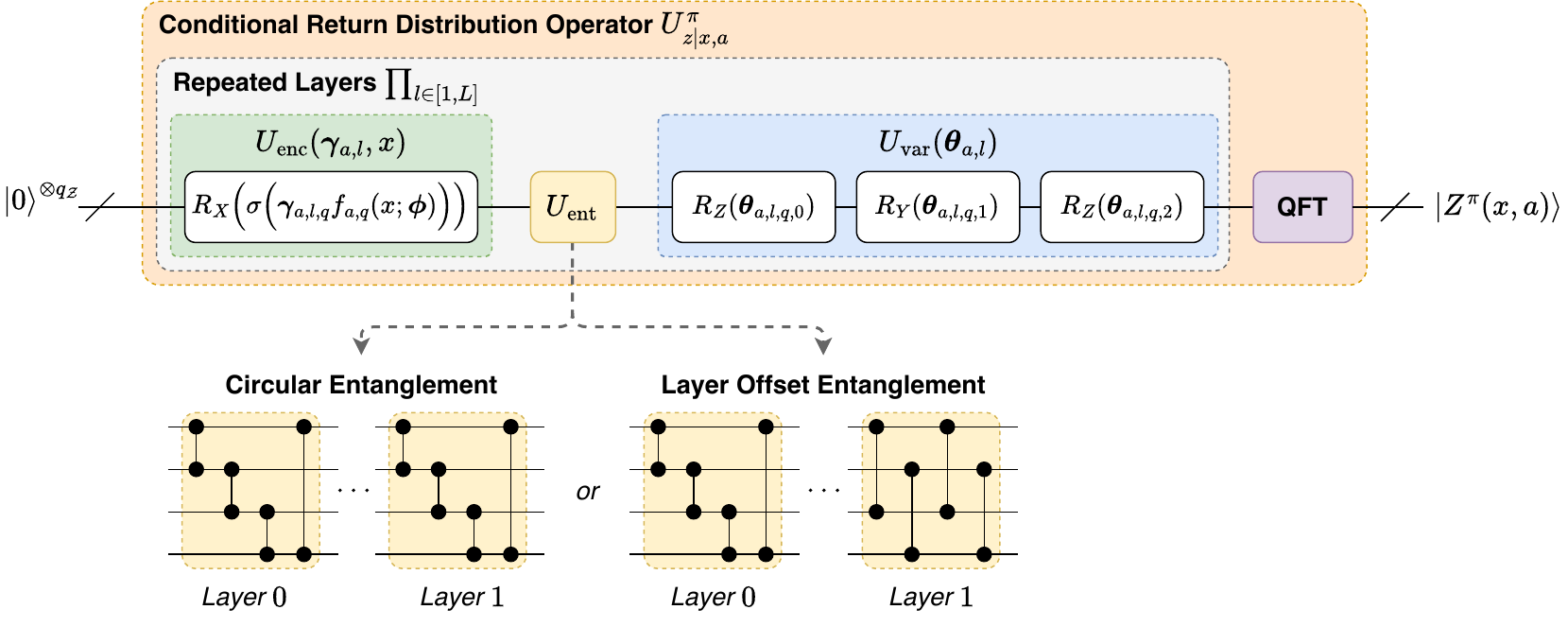}
\caption[Quantum circuit for \acs{qnrl} model.]{Quantum circuit for our proposed conditional return distribution generator $U^{\pi}_{z \,|\, x,a}$, which has cascaded layers of encoding (green), entanglement (yellow), and variational (blue) operators, followed by a final \acs{qft} (purple). In this work we use two variations of entanglement, which are \emph{circular entanglement} (left), and \emph{layer offset entanglement} (right) respectively.\label{fig:qnrl:framework:circuit_model}}
\end{figure*}

At the core of our proposed \ac{qnrl} is a quantum distributional model $U^{\pi}_{z \,|\, x,a}$, which is the first stage in \cref{fig:qnrl:framework:quak_overview}, that learns to prepare the conditional return distribution $\pfuncgiven{\mathscr{P}}{\mathcal{Z}}{x, \mathcal{A}}$ as a quantum state $\ket{Z^{\pi}(x, a)}$ $\forall x \in \setstate, a \in \setaction$. The \ac{vqc} design for this operator is shown in \cref{fig:qnrl:framework:circuit_model}.

The model \ac{vqc} is comprised of $L$ cascaded layers of \emph{encoding}, \emph{entanglement}, and \emph{variational} operators, each with a component that is dependent on the layer index, followed by a \emph{\ac{qft}} as the last operation in the circuit.
First, we define the environment observation encoding operator as follows:
\begin{definition}[Encoding operator]\label{def:qnrl:U_enc}
The \emph{trainable encoding layer} maps the agent's environment observation $\bm{x} \in \R^{\abs{\bm{x}}}$ into a quantum state via the operator
\begin{equation}\label{eq:qnrl:U_enc}
    \pfunc{U_{\textrm{enc}}}{\bm{\gamma}, \bm{\phi}, \bm{x}} = \bigotimes_{d=0}^{q_{\setreturn}-1} \pfunc{R_{X}}{\pfunc{\sigma}{\bm{\gamma}_{a,l,d} \pfunc{f_{a,d}}{\bm{x};\bm{\phi}}}},
\end{equation}
where $a \in \setaction$ is a specific action, $l \in [0,L-1]$ is the layer index, $\bm{\gamma} \in \R^{\abs{\setaction} \times L \times q_{\setreturn}}$ is a matrix of trainable scaling parameters, $\pfunc{f}{\bm{x};\bm{\phi}} = \bm{\phi} \cdot \pfunc{\textrm{flatten}}{\bm{x}}^{T} + b$ is a single classical fully-connected \ac{nn} layer with parameters $\bm{\phi} \in \R^{\abs{\bm{x}} \times \abs{\setaction} q_{\setreturn}}$ and $b$ that maps the environment observation into a shape compatible with the number of actions and number return qubits, and $\sigma : \R \to \R$ is an optional squash activation function.
\end{definition}
Note that $\pfunc{f}{\bm{x};\bm{\phi}}$, the small classical \ac{nn}, is only necessary when operating in an environment with classical observations. This exists to learn a mapping of the classical environment features into quantum Hilbert space.
Importantly, the inclusion of $\bm{\gamma}$ in \cref{eq:qnrl:U_enc} eliminates the dependence of the classical \ac{nn} on the layer depth, dramatically reducing the total number of model parameters when in classical environments.
The encoder is the first operation in the sequence and creates an initial quantum state for the system based upon the environment state. The encoder is followed by an entanglement operator, which binds the state of individual qubits in the system, and is defined as follows:
\begin{definition}[Entanglement operator]\label{def:qnrl:U_ent}
The \emph{non-trainable entanglement layer} links the state of neighboring qubits. We employ two slightly different schemes of entanglement defined by the operator
%
\begin{equation}\label{eq:qnrl:U_ent}
    U_{\textrm{ent}} = \prod_{d=0}^{q_{\setreturn}-1} CZ_{q_{\setreturn}-1-d,D},
\end{equation}
such that
\begin{equation}
    D = \begin{cases}
        q_{\setreturn}-2-d \bmod{q_{\setreturn}-1} & \textrm{if \emph{circular}},
        \\ q_{\setreturn}-2-d-l \bmod{q_{\setreturn}-1} & \textrm{if \emph{offset}},
    \end{cases}
\end{equation}
where $l \in [0,L-1]$ varies the target neighbor qubit in the offset entanglement scheme based on the cascaded layer index, and circular entanglement uses the same target neighbor qubit pairing for every layer.
\end{definition}
We treat the entanglement scheme as a tunable hyperparameter of the model, which as we will see in the experiments (\cref{sec:qnrl:exp}) affects the shape of the learned return distribution and the overall expressiveness of the distributioanl model.
The entanglement operation is then followed by a variational operation defined as follows:
\begin{definition}[Variational operator]\label{def:qnrl:U_var}
The \emph{trainable variational layer} applies a sequence of Pauli Y and Z rotations as given by the operator
\begin{equation}\label{eq:qnrl:U_var}
    \pfunc{U_{\textrm{var}}}{\bm{\theta}} = \bigotimes_{d=0}^{q_{\setreturn}-1} \pfunc{R_{Z}}{\bm{\theta}_{a,l,d,2}} \pfunc{R_{Y}}{\bm{\theta}_{a,l,d,1}} \pfunc{R_{Z}}{\bm{\theta}_{a,l,d,0}},
\end{equation}
where $\bm{\theta} \in [0, 2\pi]^{\abs{\setaction} \times L \times q_{\setreturn} \times 3}$ is a matrix of rotation angle parameters.
\end{definition}
The \ac{vqc} is composed of multiple layers of these encoding, entanglement, and variational operations that immediately follow each other in sequence.
The last operation in the \ac{vqc} is a \ac{qft} that maps a frequency-domain representation of the return distribution to the time domain. The inclusion of the \ac{qft} in our proposed \ac{vqc} model plays a critical role in the expression of the conditional return distribution and is a key result from our experiments. To understand the role of the \ac{qft}, we first state the well-known definition of its classical sibling the \ac{dft}.
\begin{definition}[Classical \acrshort{dft} \cite{VanLoan1992ComputationalFrameworksFast}]\label{def:dft}
The \emph{classical \ac{dft}} maps a vector $\paren{x_{0},\dots,x_{N-1}} \in \C^{N}$ to another vector $\paren{y_{0},\dots,y_{N-1}} \in \C^{N}$ via the relation:
\begin{equation}\label{eq:dft}
    y_{k} = \sum_{n=0}^{N-1} x_{n} \exp{\paren*{-\frac{2 \pi j}{N} k n}}, \quad \forall k \in [0,N-1].
\end{equation}
\end{definition}
The quantum variant builds upon this definition as follows:
\begin{definition}[\acrshort{qft} \cite{Coppersmith2002ApproximateFourierTransform}]\label{def:qft}
In a $q$-qubit quantum system with $N=2^q$ basis states, the \emph{\ac{qft}} maps a quantum state $\ket{x} = \sum_{i=0}^{N-1} x_i \ket{i}$ to another quantum state $\ket{y} = \sum_{i=0}^{N-1} y_i \ket{i}$ via:
\begin{equation}\label{eq:qft:0}
    y_{k} = \frac{1}{\sqrt{N}}\sum_{n=0}^{N-1} x_{n} \exp{\paren*{\frac{2 \pi j}{N} k n}}, \quad \forall k \in [0,N-1].
\end{equation}
If the state $\ket{x}$ is a computational basis state, then the \ac{qft} can be expressed as the operator mapping:
\begin{equation}\label{eq:qft:1}
    \textrm{QFT} : \ket{x} \to \frac{1}{\sqrt{N}}\sum_{n=0}^{N-1} \exp{\paren*{\frac{2 \pi j}{N} x n}} \ket{n}.
\end{equation}
\end{definition}
Notice that sign of the exponents in \cref{eq:qft:0,eq:qft:1} are different from \cref{eq:dft}. Hence, the classical \ac{dft} is equivalent to the inverse quantum operator $\textrm{QFT}^{\dagger}$. Therefore, the inclusion of the \ac{qft} in our model maps a learned frequency representation of the return distribution back into a discrete time-domain signal, which is the actual distribution. This frequency-to-time method results in a more expressive and consistent distribution across multiple diverse training environments.
Building upon \cref{def:qnrl:U_enc,def:qnrl:U_ent,def:qnrl:U_var,def:qft} we can express the complete model \ac{vqc} as a single operator, as follows.
\begin{definition}[Conditional return distribution operator]\label{def:qnrl:U_z_x_a}
Given a quantum system in Hilbert space $\hilbert_{\mathcal{Z}}$, the operator
\begin{equation}\label{eq:qnrl:U_z_x_a}
    U^{\pi}_{z \,|\, x,a} = \textrm{QFT} \prod_{l=0}^{L-1} \pfunc{U_{\textrm{var}}}{\bm{\theta}_{a,l}} U_{\textrm{ent}} \pfunc{U_{\textrm{enc}}}{\bm{\gamma}_{a,l}, \bm{\phi}, \bm{x}},
\end{equation}
prepares the quantum state representing a \emph{conditional distribution of returns} in density matrix form
\begin{equation}
\begin{split}
    U^{\pi}_{z \,|\, x,a} 
    & : \ket{0^{\otimes q_{\setreturn}}} \bra{0^{\otimes q_{\setreturn}}} 
    \to \rho_{Z^{\pi} \given x,a},
\end{split}
\end{equation}
such that
\begin{equation}\label{eq:qnrl:rho_Z_x_a}
\begin{split}
    \rho_{Z^{\pi} \given x,a}
    & = \ket{\pfunc{Z^{\pi}}{x,a}}\bra{\pfunc{Z^{\pi}}{x,a}}
    \\ & = \sum_{z_j, z_l \in \mathcal{Z}} \sqrt{p(z_j \given x, a) p(z_l \given x, a)} \ket{z_j} \bra{z_l},
\end{split}
\end{equation}
is conditioned on a given action $a \in \setaction$ and environment state $x \in \setstate$, where $U_{\textrm{enc}}$, $U_{\textrm{ent}}$, $U_{\textrm{var}}$, and $\textrm{QFT}$ follow from \cref{def:qnrl:U_enc,def:qnrl:U_ent,def:qnrl:U_var,def:qft} respectively.
\end{definition}

Here, we note that \cref{def:qnrl:U_z_x_a} operates on a single Hilbert space of return atoms and is defined for a single arbitrary action. We further leverage quantum superposition to expand the effective range of this operator to both the action and environment state Hilbert spaces by defining a controlled variant that incorporates all possible action combinations. To effectively express this operation, it will be convenient to first define an operation on the action Hilbert space as follows:
\begin{definition}[Action distribution operator]\label{def:qnrl:U_a_x}
Given an initial state of all zeros in the combined Hilbert space $\hilbert_{\mathcal{A}} \otimes \hilbert_{\mathcal{Z}}$, the following operator prepares a \emph{distribution over the action space}:
\begin{equation}\label{eq:qnrl:U_a_x}
\begin{split}
\paren*{U_{a \given x} \otimes \mathbb{I}^{\otimes q_{\setreturn}}}
& : \ket{0^{\otimes q_{\mathcal{A}} + q_{\mathcal{Z}}}}\bra{0^{\otimes q_{\mathcal{A}} + q_{\mathcal{Z}}}} 
%
\\ & \to \rho_{A \given x} \otimes \ket{0^{\otimes q_{\mathcal{Z}}}}\bra{0^{\otimes q_{\mathcal{Z}}}},
\end{split}
\end{equation}
such that
\begin{equation}\label{eq:qnrl:rho_a_x}
    \rho_{A \given x} = \sum_{a_i, a_j \in \mathcal{A}} \sqrt{p(a_i \given x) p(a_j \given x)} \ket{a_i} \bra{a_j}.
\end{equation}
\end{definition}
For example, if an equal superposition of actions is desired, i.e., $p(a \given x) = 1/\abs{\mathcal{A}}$, one could select $U_{a \given x} = H^{\otimes q_{\mathcal{A}}}$, which is simply the Hadamard gate on every action qubit.
Equipped with the definition of the action distribution, we define a variant of \cref{def:qnrl:U_z_x_a} that superimposes the return distributions of all actions in parallel as follows:
\begin{definition}[Controlled conditional return distribution operator]\label{def:qnrl:CU_z_x_a}
Given a quantum state in the combined Hilbert space $\hilbert_{\mathcal{A}} \otimes \hilbert_{\mathcal{Z}}$ representing a distribution over the action set $\mathcal{A}$ according to \cref{eq:qnrl:U_a_x}, the \emph{conditional probability distribution of returns controlled on the action space} is prepared using the following operator:
\begin{equation}\label{eq:qnrl:CU_z_x_a}
\begin{split}
& CU^{\pi}_{z \given x, a} 
: \rho_{A \given x} \otimes \ket{0^{\otimes q_{\mathcal{Z}}}}\bra{0^{\otimes q_{\mathcal{Z}}}}
\to \rho_{Z^{\pi} \given x, A},
\end{split}
\end{equation}
such that
\begin{equation}\label{eq:qnrl:rho_Z_x_A}
    \rho_{Z^{\pi} \given x, A} = \sum_{\substack{a_i, a_k \in \mathcal{A} \\ z_j, z_l \in \mathcal{Z}}} \sqrt{g_{x}(a_i, z_j) g_{x}(a_k, z_l)} \ket{a_i, z_j} \bra{a_k, z_l},
\end{equation}
where $U^{\pi}_{z \given x, a}$ follows from \cref{eq:qnrl:U_z_x_a}, and $g_{x}(a, z) = p(a \given x) p(z \given x, a)$ for brevity. 
\end{definition}
Note that \cref{eq:qnrl:CU_z_x_a} is not strictly unitary because we only include the zero term $\mathbb{I}^{\otimes q_{\setaction}} \otimes \ket{0^{\otimes q_{\mathcal{Z}}}}\bra{0^{\otimes q_{\mathcal{Z}}}}$ on the input. In other words, the effect of applying $CU^{\pi}_{z \given x,a}$ on an all-zero input state in the $\hilbert_{\mathcal{Z}}$ space. While a complete definition for $CU^{\pi}_{z \given x,a}$ would also include effects of all other non-zero states in the $\hilbert_{\mathcal{Z}}$ space (i.e., $\{\ket{i}\bra{i}\}_{i=1}^{2^{\abs{\mathcal{Z}}}-1}$) their effect as a result of applying onto a known state prepared from \cref{eq:qnrl:U_a_x} is inconsequential, i.e., we only consider systems that have a known all-zero starting state. We therefore omit these extra terms in the derivations for brevity.

Now that we have defined all the distributional model components, the next step is to \emph{translate this generative model into an actionable policy} at a high level.
The challenge here is that the learned features of the quantum return distribution, which reside in Hilbert space, are lost upon measurement sampling into the classical domain during traditional hybrid quantum-classical policy derivation methods. Hence, a purely quantum approach is necessary to preserve these complex feature relationships in the policy.
Towards achieving this, we next discuss our proposed \ac{quak} algorithm for generating an action policy using the conditional distributional model entirely in quantum Hilbert space.

\subsection{Proposed Amplitude Kickback Algorithm} \label{sec:qnrl:method:quak}

\begin{figure*}[t!]
\centering
\includegraphics[width=\linewidth]{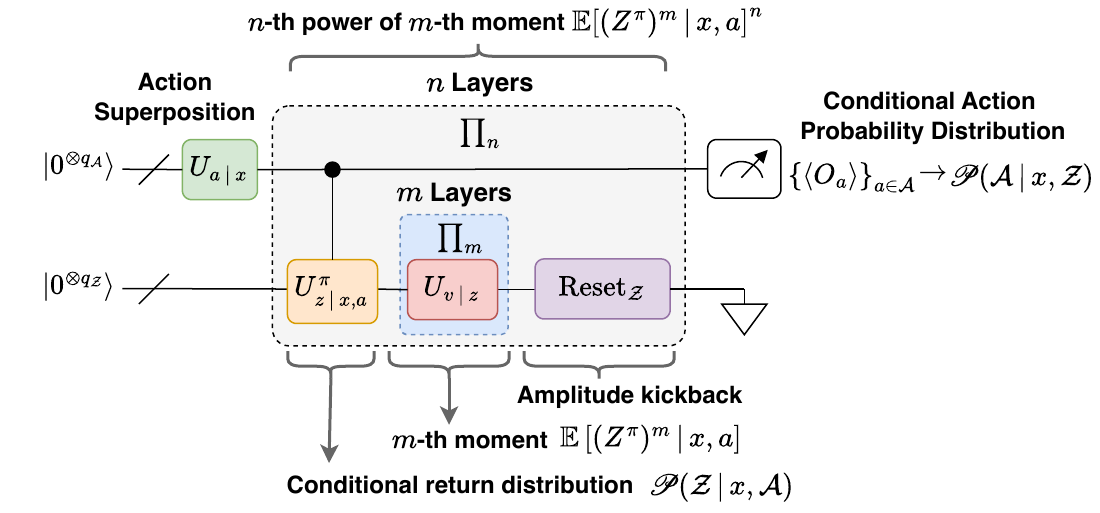}
\caption[Quantum circuit for \acs{quak} algorithm.]{Quantum circuit for \acs{quak} algorithm (gray). An initial action distribution (green) is used to seed the cascaded generation of conditional return distributions (orange), followed by cascaded nonlinear return-value maps via Kraus operator (red), and a final reset over the return qubits (purple) to generate an action distribution that is proportional to the normalized $n$-th power of the $m$-th moment of each respective return distribution. Measurements on the action qubits (white) sample from this weighted distribution.\label{fig:qnrl:framework:quak_circuit}}
\end{figure*}

Our quantum distributional model learns a quantum state given by \cref{eq:qnrl:rho_Z_x_A} that superimposes the conditional return distribution $\probspace{\setreturn \,|\, x, \setaction}$ for all actions. Although this provides a model of the agent's environment reward behavior, this alone is not sufficient to construct a policy. We particularly need a relationship for the reward behavior between actions that can effectively serve as a scoring method. In \cite{Bellemare2017DistributionalPerspectiveReinforcement}, this scoring is done classically through the expectation of the return distribution for every action followed by a selection phase using $\argmax$ as given by \cref{eq:qnrl:qfunc,eq:qnrl:a-star}. This classical process is computationally intensive and does not scale when $\setstate$ and $\setaction$ are large. This is a challenge when generating probability distributions via classical \acp{nn} because the model size grows exponentially proportional to $\setaction$. Likewise, the policy optimization process is also negatively affected by the size of $\setstate$ due to the search space exploring the environment, which increases the number of distributions that must be generated to learn an optimal policy.
To address these challenges, we propose \ac{quak}, a purely quantum algorithm for simultaneously computing the $n$-th power of the $m$-th moment of superimposed random distributions and comparing these moments via quantum measurement entirely in quantum Hilbert space. An overview of this algorithm is shown in \cref{fig:qnrl:framework:quak_overview} and the quantum system architecture design is shown in \cref{fig:qnrl:framework:quak_circuit}. We summarize this algorithm in \cref{alg:qnrl:quak}.
To formally define \ac{quak}, however, it will be convenient to first state the well-known definition for quantum projective measurement:

\begin{definition}[Projective measurement {\cite[Sec.~8.2]{Nielsen2012QuantumComputationQuantum}}]\label{def:qnrl:measurement}
A generalized \emph{quantum projective measurement} $M_{\mathcal{B}}$ is an operation that maps a quantum state in density matrix form $\rho$ to a mixed state using a set of projection operators $\set{\Pi_b}_{b \in \mathcal{B}}$ given by
\begin{equation}\label{eq:qnrl:measurement}
\bfunc{M_{\mathcal{B}}}{\rho} 
= \frac{\bfunc{\mathcal{E}_{\mathcal{B}}}{\rho}}{\bfunc{\textrm{tr}}{\bfunc{\mathcal{E}_{\mathcal{B}}}{\rho}}}
= \frac{\sum_{b \in \mathcal{B}} \Pi_b \rho \Pi^{\dagger}_b}{\bfunc{\textrm{tr}}{\sum_{b \in \mathcal{B}} \Pi_b \rho \Pi^{\dagger}_b}}
= \rho',
\end{equation}
where $\rho \in \mathcal{C}^{2^q}$ is the density matrix for a quantum system of $q$ qubits, $b \in \mathcal{B}$ is the measurement basis, $\bfunc{\mathcal{E}_{\mathcal{B}}}{\cdot}$ is the sum representation for the operators in basis $\mathcal{B}$, and $\bfunc{\textrm{tr}}{\cdot}$ is the trace operator (which is simply the sum of the diagonal elements). Notably, this is the \emph{generalized} form of a projective measurement that includes the trace normalization term in the denominator, which is compatible with all completely positive and trace non-increasing projectors such that $\sum_{b} \Pi_b \Pi^{\dagger}_b \leq \mathbb{I}$, also referred to as \emph{Kraus operators} \cite{Kraus1983StatesEffectsOperations,Heinosaari2012ExtendingQuantumOperations,Wood2015TensorNetworksGraphical,Nielsen2012QuantumComputationQuantum}. For example, if measuring in the Pauli-$Z$ basis, $\Pi_b = \ket{b}\bra{b}$, $\forall b \in \set{0,2^{q}-1}$. 
Further, the post-measurement state of a particular outcome $b$ is defined as
\begin{equation}\label{eq:qnrl:measurement_outcome}
\bfunc{M_{b \in \mathcal{B}}}{\rho} 
= \frac{\Pi_b \rho \Pi^{\dagger}_b}{\bfunc{\textrm{tr}}{\sum_{b \in \mathcal{B}} \Pi_b \rho \Pi^{\dagger}_b}}
= \rho'_{b},
\end{equation}
with probability
\begin{equation}\label{eq:qnrl:measurement_outcome_prob}
\bfunc{P}{\rho'_{b}}
= \bfunc{\textrm{tr}}{\rho'_{b}}
= \frac{\bfunc{\textrm{tr}}{\Pi_b \rho \Pi^{\dagger}_b}}{\bfunc{\textrm{tr}}{\sum_{b \in \mathcal{B}} \Pi_b \rho \Pi^{\dagger}_b}}.
\end{equation}
\end{definition}

Our proposed \ac{quak} also requires the well-known quantum reset operation, which is defined as follows:

\begin{definition}[Qubit reset {\cite[Sec.~12.4]{Nielsen2012QuantumComputationQuantum},~\cite{Wierichs2024IntroductionMidcircuitMeasurements}}] 
In a quantum system spanning two Hilbert spaces $\hilbert_{\mathcal{A}} \otimes \hilbert_{\mathcal{Z}}$, a \emph{reset operation} over subspace $\hilbert_{\mathcal{Z}}$ of an arbitrary state $\rho \in \hilbert_{\mathcal{A}} \otimes \hilbert_{\mathcal{Z}}$ is defined as
\begin{equation}\label{eq:qnrl:reset}
\bfunc{\textrm{Reset}_{\mathcal{Z}}}{\rho}
= \bfunc{\textrm{tr}_{\mathcal{Z}}}{\bfunc{M_{\mathcal{Z}}}{\rho}} \otimes \ket{0^{\otimes q_{\mathcal{Z}}}}\bra{0^{\otimes q_{\mathcal{Z}}}},
\end{equation}
which collapses and conditionally sets qubits $q_{\mathcal{Z}}$ in $\hilbert_{\mathcal{Z}}$ to zero based on the mixed measurement result $\bfunc{M_{\mathcal{Z}}}{\rho}$ over the eigenbasis of $\mathcal{Z}$ according to \cref{eq:qnrl:measurement}, and where $\textrm{tr}_{\mathcal{Z}} : \hilbert_{\mathcal{A}} \otimes \hilbert_{\mathcal{Z}} \to \hilbert_{\mathcal{A}}$ is the partial trace operator which traces out, i.e., removes, the subspace $\hilbert_{\mathcal{Z}}$.
\end{definition}

We can now prove that a quantum generative distributional model can be distilled into an action policy distribution entirely within quantum space by dividing our \ac{quak} process into two parts: \begin{enumerate*}
    \item the core \emph{amplitude kickback} phase, and 
    \item the \emph{policy} (or conditional action probability distribution sampler).
\end{enumerate*}
We will show that \ac{quak} allows quantum return distributions to be manipulated in their native complex form, thus not losing information through classical sampling, and offers a method of parallelization by superimposing them in $\hilbert_{\setaction} \otimes \hilbert_{\setreturn}$, which enables more efficient computation than classical expectation.
\begin{theorem}[\ac{quak}]\label{thm:qnrl:quak:moment}
In a quantum system with two Hilbert spaces $\hilbert_{\setaction} \otimes \hilbert_{\setreturn}$ with $q_{\setaction}$ and $q_{\setreturn}$ qubits representing action and return distributions, the conditional $m$-th moment of the return distribution raised to an arbitrary power $n$, i.e., $\bfunc{\mathbb{E}}{(Z^{\pi})^{m} \given x,a}^{n}$ $\forall a \in \setaction$, can be simultaneously computed across all actions entirely in the combined quantum space using the operator
\begin{equation}\label{eq:qnrl:quak:0}
\begin{split}
    & \textrm{QuAK} =
    \\ & \underbrace{\prod_{n} \vphantom{\prod_{m}}}_{\textrm{Power}} \Bigl[ \underbrace{\textrm{Reset}_{\setreturn} \vphantom{\prod_{m}}}_{\textrm{Kickback}} \underbrace{\prod_{m} \bigl[ \mathbb{I} \otimes U_{v \given z} \bigr]}_{\textrm{Moment}} \underbrace{CU^{\pi}_{z \given x,a} \vphantom{\prod_{m}}}_{\textrm{Return dist.}} \Bigr] \bigl(\underbrace{U_{a \given x} \otimes \mathbb{I} \vphantom{\prod_{m}}}_{\textrm{Action dist.}}\bigr),
\end{split}
\end{equation}
which is the quantum state map
\begin{equation}\label{eq:qnrl:quak:1}
\begin{split}
    \textrm{QuAK}
    & : \ket{0^{\otimes q_{\mathcal{A}} + q_{\mathcal{Z}}}}\bra{0^{\otimes q_{\mathcal{A}} + q_{\mathcal{Z}}}} 
    \\ & \to \rho_{A \given x,Z^{\pi}} \otimes \ket{0^{\otimes q_{\setreturn}}} \bra{0^{\otimes q_{\setreturn}}}
\end{split}
\end{equation}
such that
\begin{equation}\label{eq:qnrl:rho_A_x_Z}
    \rho_{A \given x,Z^{\pi}}
    = \frac{\sum_{a \in \mathcal{A}} p(a \given x) \bfunc{\mathbb{E}}{(Z^{\pi})^{m} \given x,a}^{n} \ket{a}\bra{a}}{\sum_{a' \in \mathcal{A}} p(a' \given x) \bfunc{\mathbb{E}}{(Z^{\pi})^{m} \given x,a'}^{n}}
\end{equation}
is the distribution of actions proportional to the $n$-th power of the $m$-th moment of respective returns, and
\begin{align}\label{eq:qnrl:U_v_z}
\begin{split}
    U_{v \given z} 
    &= \sum_{z \in \mathcal{Z}} \sqrt{v_z} \ket{z}\bra{z},
    \\& \forall \set{v_{z}}_{z \in \mathcal{Z}} = \pfunc{G}{\mathcal{Z}} ~\textrm{s.t.}~ \sum_{z \in \mathcal{Z}} v_z = 1,
\end{split}
\end{align}
is a \emph{completely positive}, \emph{self-adjoint}, and \emph{trace non-increasing} value-encoding operator for the return-atom set $\mathcal{Z}$, and $\set{v_{z}}_{z \in \mathcal{Z}} = \pfunc{G}{\mathcal{Z}}$ is a probability generating function for the return values such that $\sum_{z \in \mathcal{Z}} v_z = 1$, i.e., the return values form a probability distribution.
\end{theorem}
\begin{proof}
See \cref{app:qnrl:proof:quak:moment}.
\end{proof}

From \cref{thm:qnrl:quak:moment}, we observe that quantum probability distributions prepared in $\hilbert_{\setreturn}$, and superimposed in $\hilbert_{\setaction} \otimes \hilbert_{\setreturn}$, can be translated into a distribution of their moments in $\hilbert_{\setaction}$ using only quantum operations. In other words, we can perform statistical analysis using the moments of quantum distributions directly within their native complex space, without classical sampling or estimation.

In \ac{quak} there is one \ac{qrl} agent that has one quantum system of at least two Hilbert spaces $\hilbert_{\setaction} \otimes \hilbert_{\setreturn} = \C^{2^{q_{\setaction}}} \otimes \C^{2^{q_{\setreturn}}}$ with $q_{\setaction} = \ceil*{\log_{2}{\abs{\setaction}}}$ and $q_{\setreturn} = \ceil*{\log_{2}{\abs{\setreturn}}}$ qubits for the action and return spaces respectively. 
The algorithm is a \ac{vqc} that begins by first preparing a quantum state representing an initial action distribution $\probspace{\setaction \given x}$ conditioned on their current environment state $x \in \setstate$ according to \cref{eq:qnrl:U_a_x}, which we refer to as $\rho_{A \given x}$ as given by \cref{eq:qnrl:rho_a_x}. Generally, it is unknown which action will produce a higher expected return a priori, hence the most common starting state will be equal superposition of all actions via the Hadamard operation across all action qubits. This is not a limitation, however, as any starting action distribution could be used. Also, note that this initial action distribution only depends on the environment state, not the return distribution.

The agent then employs the quantum distributional model given by \cref{eq:qnrl:CU_z_x_a} with parameters $\bm{\gamma}$, $\bm{\phi}$, and $\bm{\theta}$ for the encoding and variational layers. This operator generates return distributions $\probspace{\setreturn \given x, \setaction}$ conditioned on the current environment observation $x$ and the action set, thus it is controlled by the action qubits and operates on the return qubits respectively. Importantly, these distributions are superimposed within $\hilbert_{\setaction} \otimes \hilbert_{\setreturn}$, i.e., coexisting in the same space. 
From a machine learning perspective, this controlled nature makes \cref{eq:qnrl:CU_z_x_a} behave as if it were separate \ac{nn} models for each action.
This is unique compared to classical \ac{drl} where a single \ac{nn} generates the return distributions for all actions at once as a large matrix, which is both computationally intensive and wasteful if only a subset of actions are desired. Using \cref{eq:qnrl:CU_z_x_a}, we can selectively generate a distribution for any given action, and by superimposing multiple actions in $\hilbert_{\setaction}$ we can generate multiple distributions that occupy the same space. 
In terms of \acs{ml} optimization, this controlled nature means only a subset of the weights in \cref{eq:qnrl:CU_z_x_a} are tuned during the training process based upon the action with the highest expected return, i.e., the gradient of non-selected parameters are zero, which further reduces the computational burden compared to classical \ac{drl}.

Building upon \cref{sec:qnrl:system}, this controlled strategy can be further expanded in discrete environment state spaces $\setstate$ by superimposing the environment state onto a third Hilbert space $\hilbert_{\setstate}$ and controlling \cref{eq:qnrl:CU_z_x_a} on the combined space $\hilbert_{\setstate} \otimes \hilbert_{\setaction}$. This expansion would be useful in applications such as graph optimization where sets of state and action combinations could be superimposed to find optimal navigation routes.

The distribution model is followed by a moment-generating operation $U_{v \given z}$ from \cref{eq:qnrl:U_v_z} which is both a critical phase in and a key finding of our proposed algorithm. 
To evaluate the expectation of return atoms we must represent the \emph{values} of the return distribution $\set{z}_{z \in \setreturn}$ in quantum Hilbert space. In typical \ac{rl} settings, the \emph{raw values} of the returns are unbounded, and thus cannot be linearly mapped to a quantum operation.
In our \ac{qnrl} setting, however, we are interested in the \emph{comparison} between expectation of returns amongst the action set. Because the set of return atoms is both discrete and explicitly defined in our setting, the range of the return values is known a priori. In light of this, we can apply a normalization to the raw return values to preserve their relative magnitudes, thereby allowing us to represent them in quantum space as a bounded quantum operator.
We therefore define \cref{eq:qnrl:U_v_z}, which is an operator over $\hilbert_{\setreturn}$ that encodes the return values into a quantum state that is both \emph{diagonal} in the $\setreturn$ basis and \emph{self-adjoint}, i.e., Hermitian, where $\pfunc{G}{\setreturn}$ is a probability generating function such that the return values form a probability distribution. We will use the following generating function:
\begin{equation}
    G(\setreturn) = \set*{v_{z} = \frac{z - \min{\setreturn}}{\sum_{z' \in \setreturn} z' - \min{\setreturn}}}_{z \in \setreturn}.
\end{equation}

Because the return values are real and form a probability distribution, from \cref{def:qnrl:measurement}, we can see that $U_{v \given z}$ is \emph{completely positive}, \emph{self-adjoint}, and $U_{v \given z}^{\dagger} U_{v \given z} \leq \mathbb{I}$, i.e., \emph{trace non-increasing}, which means that $U_{v \given z}$ is a valid Kraus operator \cite{Kraus1983StatesEffectsOperations,Heinosaari2012ExtendingQuantumOperations,Wood2015TensorNetworksGraphical,Nielsen2012QuantumComputationQuantum} and can therefore be implemented on quantum hardware \cite{Schlimgen2022QuantumStatePreparation}.
Although $U_{v \given z}$ is self-adjoint, it is \emph{not unitary}. This means the resulting state after solely applying $U_{v \given z}$ will not be normalized. This normalization is resolved, however, by subsequently applying the $\textrm{Reset}_{\setreturn}$ operation \cref{eq:qnrl:reset} on the return space $\hilbert_{\setreturn}$.
This reset operation also serves a dual purpose as the act of resetting the return qubits performs a trace over $\hilbert_{\setreturn}$ which is equivalent to the expectation of returns $\expectgiven{Z^{\pi}}{x,a}$ $\forall a \in \setaction$ disguised as normalized quantum state amplitudes. Because of this trace behavior, we can reapply the value operation $U_{v \given z}$ an arbitrary number of times $m$ to alter the expectation into an arbitrary moment computation, i.e., the random variable raised to the power of $m$, given by $\expectgiven{(Z^{\pi})^{m}}{x,a}$ $\forall a \in \setaction$. These normalized moment amplitudes are ``kicked back'' to the action space $\hilbert_{\setaction}$ as a result of the reset, which re-weights the action probability distribution according to the relative moment of returns for each action.
To further spread this updated action distribution we can reapply the previous conditional return distribution generation \cref{eq:qnrl:CU_z_x_a}, moment computation \cref{eq:qnrl:U_v_z}, and reset operations \cref{eq:qnrl:reset} an arbitrary number of times $n$ to effectively raise each moment to the power of $n$, which is $\expectgiven{(Z^{\pi})^{m}}{x,a}^{n}$ $\forall a \in \setaction$, resulting in the state $\rho_{A \given x,Z^{\pi}}$ as given by \cref{eq:qnrl:rho_A_x_Z}.

It follows from \cref{thm:qnrl:quak:moment} that the action distribution conditioned on the moment of returns, which in our \ac{drl} setting is the agent \emph{policy}, can simply be found by sampling the action qubits after applying \cref{eq:qnrl:quak:0} as we show next.

\begin{theorem}[Action policy distribution]\label{thm:qnrl:quak:policy}
Sampling the action space $\hilbert_{\setaction}$ via quantum measurement after applying amplitude kickback from \cref{thm:qnrl:quak:moment} forms the conditional action probability distribution, or \emph{policy}, as follows:
\begin{equation}\label{eq:qnrl:dist_A_x_Z}
    \pfunc{\pi}{A \given x, Z} 
    \vcentcolon= \probspace{\setaction \given x, \setreturn} 
    = \set*{P[\rho_{a \given x,Z^{\pi}}]}_{a \in \setaction},
\end{equation}
such that the probability of measuring action $a \in \setaction$ is
\begin{equation}\label{eq:qnrl:prob_a_x_Z}
\begin{split}
    P[\rho_{a \given x,Z^{\pi}}]
    & = \bfunc{\textrm{tr}}{\bfunc{M_{a \in \setaction}}{\rho_{A \given x,Z^{\pi}}}}
    \\ & = \frac{p(a \given x) \bfunc{\mathbb{E}}{(Z^{\pi})^{m} \given x,a}^{n}}{\sum_{a' \in \mathcal{A}} p(a' \given x) \bfunc{\mathbb{E}}{(Z^{\pi})^{m} \given x,a'}^{n}},
\end{split}
\end{equation}
where $P[\cdot]$ follows from \cref{eq:qnrl:measurement_outcome_prob}, and $\rho_{A \given x,Z^{\pi}}$ from \cref{eq:qnrl:rho_A_x_Z}.
\end{theorem}
\begin{proof}
See \cref{app:qnrl:proof:quak:moment}.
\end{proof}

\begin{algorithm2e}[t!]
\caption{\ac{quak} algorithm.}\label{alg:qnrl:quak}
\scriptsize
\KwRequire{Quantum Hilbert space $\hilbert_{\setaction} \otimes \hilbert_{\setreturn}$ with qubits $q_{\setaction}$ and $q_{\setreturn}$, input environment state $x \in \setstate$, initial quantum system state $\ket{0^{\otimes q_{\setaction} + q_{\setreturn}}}$}
Apply $U_{a \,|\, x} \otimes \mathbb{I}$ from \cref{eq:qnrl:U_a_x} to prepare initial action distribution $\pfunc{\mathscr{P}}{\setaction \,|\, x}$\;
\Comment{$n$-th power of distribution moment}
\Repeat{$n$}{
    Apply $CU^{\pi}_{z \,|\, x,a}$ from \cref{eq:qnrl:CU_z_x_a} to generate superimposed return distributions conditioned on the action distribution $\pfunc{\mathscr{P}}{\setreturn \,|\, x, a}~\forall a \in \setaction$\;
    \Comment{$m$-th moment of distribution}
    \Repeat{$m$}{
        Apply $U_{v \,|\, z}$ from \cref{eq:qnrl:U_v_z} to map the return values onto the superimposed return distributions\;
    }
    \Comment{Amplitude kickback}
    Apply $\textrm{Reset}_{\setreturn}$ from \cref{eq:qnrl:reset} to compute expectation of the $m$-th moment of the superimposed return distributions $\expectgiven{(Z^{\pi})^{m}}{x,a}~\forall a \in \mathcal{A}$ and reset the return qubits to the zero state\;
}
Measure $\hilbert_{\setaction}$ to sample action from policy distribution $a \sim \pfunc{\pi}{\cdot \given x, Z}$ with probability $\bfunc{P}{\rho_{a \given x,Z^{\pi}}}$ from \cref{eq:qnrl:prob_a_x_Z}\;
\end{algorithm2e}

From \cref{thm:qnrl:quak:policy}, we make a fundamental observation: multiple quantum probability distributions superimposed in $\hilbert_{\setaction} \otimes \hilbert_{\setreturn}$ can be compared via their moments \emph{entirely in quantum Hilbert space} without needing to encode or decode between classical computation part-way through the quantum algorithm. This is done by running our proposed \ac{quak} algorithm \cref{eq:qnrl:quak:0} and then sampling the action space $\hilbert_{\setaction}$ with probability \cref{eq:qnrl:prob_a_x_Z}, as shown in \cref{alg:qnrl:quak}. Notably, this is in contrast with previous \ac{qrl} works that pair a \ac{vqc} with a downstream classical \ac{nn} to estimate the policy distribution return Q-value \cite{Jerbi2021ParametrizedQuantumPolicies,Skolik2021QuantumAgentsGym,Hu2019DistributionalReinforcementLearning}. In these works the \ac{vqc} is effectively a replacement for a classical \ac{nn} and has no direct interpretation on the data itself. In our proposed approach, however, \cref{eq:qnrl:U_z_x_a} generates the actual distribution of returns as the quantum state \cref{eq:qnrl:rho_Z_x_a}. As such, one can reproduce a classical representation of either
\begin{enumerate*}
    \item the conditional distribution of returns $\probspace{\setreturn \given x, a}$ by measuring $\hilbert_{\setreturn}$ after applying \cref{eq:qnrl:U_z_x_a} for any $a \in \setaction$, or
    \item the conditional distribution of actions $\probspace{\setaction \given x, \setreturn}$ by measuring $\hilbert_{\setaction}$ after applying \cref{eq:qnrl:quak:0}. 
\end{enumerate*}
Phrased simply, we can interpret our quantum model more as a distribution sampler, because \emph{the quantum state is the distribution}. Further, we can choose to either \emph{generate a distribution of returns} for a specific action, or \emph{generate a distribution of actions} based upon their moment of returns.
Importantly, \ac{quak} is a generalized distribution algorithm that is not specifically tied to \ac{rl} applications. As such, it can be applied to other use cases where distributions must be generated and their moments computed and compared.

Combining \cref{thm:qnrl:quak:moment,thm:qnrl:quak:policy} we form the basis of our \ac{qnrl} framework for training quantum distributional models, which we discuss next.

\subsection{Quantum-native RL Algorithm} \label{sec:qnrl:method:qnrl}

Our \ac{qnrl} is a learning algorithm for training the quantum distributional model outlined in \cref{eq:qnrl:CU_z_x_a} using the action selection policy in \cref{thm:qnrl:quak:policy}. Particularly, it is a quantum-native variation of the classical \ac{drl} algorithm C51 \cite{Bellemare2017DistributionalPerspectiveReinforcement}, where the moment generation and action selection phases are performed in quantum space. The algorithm is shown in \cref{alg:qnrl:qnrl}.
In \ac{qnrl} there is a single agent that employs the quantum distributional model architecture given by \cref{eq:qnrl:CU_z_x_a} with two sets of parameters for a \emph{policy} model $\bm{\Theta}_{\textrm{policy}} = \tuple{\bm{\gamma}_{\textrm{policy}}, \bm{\phi}_{\textrm{policy}}, \bm{\theta}_{\textrm{policy}}}$, and a \emph{target} model $\bm{\Theta}_{\textrm{target}} = \tuple{\bm{\gamma}_{\textrm{target}}, \bm{\phi}_{\textrm{target}}, \bm{\theta}_{\textrm{target}}}$, and both sets of parameters are initialized to the same values $\bm{\Theta}_{\textrm{policy}} = \bm{\Theta}_{\textrm{target}}$ at the start of training.
The agent interacts within the environment according to the quantum distributional policy at each time step $t \in [0,T-1]$ by first receiving state $x_{t}$ and then applying our proposed \ac{quak} in \cref{alg:qnrl:quak} to sample the action distribution associated with $\bm{\Theta}_{\textrm{policy}}$ via $a_{t} = \arg\max_{a \in \setaction} \pi_{\textrm{policy}}(a \given x_{t}, Z_{\bm{\Theta}_{\textrm{policy}}})$, from which the environment produces both a reward and next state $r_{t}$ and $x_{t+1}$ respectively. These values are stored in the buffer $\mathcal{D}$ as a transition at each time step $\tuple{x_{t}, a_{t}, r_{t}, x_{t+1}}$.
After a sufficient number of transitions have been stored, the policy is updated using a subset, i.e., minibatch, of transitions. First, the target action is selected by applying \cref{alg:qnrl:quak} using $\bm{\Theta}_{\textrm{target}}$ via $a^{*}_{\tau} = \arg\max_{a \in \setaction} \pi_{\textrm{target}}(a \given x_{\tau+1}, Z_{\bm{\Theta}_{\textrm{target}}})$. This target action is then used to generate the return distribution $\mathscr{P}(\setreturn \given x_{\tau+1},a^{*}_{\tau},\bm{\Theta}_{\textrm{target}})$ via \cref{eq:qnrl:CU_z_x_a}. We then update the atoms of this target distribution according to the discounted reward $\setreturn_{\textrm{target}} = \set{r_{\tau} + \eta z}_{z \in \setreturn}$, and then perform the projection procedure of \cite[Section 4.2]{Bellemare2017DistributionalPerspectiveReinforcement} to realign the atoms of $\setreturn_{\textrm{target}}$ onto the support of $\setreturn$. We then generate the policy return distribution conditioned on the current action $\mathscr{P}(\setreturn \given x_{\tau},a_{\tau},\bm{\Theta}_{\textrm{policy}})$ via \cref{eq:qnrl:CU_z_x_a}. Finally, we compute the policy model loss as the cross-entropy between the target $\mathscr{P}(\setreturn \given x_{\tau+1},a^{*}_{\tau},\bm{\Theta}_{\textrm{target}})$ and policy $\mathscr{P}(\setreturn \given x_{\tau},a_{\tau},\bm{\Theta}_{\textrm{policy}})$ return distributions. 
The final phase of \cref{alg:qnrl:qnrl} periodically updates the target model parameters by incrementally moving $\bm{\Theta}_{\textrm{target}}$ toward $\bm{\Theta}_{\textrm{policy}}$ via a soft update rule.

The quantum advantage of our proposed \cref{alg:qnrl:qnrl} arises from the action selection and return distribution generation phases. In the classical C51 algorithm, an \ac{nn} generates all return distributions for all actions at once as a multidimensional matrix. Action selection then requires the expectation of these distributions be computed, followed by a comparison operator. In contrast, our proposed \ac{quak} reduces action selection to sampling a quantum circuit on $\hilbert_{\setaction}$, because the action state amplitudes consist of the normalized moments of each return distribution due to amplitude kickback. Concisely, the distribution generation and moment computation phases occur entirely in quantum Hilbert space, and the comparison phase reduces to quantum measurement.
Furthermore, the policy update is also simplified because our quantum distributional model generates a return distribution directly for a given action. We can also retrieve this return distribution by sampling the quantum circuit on $\hilbert_{\setreturn}$ because \emph{the learned quantum state is the distribution of returns}.

\begin{algorithm2e}[t!]
\caption{\ac{qnrl} algorithm.}\label{alg:qnrl:qnrl}
\scriptsize
\KwRequire{Set of environment states $\setstate$, actions $\setaction$, and returns $\setreturn$, quantum system $\hilbert_{\setaction} \otimes \hilbert_{\setreturn}$ with qubits $q_{\setaction}$ and $q_{\setreturn}$, policy parameters $\bm{\Theta}_{\textrm{policy}} = \tuple{\bm{\gamma}_{\textrm{policy}}, \bm{\phi}_{\textrm{policy}}, \bm{\theta}_{\textrm{policy}}}$, target parameters $\bm{\Theta}_{\textrm{target}} = \tuple{\bm{\gamma}_{\textrm{target}}, \bm{\phi}_{\textrm{target}}, \bm{\theta}_{\textrm{target}}}$, and replay buffer $\mathcal{D}$.}
Initialize $\bm{\Theta}_{\textrm{policy}} = \bm{\Theta}_{\textrm{target}}$, replay buffer $\mathcal{D} = \set{}$, and time step $t = 0$\;
\While{$t < \textrm{max steps}$}{
    \Comment{Environment interaction.}
    Get environment state $x_{t} \in \setstate$\;
    
    Run \cref{alg:qnrl:quak} to sample policy action $a_{t} = \arg\max_{a \in \setaction} \pi_{\textrm{policy}}(a \given x_{t}, Z_{\bm{\Theta}_{\textrm{policy}}})$\;
    
    Apply action $a_{t}$ and get reward $r_{t}$ and next state $x_{t+1}$\;
    
    Update local replay buffer $\mathcal{D} = \mathcal{D} \cup \set{\tuple{x_{t}, a_{t}, r_{t}, x_{t+1}}}$\;
    
    \If{$x_{t+1}$ is terminal}{
        Reset environment\;
    }

    \Comment{Model update.}
    \If{$t > \textrm{learning starts}$}{
        \Comment{Policy model update.}
        \If{$t \bmod \textrm{policy update frequency} = 0$}{
            Sample a batch of $B$ training samples from buffer $\set{\tuple{x_{\tau}, a_{\tau}, r_{\tau}, x_{\tau+1}} \sim \mathcal{D}}_{0,\dots,B-1}$\;

            Run \cref{alg:qnrl:quak} to sample target action $a^{*}_{\tau} = \arg\max_{a \in \setaction} \pi_{\textrm{target}}(a \given x_{\tau+1}, Z_{\bm{\Theta}_{\textrm{target}}})$\;

            Prepare target return distribution $\mathscr{P}(\setreturn \given x_{\tau+1},a^{*}_{\tau},\bm{\Theta}_{\textrm{target}})$ via \cref{eq:qnrl:CU_z_x_a}\;

            Update atoms of target distribution with discounted reward $\setreturn_{\textrm{target}} = \set{r_{\tau} + \eta z}_{z \in \setreturn}$\;

            Realign the atoms of the target distribution $\setreturn_{\textrm{target}}$ onto the support of $\setreturn$ according to \cite[Section 4.2]{Bellemare2017DistributionalPerspectiveReinforcement}\;

            Prepare policy return distribution $\mathscr{P}(\setreturn \given x_{\tau},a_{\tau},\bm{\Theta}_{\textrm{policy}})$ via \cref{eq:qnrl:CU_z_x_a}\;
                
            Compute cross-entropy loss $\mathcal{L} = -\sum_{z \in \setreturn} p(z \given x_{\tau+1},a^{*}_{\tau},\bm{\Theta}_{\textrm{target}}) \bfunc{\log}{p(z \given x_{\tau},a_{\tau},\bm{\Theta}_{\textrm{policy}})}$\;
        }
        
        \Comment{Target model update.}
        \If{$t \bmod \textrm{target update frequency} = 0$}{
            Incrementally move $\bm{\Theta}_{\textrm{target}}$ toward $\bm{\Theta}_{\textrm{policy}}$ via soft update rule\;
        }
    }
    $x_{t} = x_{t+1}$\;
    $t = t + 1$\;
}
\end{algorithm2e}

\section{Experimental Results and Analysis}\label{sec:qnrl:exp}

Next, we discuss the results of our experiments. First, we will discuss the environments used to train and evaluate our models, followed by our experiment configuration. We conclude this section with discussions on the results for each of the environment experiments.

\subsection{Environments}\label{sec:qnrl:exp:env}

We use several well-known environments across multiple domains as benchmarks for our proposed quantum-native and baseline models, which include the classic control environments \texttt{CartPole} \cite{Barto1983NeuronlikeAdaptiveElements} and \texttt{Acrobot} \cite{Sutton2018ReinforcementLearningIntroduction}, the grid-world environments \texttt{CliffWalking} \cite{Sutton2018ReinforcementLearningIntroduction} and \texttt{FrozenLake} \cite{Brockman2016OpenAIGym}, and the Atari environments \texttt{Breakout} and \texttt{SpaceInvaders} \cite{Bellemare2013ArcadeLearningEnvironment,Machado2018RevisitingArcadeLearning}. 
In particular, the continuous observation spaces of \texttt{CartPole} and \texttt{Acrobot}, and their contrast with the discrete observation spaces of \texttt{CliffWalking} and \texttt{FrozenLake}, serve as interesting case studies for how return probability distributions originate from both continuous and discrete sources and how the generated models fair in the presence of unseen data; that is, preparing distributions of unknown observations. 
In a similar vain, the image observations from \texttt{Breakout} and \texttt{SpaceInvaders} require an entirely different encoding structure, and, thus, they also serve as case study for how classical and quantum systems generate return probability distributions from this image data.
Further, all environments require sufficient exploration to either reach a goal state or achieve a high score, which is also an interesting study for how the generated return distributions facilitate this search.
In all scenarios we evaluate agents using the \emph{episode reward} metric as defined by each respective environment. 
See \cref{app:qnrl:env} for environment details.

\subsection{Experimental Setup}\label{sec:qnrl:exp:setup}

We compare our \ac{qnrl} against the classical \emph{categorical \ac{dqn}} algorithm with a classical fully-connected \ac{nn} as proposed by \cite{Bellemare2017DistributionalPerspectiveReinforcement} and as implemented by \cite{Huang2022CleanRLHighqualitySinglefile}. We refer to this baseline by its canonical name \textbf{\texttt{C51}}. This is a meaningful baseline because it is the foundation for distributional learning, and thus an excellent comparison for re-framing the distributional generation and policy selection processes into quantum space.
All models were built in JAX using \texttt{flax.nnx} \cite{Heek2024FlaxNeuralNetwork}, optimized using \texttt{optax} \cite{DeepMind2020DeepMindJAXEcosystem}, and use \texttt{pennylane} \cite{Bergholm2022PennyLaneAutomaticDifferentiation} as the quantum circuit backend.
We fix the number of distributional atoms to 32 for all models, resulting in a quantum circuit of $q_{\setreturn} = 5$ qubits. The return value range for \texttt{FrozenLake} is $z \in [-2, 2]$, for \texttt{CartPole}, \texttt{CliffWalking}, and \texttt{Acrobot} we use the range $z \in [-100, 100]$, and for \texttt{Breakout} and \texttt{SpaceInvaders} we use $z \in [-10, 10]$. We use a constant return discount factor of $\eta = 0.99$ for all models.
All models were trained continually for a total of $T = 100,000$ time steps for both the classic control and grid-world environments, and $T = 200,000$ for the Atari environments, with automatic environment reset after each episode, a circular buffer size of $|\mathcal{D}| = 10,000$ transitions, a batch size of $128$ trajectory samples for both the classic control and grid-world environments, and a batch size of $32$ for the Atari environments.
During training, all models employ a decaying $\epsilon$-greedy policy strategy for $50\%$ of the training window to both inject randomness from the environment and encourage exploration, with $1 \geq \epsilon \geq  0.05$ for both the classic control and grid-world environments, and $1 \geq \epsilon \geq 0.01$ for the Atari environments respectively. This $\epsilon$-greedy strategy is removed during evaluation.
The baseline classical models are trained with the \texttt{Adam} optimizer \cite{Kingma2014AdamMethodStochastic} using a learning rate of $lr = 10^{-3}$ for both the classic control and grid-world environments, $lr = 2.5 \times 10^{-4}$ for the Atari environments. The quantum models are trained using the \texttt{AdamW} optimizer \cite{Loshchilov2019DecoupledWeightDecay} with different learning rates for the classical encoding and quantum parameters, specifically $\set{lr_{\textrm{classical}} = 10^{-3}, lr_{\textrm{quantum}} = 10^{-2}}$ for both the classic control and grid-world environments, and $\set{lr_{\textrm{classical}} = 2.5 \times 10^{-4}, lr_{\textrm{quantum}} = 10^{-2}}$ for the Atari environments respectively.
All models are trained using \texttt{cross-entropy} loss.
Every model configuration is uniquely identified by a hash string generated from its hyperparameters. A description of these hyperparameters are provided in \cref{app:qnrl:hyper}.
All model configurations are trained on 4 unique seeds, and subsequently evaluated on a separate set of 10 unique seeds. The performance of all our models is reported as an aggregate over these training and evaluation seeds respectively.
For each framework, we compare model configurations based upon both their mean evaluation episode reward, and total size (i.e., number of parameters). We select the ``best'' configuration based on which results in the \emph{highest reward improvement} with the \emph{lowest size}. This combination of reward and size metrics results in the most fair comparison between models; particularly the purely classical baselines which tend to have slightly higher absolute performance at the cost of a significantly larger model (i.e., an order of magnitude).


\subsection{Results and Analysis}\label{sec:qnrl:exp:res}

The first suite of experiments compare the training performance, evaluation performance, and model sizes of \texttt{QnRL} with baseline \texttt{C51} on \texttt{CartPole} in \cref{sec:qnrl:exp:CartPole}, \texttt{CliffWalking} in \cref{sec:qnrl:exp:CliffWalking}, \texttt{FrozenLake} in \cref{sec:qnrl:exp:FrozenLake}, \texttt{Acrobot} in \cref{sec:qnrl:exp:Acrobot}, \texttt{SpaceInvaders} in \cref{sec:qnrl:exp:SpaceInvaders}, and finally \texttt{Breakout} in \cref{sec:qnrl:exp:Breakout} respectively.

\subsubsection{Results for CartPole}\label{sec:qnrl:exp:CartPole}

\begin{figure}[t!] 
\centering
\subcaptionbox{Train Episode Reward \label{fig:qnrl:exp:CartPole:train-episode-reward}}{\includegraphics[width=0.6\columnwidth]{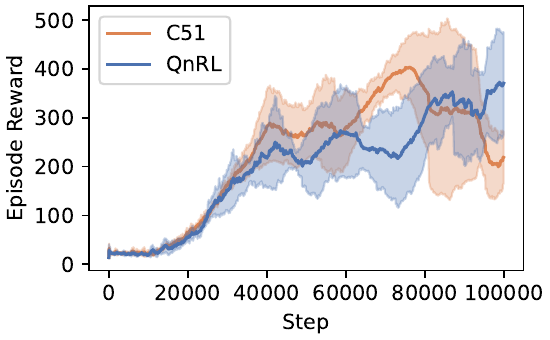}}
\hfil\\
\subcaptionbox{Evaluation Episode Reward \label{fig:qnrl:exp:CartPole:eval-episode-reward}}{\includegraphics[width=0.45\columnwidth]{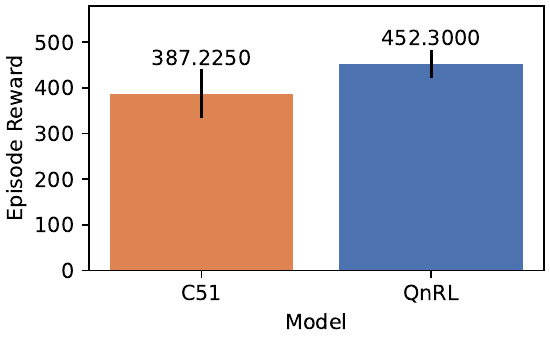}}
\hfil
\subcaptionbox{Model Parameter Count \label{fig:qnrl:exp:CartPole:total-model-parameter-count}}{\includegraphics[width=0.45\columnwidth]{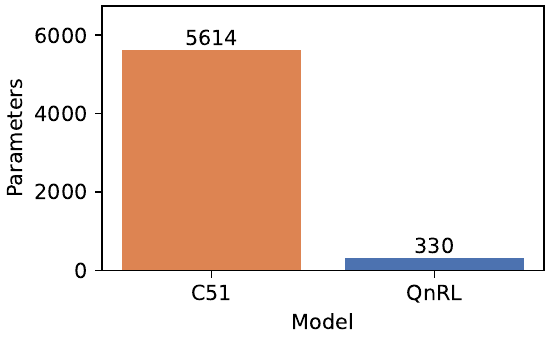}}
\caption[\texttt{CartPole} performance comparison of \acs{qnrl} with baselines.]{Comparison of \texttt{CartPole} metrics (a) train episode reward, (b) evaluation episode reward, and (c) model parameter count for best model configuration of \texttt{C51} (orange) and \texttt{QnRL} (blue). Training performance averaged over 4 seeds, evaluation performance averaged over 10 seeds, with $\pm 1$ std.~dev.\ shown. These figures generally show that \texttt{QnRL} achieves higher and more stable training and evaluation performance with a significantly fewer parameters.\label{fig:qnrl:exp:CartPole}}
\end{figure}


The first set of experiments compare the performance of \texttt{QnRL} (model \texttt{45913de1}) with baseline \texttt{C51} (model \texttt{bdc26ac3}) on the \texttt{CartPole} environment as shown in \cref{fig:qnrl:exp:CartPole}. 
Looking at the training episode reward in \cref{fig:qnrl:exp:CartPole:train-episode-reward} we see that \texttt{QnRL} achieves a $69.5\%$ higher final score (reward of $370.66$) than \texttt{C51} (reward of $218.63$), which drops near the end of training. From this we can infer that the classical model suffers from over-training.
\Cref{fig:qnrl:exp:CartPole:eval-episode-reward} also shows that \texttt{QnRL} achieves a $16.8\%$ higher evaluation score (reward of $452.3$) than \texttt{C51} (reward of $387.23$), and with a lower standard deviation.
This training and evaluation performance is interesting when considering the model sizes as shown in \cref{fig:qnrl:exp:CartPole:total-model-parameter-count}, where we see that \texttt{QnRL} is $94.1\%$ smaller than \texttt{C51}. 
Here, \texttt{QnRL} uses $L=7$ layers, power $n=1$, moment $m=1$, and \texttt{offset} entanglement. This composition is important for the \ac{quak} action selection phase because our model compares the 1st moment (i.e., the mean) of the return distributions. Further, the use of \texttt{offset} entanglement indicates that distribution atoms have a more complex relationship to sequential neighbors (as opposed to constant next-neighbor), which we capture as a function of the layer depth.
From this we can infer that our \texttt{QnRL} performs better because it is able to replicate this relationship between atoms via the offset entanglement, and retain the complexity of this correlation by comparing the distribution means in quantum space.
Hence, we observe that there is a quantum advantage in performance, model size, and by extension computational complexity by learning the return distribution as a quantum state in Hilbert space.

\subsubsection{Results for CliffWalking}\label{sec:qnrl:exp:CliffWalking}

\begin{figure}[t!] 
\centering
\subcaptionbox{Train Episode Reward \label{fig:qnrl:exp:CliffWalking:train-episode-reward}}{\includegraphics[width=0.6\columnwidth]{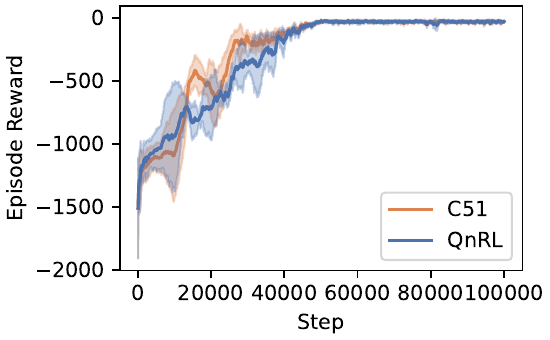}}
\hfil\\
\subcaptionbox{Evaluation Episode Reward \label{fig:qnrl:exp:CliffWalking:eval-episode-reward}}{\includegraphics[width=0.45\columnwidth]{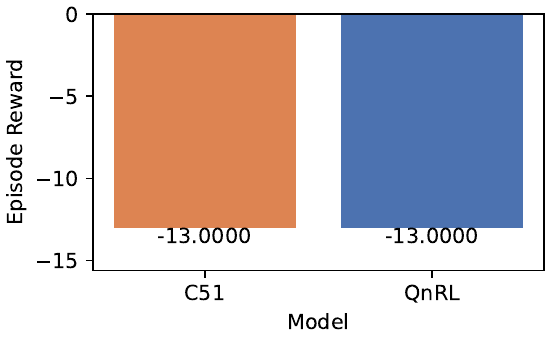}}
\hfil
\subcaptionbox{Model Parameter Count \label{fig:qnrl:exp:CliffWalking:total-model-parameter-count}}{\includegraphics[width=0.45\columnwidth]{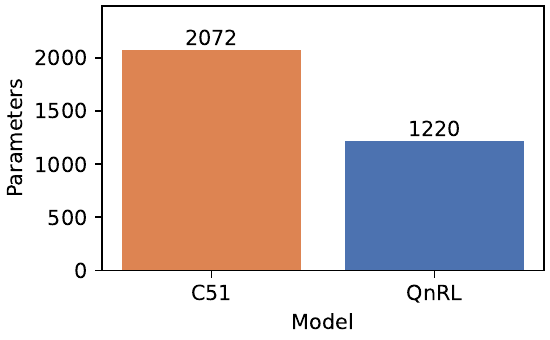}}
\caption[\texttt{CliffWalking} performance comparison of \acs{qnrl} with baselines.]{Comparison of \texttt{CliffWalking} metrics (a) train episode reward, (b) evaluation episode reward, and (c) model parameter count for best model configuration of \texttt{C51} (orange) and \texttt{QnRL} (blue). Training performance averaged over 4 seeds, evaluation performance averaged over 10 seeds, with $\pm 1$ std.~dev.\ shown. These figures generally show that \texttt{QnRL} achieves more stable performance that is on par with baselines using significantly fewer parameters.\label{fig:qnrl:exp:CliffWalking}}
\end{figure}


The second set of experiments compare the performance of \texttt{QnRL} (model \texttt{12c045e6}) with baseline \texttt{C51} (model \texttt{15173423}) on the \texttt{CliffWalking} environment as shown in \cref{fig:qnrl:exp:CliffWalking}.
Here, \texttt{QnRL} uses $L=3$ layers, power $n=1$, moment $m=1$, and \texttt{circular} entanglement. This means that the model compares the 1st moment of the return distributions, and that the return atoms are directly related to their next neighbors, i.e., entanglement is \texttt{circular} with constant neighbor assignments.
The training evaluation performance in \cref{fig:qnrl:exp:CliffWalking:train-episode-reward} shows that \texttt{C51} converges to a final score that is $11.7\%$ higher (episode reward of $-27.55$) than \texttt{QnRL} (reward of $-31.20$), but from the evaluation performance in \cref{fig:qnrl:exp:CliffWalking:eval-episode-reward} we see that both models reach the maximum possible reward of $-13$ for all 10 seeds. From this, we can infer that both models learn return distributions that sufficiently generalize to observations in the environment. 
Here, \texttt{QnRL} performs slightly worse in training because of the over-optimization of the learned distribution. In particular, the model is learning return distributions for some actions that produce very similar means, which causes \ac{quak} to estimate a similar probability for those actions. This, coupled with a small $\epsilon$-greedy element causes the model to take a non-optimal action slightly more frequently than \texttt{C51}.
This performance, however, is particularly interesting when considering the model sizes shown in \cref{fig:qnrl:exp:CliffWalking:total-model-parameter-count}, where \texttt{QnRL} achieves the same evaluation score using $41.1\%$ fewer parameters. Further, the composition of this size is also interesting because we use a one-hot encoding scheme for the observations in the \texttt{CliffWalking} environment, which results in a relatively large 48-dimensional input mapping. 
%
For \texttt{C51} we map $48 \mapsto h \mapsto |\setaction| \times |\setreturn|$, where $h=\set{16, 8}$ are the hidden layer dimensions, and for \texttt{QnRL} we map $48 \mapsto |\setaction| \times q_{\setreturn}$ and then perform an Einstein summation with a small quantum parameter kernel $\bm{\gamma}$ (see \cref{def:qnrl:U_enc}) to multiplex the result into the layer dimension, where $q_{\setreturn}=5$ is the number of return qubits for all quantum models in our experiments. Hence, even when classically encoding observations with many dimensions, we see a quantum advantage for learning a return distribution in Hilbert space, because the number of parameters necessary to map that classical data into quantum space is significantly smaller than its purely classical counterpart.

\subsubsection{Results for FrozenLake}\label{sec:qnrl:exp:FrozenLake}

\begin{figure}[t!] 
\centering
\subcaptionbox{Train Episode Reward \label{fig:qnrl:exp:FrozenLake:train-episode-reward}}{\includegraphics[width=0.6\columnwidth]{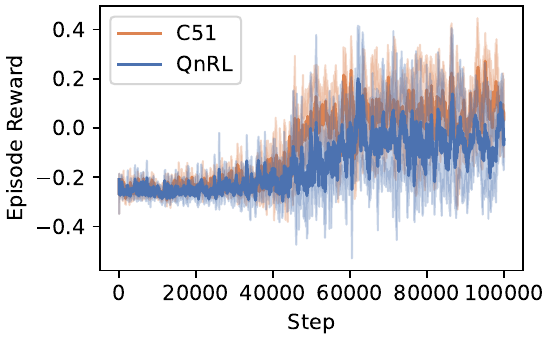}}
\hfil\\
\subcaptionbox{Evaluation Episode Reward \label{fig:qnrl:exp:FrozenLake:eval-episode-reward}}{\includegraphics[width=0.45\columnwidth]{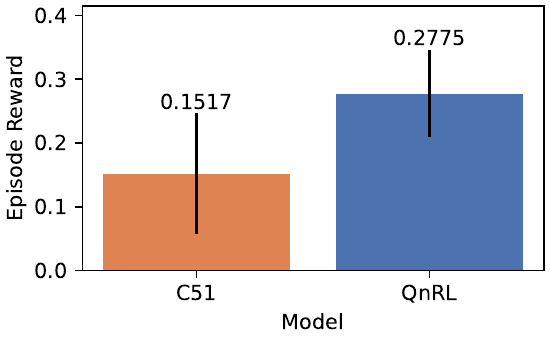}}
\hfil
\subcaptionbox{Model Parameter Count \label{fig:qnrl:exp:FrozenLake:total-model-parameter-count}}{\includegraphics[width=0.45\columnwidth]{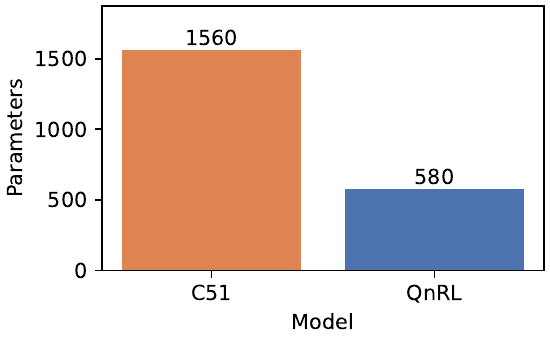}}
\caption[\texttt{FrozenLake} performance comparison of \acs{qnrl} with baselines.]{Comparison of \texttt{FrozenLake} metrics (a) train episode reward, (b) evaluation episode reward, and (c) model parameter count for best model configuration of \texttt{C51} (orange) and \texttt{QnRL} (blue). Training performance averaged over 4 seeds, evaluation performance averaged over 10 seeds, with $\pm 1$ std.~dev.\ shown. These figures generally show that \texttt{C51} exhibits higher training performance, whereas \texttt{QnRL} achieves higher and more stable evaluation performance with significantly fewer parameters.\label{fig:qnrl:exp:FrozenLake}}
\end{figure}


The third set of experiments compare the performance of \texttt{QnRL} (model \texttt{a275ee45}) with baseline \texttt{C51} (model \texttt{1c44782b}) on the \texttt{FrozenLake} environment as shown in \cref{fig:qnrl:exp:FrozenLake}.
This \texttt{QnRL} model uses $L=3$ layers, power $n=1$, moment $m=2$, and \texttt{offset} entanglement, which is the first appearance of higher-order distribution effects. Here, our quantum model compares the 2nd moment (i.e., the variance) of the return distributions, which indicates that the observations and returns in this environment are more deeply related, or more sensitive, to the ``spread'' across the distribution support. Further, the \texttt{offset} entanglement indicates that the behavior of these atoms are directly influenced by several several neighbors.
The training performance in \cref{fig:qnrl:exp:FrozenLake:train-episode-reward} shows that \texttt{C51} achieves a $35.7\%$ higher maximum score (reward of $0.2712$ at step $95137$) than \texttt{QnRL} (reward of $0.1998$ at step $62078$) across 4 seeds, but interestingly \texttt{QnRL} is $28.5\%$ more stable (std.~dev.\ of $0.0949$) than \texttt{C51} (std.~dev.\ of $0.1326$) throughout training.
Looking at the evaluation performance in \cref{fig:qnrl:exp:FrozenLake:eval-episode-reward}, however, we see that \texttt{QnRL} achieves an $82.9\%$ higher average score with less variation (reward of $0.2775$) than \texttt{C51} (reward of $0.1517$) across 10 seeds.
We believe the reasons for why \texttt{QnRL} performs worse here in training are similar to those of its training performance in \texttt{CliffWalking}, which are the moment separation between distributions and the changing stochastic nature of the policy due to varying $\epsilon$-greedy components respectively. In this case, the quantum model is generating distributions that have similar variances for some actions, causing there to be very close probability. We see this prominently in both cases because the quantum models have seemingly plateaued in training performance, which indicates that the generated distributions cannot be further tuned for increased mean and variance separation. During training, the models are optimized using a decaying $\epsilon$-greedy component, which settles at a fixed value for the remainder of the training duration. From the perspective of the models, this $\epsilon$ value injects a level of randomness into the action selection process within the environment, from which non-optimal actions will be sampled with a fixed uniform probability. During evaluation, however, the frequency of randomness is altered by selecting a different fixed $\epsilon$ value as a way of testing model performance under different stochastic conditions. We believe \texttt{C51} often performs better during training because it more rapidly learns the random frequency of the chosen fixed $\epsilon$, and performs worse during evaluation because it becomes over-tuned to a specific random frequency. The converse is generally true of \texttt{QnRL}, where its performance supersedes that of the classical method when the level of policy randomness changes during evaluation. From this we can deduce that quantum distributional systems are better able to compensate for changes in policy randomness under varying environment conditions.
The model size in \cref{fig:qnrl:exp:FrozenLake:total-model-parameter-count} is analogous to previous experiments, \cref{sec:qnrl:exp:CliffWalking,sec:qnrl:exp:CartPole}, whereby \texttt{QnRL} achieves this performance with $62.8\%$ fewer parameters than \texttt{C51}. Interestingly, this size comparison is analogous to the \texttt{CliffWalking} experiment in \cref{sec:qnrl:exp:CliffWalking}, in that here we also use a one-hot encoding scheme for the environment map grid. This too shows a quantum advantage for learning a return distribution in Hilbert space from high-dimensional classical input observations.

\subsubsection{Results for Acrobot}\label{sec:qnrl:exp:Acrobot}

\begin{figure}[t!] 
\centering
\subcaptionbox{Train Episode Reward \label{fig:qnrl:exp:Acrobot:train-episode-reward}}{\includegraphics[width=0.6\columnwidth]{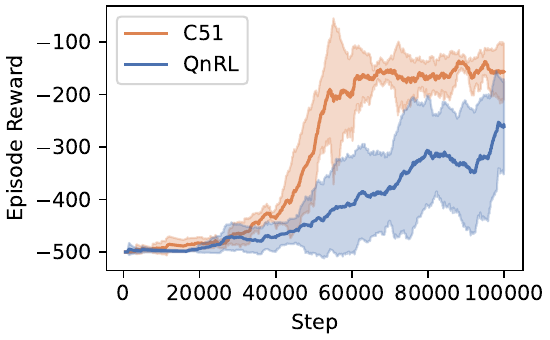}}
\hfil\\
\subcaptionbox{Evaluation Episode Reward \label{fig:qnrl:exp:Acrobot:eval-episode-reward}}{\includegraphics[width=0.45\columnwidth]{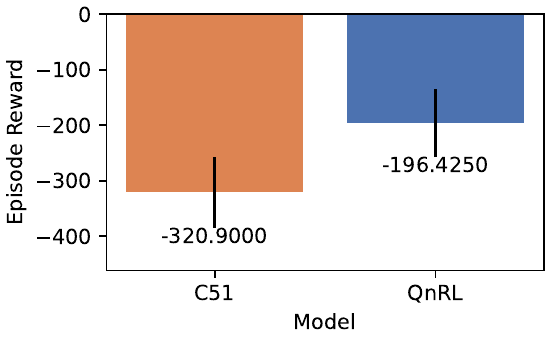}}
\hfil
\subcaptionbox{Model Parameter Count \label{fig:qnrl:exp:Acrobot:total-model-parameter-count}}{\includegraphics[width=0.45\columnwidth]{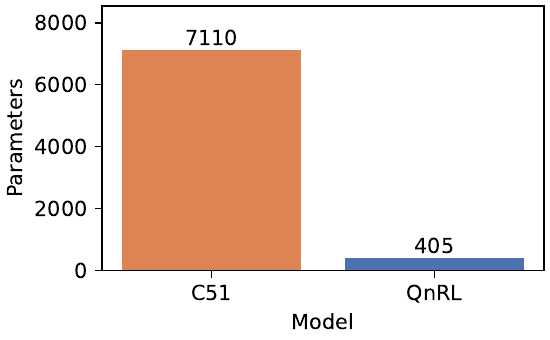}}
\caption[\texttt{Acrobot} performance comparison of \acs{qnrl} with baselines.]{Comparison of \texttt{Acrobot} metrics (a) train episode reward, (b) evaluation episode reward, and (c) model parameter count for best model configuration of \texttt{C51} (orange) and \texttt{QnRL} (blue). Training performance averaged over 4 seeds, evaluation performance averaged over 10 seeds, with $\pm 1$ std.~dev.\ shown. These figures generally show that \texttt{C51} achieves strictly better training performance, but \texttt{QnRL} performs better in evaluation with significantly fewer parameters.\label{fig:qnrl:exp:Acrobot}}
\end{figure}


The fourth set of experiments compare the performance of \texttt{QnRL} (model \texttt{1bacc532}) with baseline \texttt{C51} (model \texttt{c3881283}) on the \texttt{Acrobot} environment as shown in \cref{fig:qnrl:exp:Acrobot}.
The \texttt{QnRL} model uses $L=5$ layers, power $n=1$, moment $m=2$, and \texttt{circular} entanglement. This indicates that the relationship between the observation space and the return distributions exhibits 2nd-order effects, and the behavior of the return distribution atoms are directly affected by their next neighbor.
Looking at the training episode reward in \cref{fig:qnrl:exp:Acrobot:train-episode-reward} we see that \texttt{C51} learns an $46.13\%$ higher and more stable maximum score (reward of $-136.04$ at step $88156$) compared to \texttt{QnRL} (reward of $-252.54$ at step $98527$). We also see that the mean episode reward for \texttt{C51} plateaus around step $60,000$, wheras \texttt{QnRL} steadily increases throughout training. This plateau suggests that the expressiveness of the classical distribution has peaked with the given parameters; that is, the shape of the distribution cannot be tuned further.
This training performance would initially suggest that \texttt{C51} learns a better distribution than \texttt{QnRL}. The evaluation performance shown in \cref{fig:qnrl:exp:Acrobot:eval-episode-reward} paints a different picture however.
Here we see that \texttt{QnRL} achieves a $38.8\%$ higher average evaluation score (reward of $-196.43$) than \texttt{C51} (reward of $-320.90$) and with similar stability. 
This difference between training and evaluation performance indicates that the quantum model favors learning a distribution that is more generalized to the wider observation space than the classical model. This is an interesting behavior considering the \texttt{Acrobot} environment because observations are continuous real numbers over a range of arm joint angles, which is a more complex task than the previous \texttt{CartPole} experiment discussed in \cref{sec:qnrl:exp:CartPole}. Given this continuous nature, it is likely that a model will not experience all possible joint angle combinations during training. Hence, a model must learn to generate return distributions that sufficiently represent the relationship between regions of the observation space with estimated rewards over the action space. We see this tradeoff in practice by comparing \cref{fig:qnrl:exp:Acrobot:train-episode-reward,fig:qnrl:exp:Acrobot:eval-episode-reward}, where clearly \texttt{C51} generates return distributions that inadequately represent the relationship for unseen joint angle combinations.
This comparison is also interesting when we consider the model sizes shown in \cref{fig:qnrl:exp:Acrobot:total-model-parameter-count}. Here we see that \texttt{QnRL} achieves this performance with $94.3\%$ fewer parameters, which is similar to all previous experiments discussed in \cref{sec:qnrl:exp:CartPole,sec:qnrl:exp:CliffWalking,sec:qnrl:exp:FrozenLake}.
Hence, as we have shown in \cref{fig:qnrl:exp:Acrobot} our proposed \texttt{QnRL} learns to generate quantum return distributions that are more expressive in the presence of unseen observations, and it does so with a significantly smaller model.

\subsubsection{Results for SpaceInvaders}\label{sec:qnrl:exp:SpaceInvaders}

\begin{figure}[t!] 
\centering
\subcaptionbox{Train Episode Reward \label{fig:qnrl:exp:SpaceInvaders:train-episode-reward}}{\includegraphics[width=0.6\columnwidth]{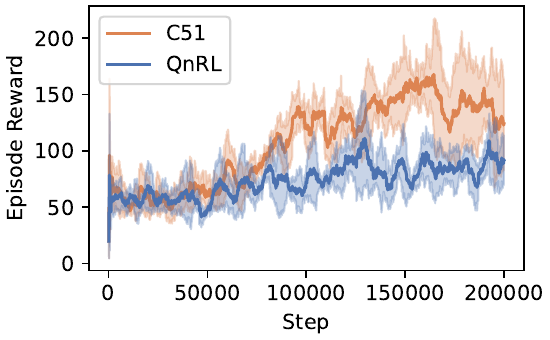}}
\hfil\\
\subcaptionbox{Evaluation Episode Reward \label{fig:qnrl:exp:SpaceInvaders:eval-episode-reward}}{\includegraphics[width=0.45\columnwidth]{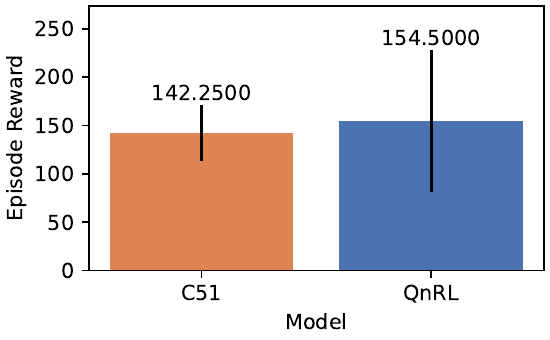}}
\hfil
\subcaptionbox{Model Parameter Count \label{fig:qnrl:exp:SpaceInvaders:total-model-parameter-count}}{\includegraphics[width=0.45\columnwidth]{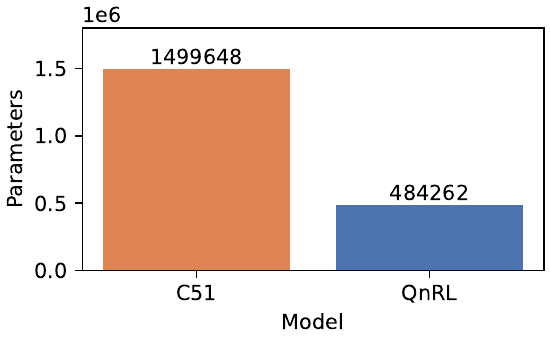}}
\caption[\texttt{SpaceInvaders} performance comparison of \acs{qnrl} with baselines.]{Comparison of \texttt{SpaceInvaders} metrics (a) train episode reward, (b) evaluation episode reward, and (c) model parameter count for best model configuration of \texttt{C51} (orange) and \texttt{QnRL} (blue). Training performance averaged over 4 seeds, evaluation performance averaged over 10 seeds, with $\pm 1$ std.~dev.\ shown. These figures generally show that \texttt{C51} achieves strictly better training performance, but \texttt{QnRL} performs better in evaluation with significantly fewer parameters.\label{fig:qnrl:exp:SpaceInvaders}}
\end{figure}

The fifth set of experiments compare the performance of \texttt{QnRL} (model \texttt{8066810a}) with baseline \texttt{C51} (model \texttt{f758df04}) on the \texttt{SpaceInvaders} environment as shown in \cref{fig:qnrl:exp:SpaceInvaders}. The \texttt{QnRL} model uses $L=7$ layers, power $n=1$, moment $m=1$, and \texttt{offset} entanglement. Looking at the training episode reward in \cref{fig:qnrl:exp:SpaceInvaders:train-episode-reward} we see that \texttt{C51} learns a $35.09\%$ higher final score (reward of $124.16$) than \texttt{QnRL} (reward of $91.91$), but \texttt{QnRL} is $61.51\%$ more stable (std.~dev.~of $14.05$) than \texttt{C51} (std.~dev.~of $36.50$). Looking at the evaluation performance in \cref{fig:qnrl:exp:SpaceInvaders:eval-episode-reward} we see that \texttt{QnRL} scores roughly 12 points higher on average (reward of $154.5$) than \texttt{C51} (reward of $142.25$), and looking at \cref{fig:qnrl:exp:SpaceInvaders:total-model-parameter-count} \texttt{QnRL} achieves this with a $67.71\%$ smaller model. The composition of this size is interesting for reasons similar to those explained for \texttt{CliffWalking} in \cref{sec:qnrl:exp:CliffWalking}. Here, both models employ the same input convolution encoding network with feature dimensions $f=\set{64, 128, 128}$ because the input data for the Atari environment are image frames. For \texttt{C51} we map the output of the final convolution layer to the dimension of the return distributions as given by $f_{\abs{f}} \times w \times h \mapsto \abs{\mathcal{A}} \times \abs{\mathcal{Z}}$, where $w$ and $h$ are the width and height of the image after the final convolution layer, and for \texttt{QnRL} we map $f_{\abs{f}} \times w \times h \mapsto \abs{\mathcal{A}} \times q_{\mathcal{Z}}$ and then multiplex into the quantum circuit layer dimension via an Einstein summation with a small quantum parameter kernel $\bm{\gamma}$ (see \cref{def:qnrl:U_enc}). From this we see a clear quantum advantage for our proposed distributional architecture because $q_{\mathcal{Z}} = \log_{2}{\abs{\mathcal{Z}}} \ll \abs{\mathcal{Z}}$, resulting in a total model size reduction by over half when using the same input encoding network.

\subsubsection{Results for Breakout}\label{sec:qnrl:exp:Breakout}

\begin{figure}[t!] 
\centering
\subcaptionbox{Train Episode Reward \label{fig:qnrl:exp:Breakout:train-episode-reward}}{\includegraphics[width=0.6\columnwidth]{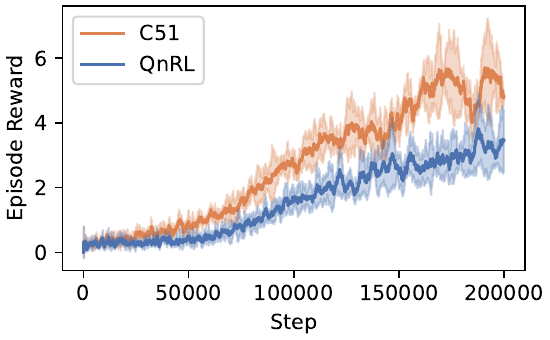}}
\hfil\\
\subcaptionbox{Evaluation Episode Reward \label{fig:qnrl:exp:Breakout:eval-episode-reward}}{\includegraphics[width=0.45\columnwidth]{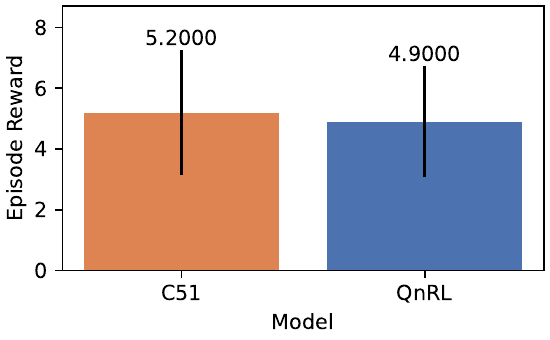}}
\hfil
\subcaptionbox{Model Parameter Count \label{fig:qnrl:exp:Breakout:total-model-parameter-count}}{\includegraphics[width=0.45\columnwidth]{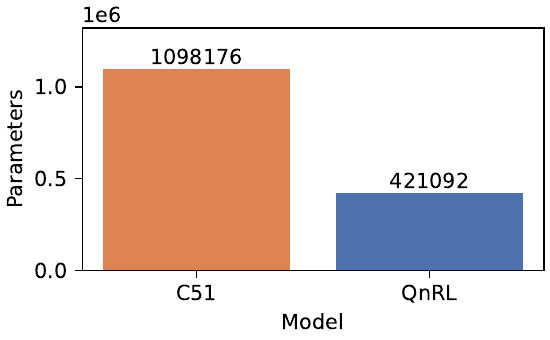}}
\caption[\texttt{Breakout} performance comparison of \acs{qnrl} with baselines.]{Comparison of \texttt{Breakout} metrics (a) train episode reward, (b) evaluation episode reward, and (c) model parameter count for best model configuration of \texttt{C51} (orange) and \texttt{QnRL} (blue). Training performance averaged over 4 seeds, evaluation performance averaged over 10 seeds, with $\pm 1$ std.~dev.\ shown. These figures generally show that \texttt{C51} achieves strictly better training performance, but both models achieve very similar performance in evaluation with \texttt{QnRL} doing so with significantly fewer parameters.\label{fig:qnrl:exp:Breakout}}
\end{figure}

The sixth and final set of experiments compare the performance of \texttt{QnRL} (model \texttt{4d2222f7}) with baseline \texttt{C51} (model \texttt{9c1e6421}) on the \texttt{Breakout} environment as shown in \cref{fig:qnrl:exp:Breakout}. The \texttt{QnRL} model uses $L=5$ layers, power $n=1$, moment $m=1$, and \texttt{offset} entanglement. The training episode reward in \cref{fig:qnrl:exp:Breakout:train-episode-reward} shows that \texttt{C51} achieves a $39.02\%$ higher final score (reward of $4.81$) than \texttt{QnRL} (reward of $3.46$), but \texttt{QnRL} is $38.61\%$ more stable (std.~dev.~of $0.9822$) than \texttt{C51} (std.~dev.~of $1.60$). From the evaluation in \cref{fig:qnrl:exp:Breakout:eval-episode-reward}, we observe that both models perform very similarly, with \texttt{C51} achieving a $6.12\%$ higher average score (reward of $5.2$) than \texttt{QnRL} (reward of $4.9$). In fact, considering that in \texttt{Breakout} only integer-valued scores are possible, when rounding to the nearest integer both models actually achieve the same episode score of $\approx 5$. The model sizes in \cref{fig:qnrl:exp:Breakout:total-model-parameter-count} show that \texttt{QnRL} requires $61.66\%$ fewer parameters, which is very similar to the performance of \texttt{SpaceInvaders} discussed in \cref{sec:qnrl:exp:SpaceInvaders} as both configurations use the same convolution encoding network dimensions. Hence, we see that our quantum distributional model can achieve very similar evaluation performance at the benefit of significantly fewer parameters.

\subsubsection{Performance Summary}\label{sec:qnrl:exp:summary}

\begin{table*}[t!]
\caption[Performance and model size summary.]{Summary of model training and evaluation performance, and comparison of model size in number of parameters across environments. Higher performance indicated by $\uparrow$, and equivalent performance indicated by $\approx$.\label{tab:qnrl:exp:summary}}
\centering
\begin{tabular}{lcccc}
    \toprule
    & \multicolumn{2}{c}{Model Performance} & \multicolumn{2}{c}{Model Size} \\
    \cmidrule(lr){2-3}
    \cmidrule(lr){4-5}
    Environment & QnRL & C51 & QnRL & C51 \\
    \midrule
    \texttt{CartPole} & $\uparrow$ train, and $\uparrow$ eval & - & 330 & 5,614 \\
    \texttt{CliffWalking} & $\approx$ eval & $\uparrow$ train, and $\approx$ eval & 1,220 & 2,072 \\
    \texttt{FrozenLake} & $\uparrow$ eval & $\uparrow$ train & 580 & 1,560 \\
    \texttt{Acrobot} & $\uparrow$ eval & $\uparrow$ train & 405 & 7,110 \\
    \texttt{SpaceInvaders} & $\uparrow$ eval & $\uparrow$ train & 484,262 & 1,499,648 \\
    \texttt{Breakout} & $\approx$ eval & $\uparrow$ train, and $\approx$ eval & 421,092 & 1,098,176 \\
    \bottomrule
\end{tabular}
\end{table*}

A detailed summary of the training and evaluation performance in addition to the sizes of all models in our experiments is provided in \cref{tab:qnrl:exp:summary}. In terms of model size, we observe that \texttt{QnRL} requires significantly fewer parameters across all environments. In terms of performance, we observe a clear pattern of \texttt{QnRL} generally performing worse than \texttt{C51} during training (with the exception of \texttt{CartPole}), and performing better at evaluation time (with the exceptions of both \texttt{CliffWalking} and \texttt{Breakout} being effectively equal). This behavior gives insight into how the learned decision strategies adjust under varying levels of randomness in the environment. Recall from \cref{sec:qnrl:exp:setup} that during training, the algorithms use a decaying $\epsilon$-greedy strategy for $50\%$ of the training window with $1 \geq \epsilon \geq  0.05$ both the classic control and grid-world environments, and $1 \geq \epsilon \geq 0.01$ for the Atari environments. The final $\epsilon$ value is maintained throughout the remainder of training as a way of injecting randomness from the environment into the model. This means that the models learn a decision strategy that incorporates an experienced random rate of $0.05$ and $0.01$ respectively. In other words, the models learn that their chosen action will be randomly shuffled approximately $5\%$ and $1\%$ of the time. We observe from the experimental results that \texttt{C51} generally learns the (effectively constant) rate of the random process better than \texttt{QnRL}. During evaluation, this $\epsilon$ random injection is altered by selecting a different rate value (in our experiments we chose $\epsilon=0$) that remains constant for the entire evaluation episode. Importantly, however, the models are unaware that this change has occurred, and thus the change in $\epsilon$ is viewed as a new random process from their perspective. We observe that under this new random process \texttt{QnRL} generally far outperforms \texttt{C51}, which indicates that \emph{\texttt{QnRL} has learned a decision strategy that better adjusts to varying levels of randomness in the environment}.

\subsection{Learned Distributions}\label{sec:qnrl:exp:dist}

We next compare the return distributions generated from each model during evaluation within the environments discussed in \cref{sec:qnrl:exp:res}. We analyze \texttt{C51} by intercepting the return distributions generated by the classical \ac{nn} for each action, and then compare this with the Q-value calculated for each distribution. We analyze \texttt{QnRL} in quantum simulation by first intercepting the density matrix of the return system $\rho \in \hilbert_{\setreturn}$, followed by generating the return distributions via sampling $\hilbert_{\setreturn}$, and then finally generating the action distribution via our proposed \ac{quak} algorithm and sampling $\hilbert_{\setaction}$. The inspection of the density matrix is key to this analysis, as it gives insight for how the quantum distribution is composed in complex Hilbert space, because recall from \cref{sec:qnrl:method} that the quantum return distribution is transparent to the agent in the action selection process via our purely quantum \ac{quak} algorithm. We visualize the density matrix as 3 separate components, which are: \begin{enumerate*}
    \item the magnitude $\abs{\rho}$,
    \item the real component $\textrm{Re}(\rho)$, and
    \item the imaginary component $\textrm{Im}(\rho)$
\end{enumerate*}. This analysis process allows us to inspect the inner workings of the quantum circuit and is one advantage of performing these experiments via quantum simulation.

\subsubsection{Distributions for CartPole}


\begin{figure}[t!] 
\centering
\subcaptionbox{\texttt{C51~(bdc26ac3)} \label{fig:qnrl:exp:CartPole:distribution:c51}}{\fcolorbox{gray!50}{white}{\includegraphics[width=0.95\linewidth]{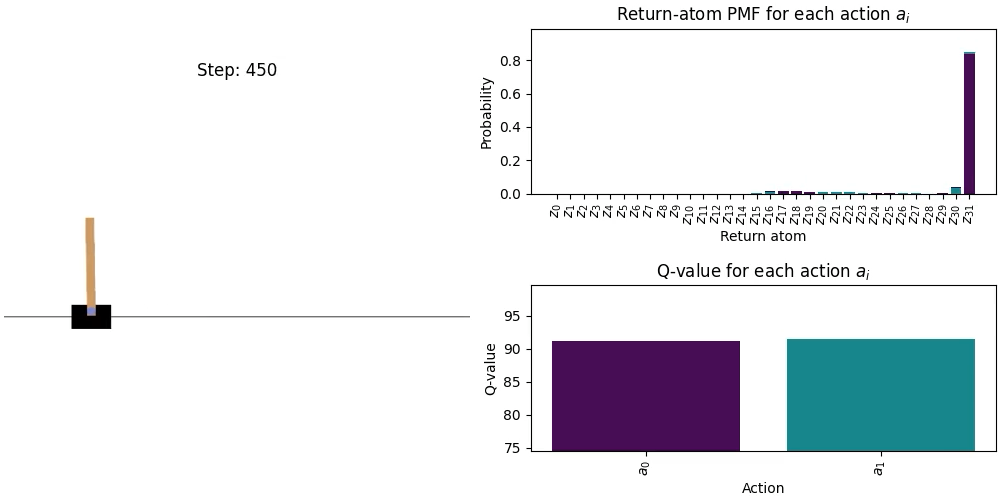}}}
\hfil\\
\subcaptionbox{\texttt{QnRL~(45913de1)} \label{fig:qnrl:exp:CartPole:distribution:qnrl}}{\fcolorbox{gray!50}{white}{\includegraphics[width=0.95\linewidth]{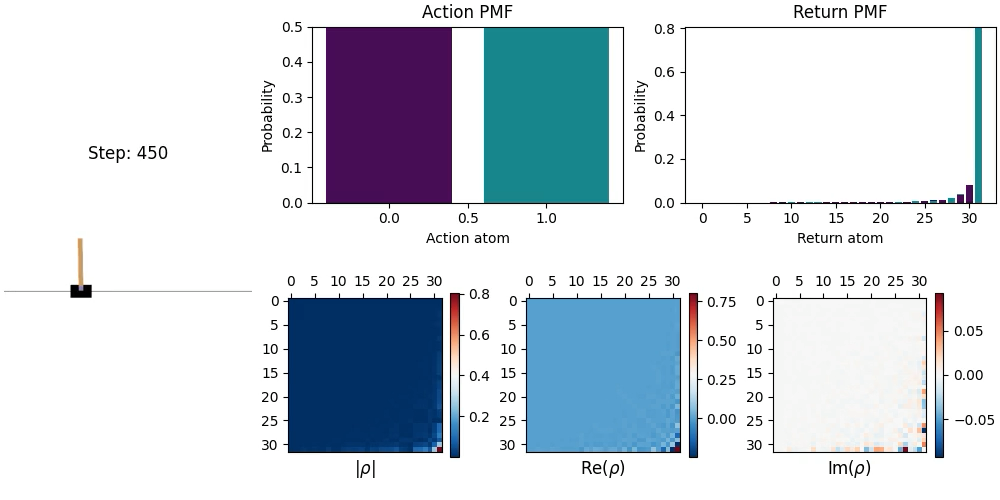}}}
\caption[Learned distributions for \acs{qnrl} and baselines on CartPole.]{Comparison of environment interaction within \texttt{CartPole} for best models of \texttt{C51} (a) and \texttt{QnRL} (b). These figures generally show that that \texttt{QnRL} produces return distributions that are either on par or superior to classical baselines.\label{fig:qnrl:exp:distribution:CartPole}}
\end{figure}

Looking at the \texttt{CartPole} environment in \cref{fig:qnrl:exp:distribution:CartPole}, we focus on step index 450, which is near the end of both episodes (both models reach the maximum of 500 steps, which is optimal). We see that both models learn distributions that are left skewed, resembling impulse functions at the last return atom $z_{31}$. Looking closer at \texttt{C51} we do see a tail starting to form between $[z_{15}, z_{20}]$, and the model estimates a similar Q-value for both actions. Looking at \texttt{QnRL} we do not see such a tail, instead resembling more of a steep exponential about the impulse, and the model estimates both actions have similar first moments (i.e., means) using $n=1$ and $m=1$, with probabilities near $50\%$. We see this shaping in the density matrix visualizations too, with higher amplitudes in the lower-right quadrant, but interestingly here we see that there are both real and complex components to the quantum distribution. This complex composition is a unique feature of \ac{qnrl} that we now see in practice. The quantum model has extra degrees of freedom to generate a distribution across the real and imaginary components, which allows it to learn correlations with the observations that may not be apparent in a purely real classical model.

\subsubsection{Distributions for CliffWalking}


\begin{figure}[t!] 
\centering
\subcaptionbox{\texttt{C51~(15173423)} \label{fig:qnrl:exp:CliffWalking:distribution:c51}}{\fcolorbox{gray!50}{white}{\includegraphics[width=0.95\linewidth]{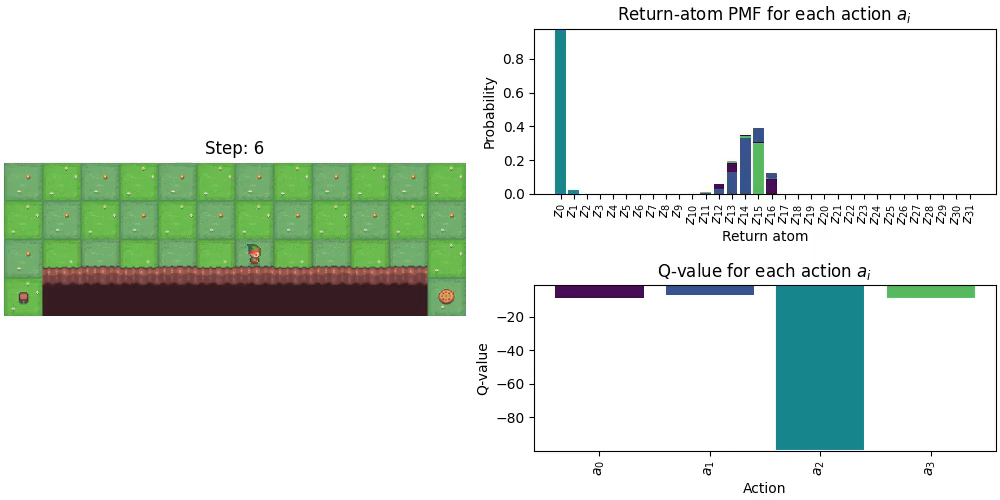}}}
\hfil
\subcaptionbox{\texttt{QnRL~(12c045e6)} \label{fig:qnrl:exp:CliffWalking:distribution:qnrl}}{\fcolorbox{gray!50}{white}{\includegraphics[width=0.95\linewidth]{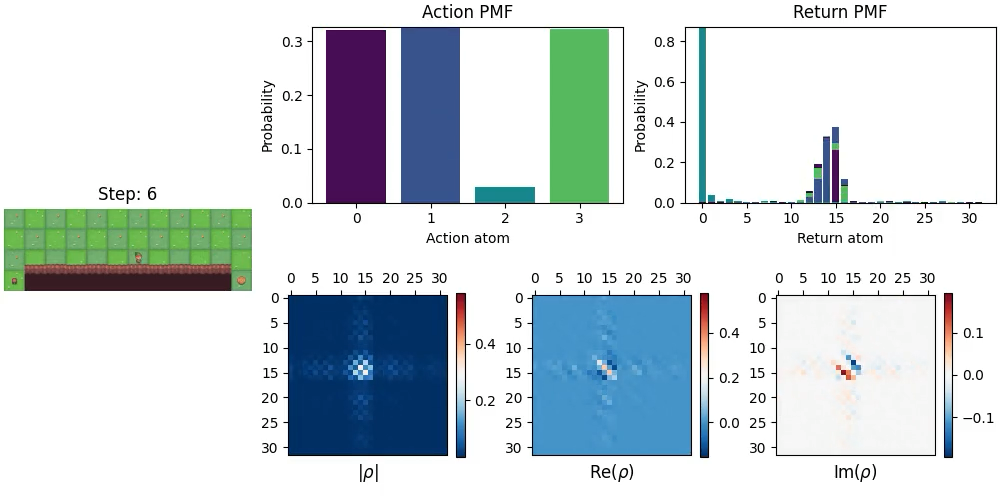}}}
\caption[Learned distributions for \acs{qnrl} and baselines on CliffWalking.]{Comparison of environment interaction within \texttt{CliffWalking} for best models of \texttt{C51} (a) and \texttt{QnRL} (b). These figures generally show that that \texttt{QnRL} produces return distributions that are either on par or superior to classical baselines.\label{fig:qnrl:exp:distribution:CliffWalking}}
\end{figure}

Looking at the \texttt{CliffWalking} environment in \cref{fig:qnrl:exp:distribution:CliffWalking}, we focus on step index 6 which is in the middle of both episodes where the agent is at risk of falling off the cliff (both models reach the goal in 13 steps, which is optimal for this environment). We see that both models learn a Gaussian distribution between $[z_{11}, z_{16}]$ for actions $\set{a_{0} = \texttt{up},~ a_{1} = \texttt{right},~ a_{3} = \texttt{left}}$, and a right skewed impulse at $z_{0}$ for action $a_{2} = \texttt{down}$. This shaping is an interesting case study because here we see that both the classical and quantum models are capable of learning independently-shaped distributions for each action. We also see that both models correctly identify $a_{2}$ as the action with the lowest expected return, as it would result in falling off the cliff. The estimated Q-value and action probabilities place $a_{1}$ with the highest expected return, and thus both models select it as the best action. This \texttt{QnRL} configuration estimates these action probabilities using $n=1$ and $m=1$, similar to the \texttt{CartPole} environment, which shows that comparing the 1st moment of the return distributions is sufficient to learn an optimal policy in this setting. 
The density matrix for \texttt{QnRL} is also interesting here because we clearly see the gaussian shaping as hot-spots across the real and imaginary components, which implies that there is a phase element, i.e., an extra dimensional factor, present in the relationship between observations and the distribution of returns. 
Notice also that we do not see the large impulse for the worst action $a_{2}$ present in the density matrix, which in indicative of it having a significantly lower expected return compared to the other actions. Put simply, the probability for action $a_{2}$ is so low that its return distribution does not influence the combined Hilbert space $\hilbert_{\setaction} \otimes \hilbert_{\setreturn}$ much at all. This is another benefit of our approach, because distributions with very low relative moments present very little amplitude to perturb the system, whereas a classical model would generate all distributions from the \ac{nn} with equal emphasis.

\subsubsection{Distributions for FrozenLake}


\begin{figure}[t!] 
\centering
\subcaptionbox{\texttt{C51~(1c44782b)} \label{fig:qnrl:exp:FrozenLake:distribution:c51}}{\fcolorbox{gray!50}{white}{\includegraphics[width=0.95\linewidth]{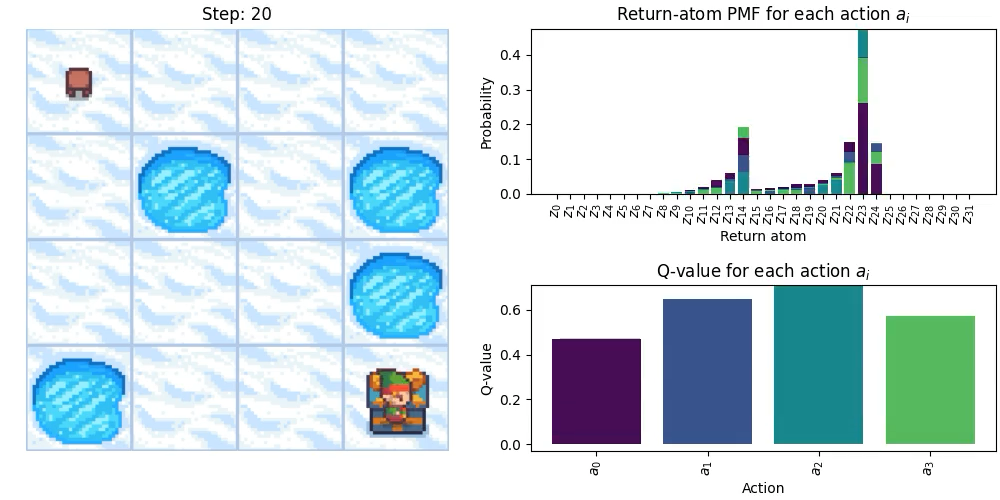}}}
\hfil
\subcaptionbox{\texttt{QnRL~(a275ee45)} \label{fig:qnrl:exp:FrozenLake:distribution:qnrl}}{\fcolorbox{gray!50}{white}{\includegraphics[width=0.95\linewidth]{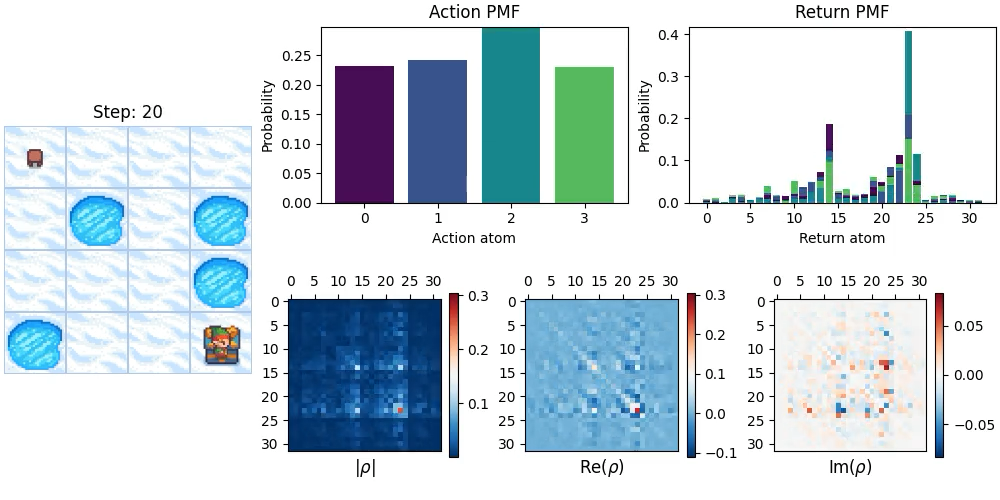}}}
\caption[Learned distributions for \acs{qnrl} and baselines on FrozenLake.]{Comparison of environment interaction within \texttt{FrozenLake} for best models of \texttt{C51} (a) and \texttt{QnRL} (b). These figures generally show that that \texttt{QnRL} produces return distributions that are either on par or superior to classical baselines.\label{fig:qnrl:exp:distribution:FrozenLake}}
\end{figure}

In \cref{fig:qnrl:exp:distribution:FrozenLake}, we show the \texttt{FrozenLake} environment. Here, we focus on the last step prior to the agent reaching the goal, which is step index 20 for both models (both models reach the goal in $21$ steps). This allows us to inspect agent behavior when adjacent to a goal scenario. We observe that both models also learn similarly-shaped return distributions that resemble twinned exponentials with impulses at $z_{14}$ and $z_{23}$ respectively. 
In particular, we see that \texttt{QnRL} generates distributions that have amplitude above zero for most its support. This means the quantum model learns that these atoms are still relevant in the action selection process, which the classical model eliminates entirely.
Both models also predict the same action ranking $a_{2} > a_{1} > a_{3} > a_{0}$ using these distributions, where $\set{a_{0} = \texttt{left},~ a_{1} = \texttt{down},~ a_{2} = \texttt{right},~ a_{3} = \texttt{up}}$, and correctly select $a_{2}$ as the optimal action. Interestingly, this \texttt{QnRL} configuration compares the 2nd moment (i.e., variance) of the return distribution, using $n=1$ and $m=2$. This indicates that second-order effects of the return distribution, i.e., the ``spread'' across its support, has the highest relevance in the relationship between the action and the observation spaces in this environment. The density matrix shows this distribution shape as well with a symmetric pattern, and we see a clear spike in the lower-right quadrant of the magnitude plot. We also see that there is an imaginary component that has a greater presence than the other experiments. This means that the observations and returns for this environment have a greater relationship in the phase component of the return distributions, which is only captured by our modeling in quantum Hilbert space. We believe this is also why the model performs better in evaluation, because these phase components give the model another dimension on which to evaluate the estimated return for each action.

\subsubsection{Distributions for Acrobot}


\begin{figure}[t!] 
\centering
\subcaptionbox{\texttt{C51~(c3881283)} \label{fig:qnrl:exp:Acrobot:distribution:c51}}{\fcolorbox{gray!50}{white}{\includegraphics[width=0.95\linewidth]{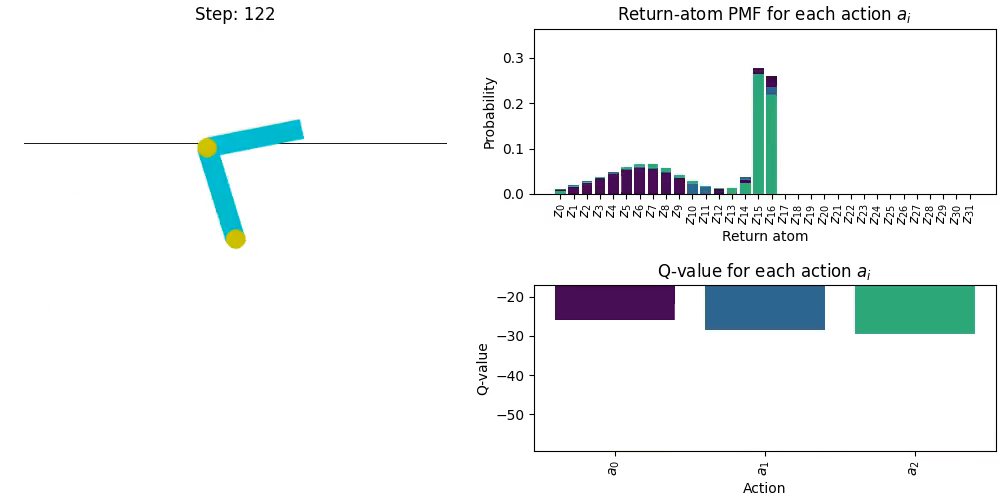}}}
\hfil
\subcaptionbox{\texttt{QnRL~(1bacc532)} \label{fig:qnrl:exp:Acrobot:distribution:qnrl}}{\fcolorbox{gray!50}{white}{\includegraphics[width=0.95\linewidth]{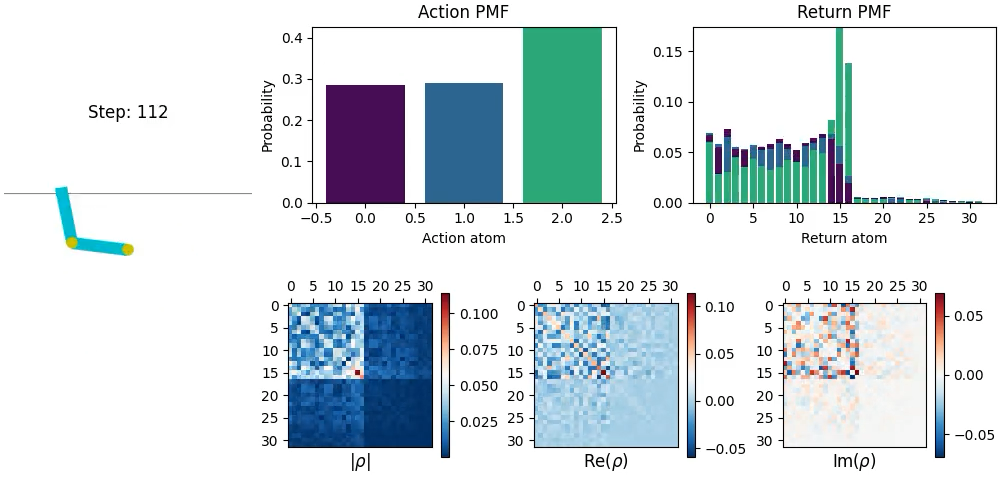}}}
\caption[Learned distributions for \acs{qnrl} and baselines on Acrobot.]{Comparison of environment interaction within \texttt{Acrobot} for best models of \texttt{C51} (a) and \texttt{QnRL} (b). These figures generally show that that \texttt{QnRL} produces return distributions that are either on par or superior to classical baselines.\label{fig:qnrl:exp:distribution:Acrobot}}
\end{figure}

In \cref{fig:qnrl:exp:distribution:Acrobot}, we show the generated distributions resulting from the \texttt{Acrobot} environment. Here, we also focus on the last step prior to the agent reaching the goal, which is step index 112 for \texttt{QnRL}, and step index 122 for \texttt{C51} respectively. While the observations in both cases are different due to their continuous range, we see that the orientation of the arms are similar in both \cref{fig:qnrl:exp:Acrobot:distribution:c51,fig:qnrl:exp:Acrobot:distribution:qnrl}. Hence, we compare the distributions of both models with this consideration in mind. Here we see that the return distribution shapes for both models are different. \texttt{C51} exhibits both a low-amplitude Gaussian within the range $[z_{0}, z_{13}]$, and two high-amplitude impulses at $z_{15}$ and $z_{16}$. \texttt{QnRL} exhibits similar impulses at $z_{15}$ and $z_{16}$, but has a relatively uniform distribution between $[z_{0}, z_{13}]$, and a very low-amplitude tail between $[z_{17}, z_{31}]$. This is interesting because the effective support of the distributions generated by both models are very similar, but the behavior of the lower atoms is unique. We observe that \texttt{C51} estimates that $a_{0} = \textrm{apply $-1$ torque}$ is the action with the highest return, but its Q-value is very similar to the other actions. In contrast, \texttt{QnRL} estimates $a_{2} = \textrm{apply $+1$ torque}$ as the best action with high probability. Interestingly, this \texttt{QnRL} configuration also compares the 2nd moment of the return distribution, using $n=1$ and $m=2$, similar to the model for \texttt{FrozenLake}. we can therefore infer that the relationship between the observation and return spaces exhibit 2nd-order effects that are not captured by the classical algorithm. The density matrix plots paint a similar picture of this complex relationship, as the quantum distribution exhibits clear amplitude and phase components over a wider range of the state map. Further, we see that the density amplitudes are clustered in the top-left quadrant, with lower-amplitude reflections in the top-right, bottom-left, and bottom-right quadrants respectively. This also shows that the observations and returns for this environment have a complex relationship, i.e., in both magnitude and phase components of the return distribution, which is only captured by learning the return distributions natively in a quantum system as in our proposed \ac{qnrl}.

\subsubsection{Distributions for SpaceInvaders}


\begin{figure}[t!] 
\centering
\subcaptionbox{\texttt{C51~(f758df04)} \label{fig:qnrl:exp:SpaceInvaders:distribution:c51}}{\fcolorbox{gray!50}{white}{\includegraphics[width=0.95\linewidth]{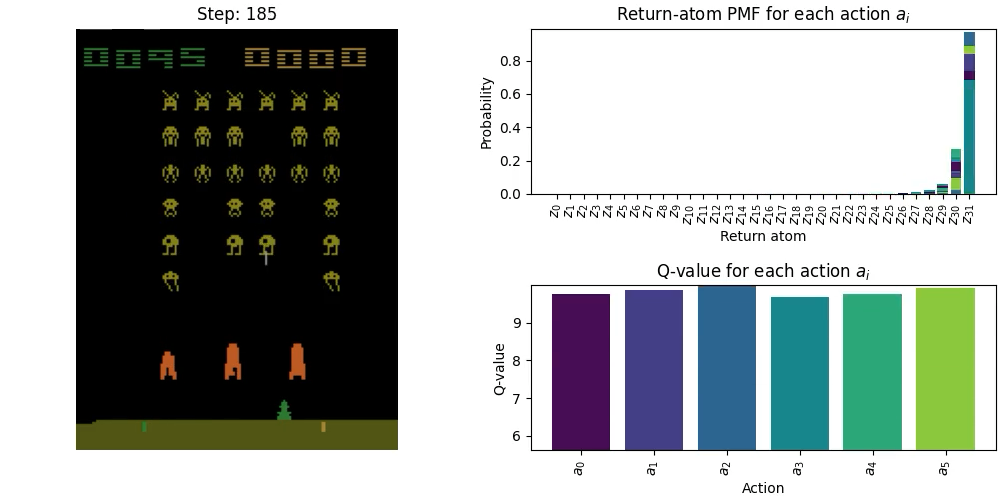}}}
\hfil
\subcaptionbox{\texttt{QnRL~(8066810a)} \label{fig:qnrl:exp:SpaceInvaders:distribution:qnrl}}{\fcolorbox{gray!50}{white}{\includegraphics[width=0.95\linewidth]{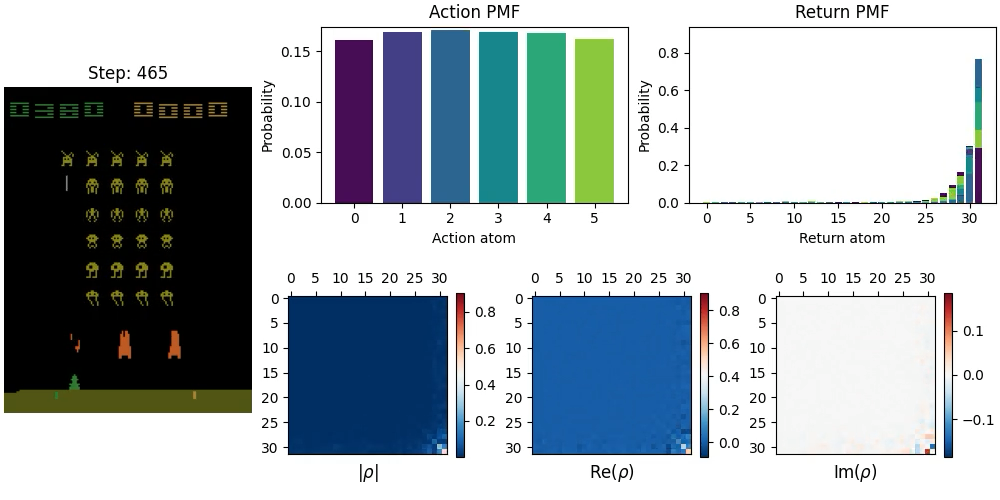}}}
\caption[Learned distributions for \acs{qnrl} and baselines on SpaceInvaders.]{Comparison of environment interaction within \texttt{SpaceInvaders} for best models of \texttt{C51} (a) and \texttt{QnRL} (b). These figures generally show that that \texttt{QnRL} produces return distributions that are either on par or superior to classical baselines.\label{fig:qnrl:exp:distribution:SpaceInvaders}}
\end{figure}

In \cref{fig:qnrl:exp:distribution:SpaceInvaders} we show the \texttt{SpaceInvaders} environment on the best evaluation run for both models, with the final scores being $245$ for \texttt{C51}, and $425$ for \texttt{QnRL} respectively. Here, we focus on time steps which produce similarly-shaped distributions, thus signifying similar environmental conditions, between \texttt{QnRL} and \texttt{C51} during their evaluation.
We see that both models generate return distributions that are left skewed, with impulses near $z_{31}$. In fact, these distributions are very similar to those learned for \texttt{CartPole}, which is interesting considering the entirely different observation data type and triple the action count. Looking at the density matrix of \texttt{QnRL} in \cref{fig:qnrl:exp:SpaceInvaders:distribution:qnrl} we also see the presence of both real and imaginary components. This suggests a higher-order relationship between return atoms exists, and given the nearly double episode score over \texttt{C51} it is beneficial for learning a more performant strategy.

\subsubsection{Distributions for Breakout}


\begin{figure}[t!] 
\centering
\subcaptionbox{\texttt{C51~(9c1e6421)} \label{fig:qnrl:exp:Breakout:distribution:c51}}{\fcolorbox{gray!50}{white}{\includegraphics[width=0.95\linewidth]{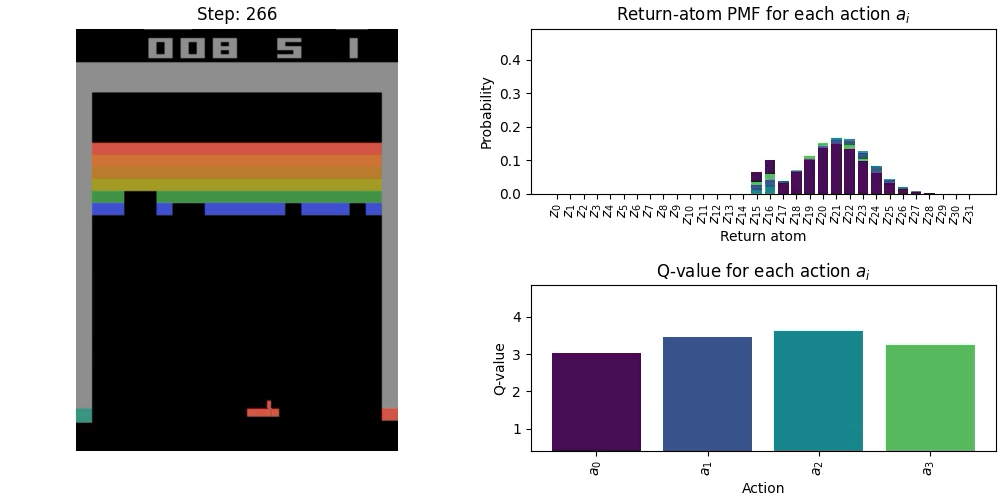}}}
\hfil
\subcaptionbox{\texttt{QnRL~(4d2222f7)} \label{fig:qnrl:exp:Breakout:distribution:qnrl}}{\fcolorbox{gray!50}{white}{\includegraphics[width=0.95\linewidth]{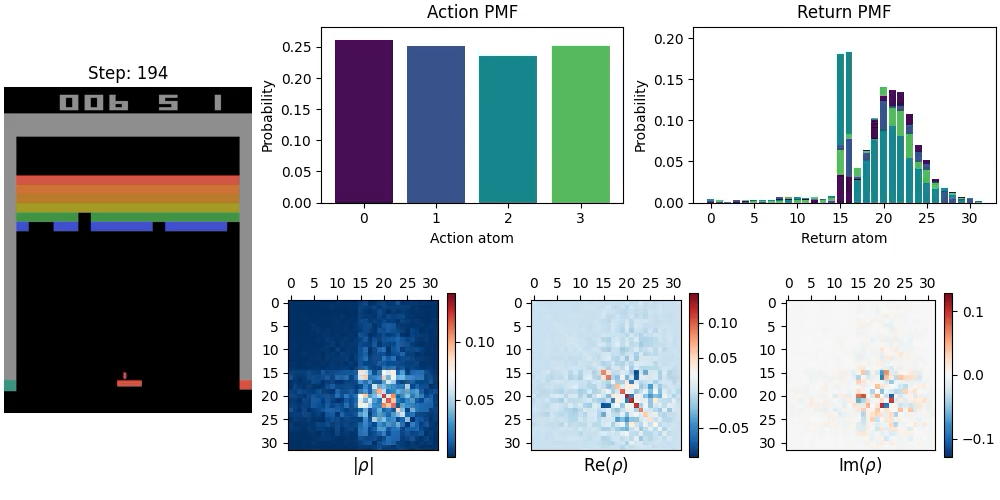}}}
\caption[Learned distributions for \acs{qnrl} and baselines on Breakout.]{Comparison of environment interaction within \texttt{Breakout} for best models of \texttt{C51} (a) and \texttt{QnRL} (b). These figures generally show that that \texttt{QnRL} produces return distributions that are either on par or superior to classical baselines.\label{fig:qnrl:exp:distribution:Breakout}}
\end{figure}

For our last distribution comparison experiment, in \cref{fig:qnrl:exp:distribution:Breakout}, we show the \texttt{Breakout} environment distributions on the best evaluation run for both models, with final scores of $11$ for \texttt{C51}, and $12$ for \texttt{QnRL} respectively. We see that the return distributions of both models are very similarly Gaussian shaped between $[z_{17}, z_{27}]$, and both exhibit a double impulse specifically at $z_{15}$ and $z_{16}$ which represents a ``poor'' action choice. We can see from the ball trajectory in the frame data that both models correctly identify this poor action through the return distributions ($a_{0}=\texttt{NOOP}$ for \texttt{C51}, and $a_{2}=\texttt{right}$ for \texttt{QnRL}), because choosing it would result in the paddle missing the ball, thus losing a player life and resulting in a lower episode reward. Conversely, both models also correctly identify and choose the ``best'' action from their Gaussian return distributions to hit the ball and continue the episode, which are $a_{2}$ for \texttt{C51}, and $a_{0}$ for \texttt{QnRL} respectively. The density matrix in \cref{fig:qnrl:exp:Breakout:distribution:qnrl} shows a clear clustering of amplitudes in the lower-right quadrant and contains both real and imaginary components, which similar to the other environments indicates a complex relationship between the return atoms exists and is captured by the quantum-native model in Hilbert space.

\subsection{Ablation study}\label{sec:qnrl:exp:ablation}

\begin{figure*}[t!] 
\centering
\subcaptionbox{\texttt{CartPole}, \texttt{C51}, Reward \label{fig:qnrl:exp:ablation:c51:CartPole:eval-episode-reward}}{\includegraphics[width=0.48\columnwidth]{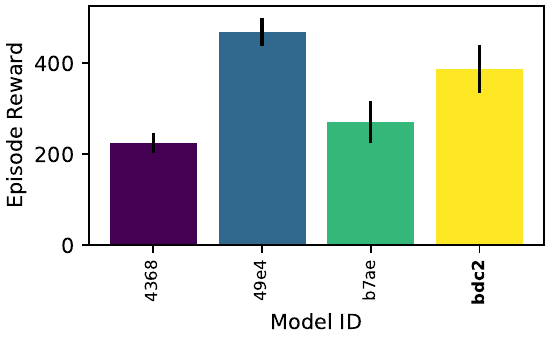}}
\hfil
\subcaptionbox{\texttt{CartPole}, \texttt{QnRL}, Reward \label{fig:qnrl:exp:ablation:qnrl:CartPole:eval-episode-reward}}{\includegraphics[width=0.48\columnwidth]{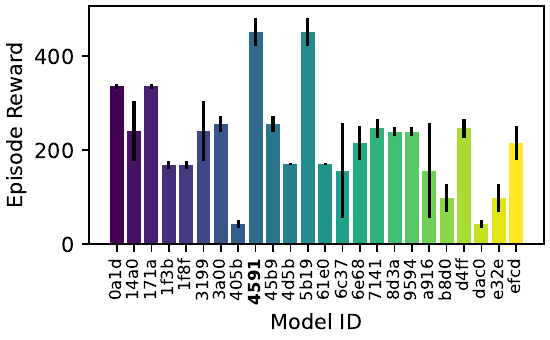}}
\hfil
\subcaptionbox{\texttt{CartPole}, \texttt{C51}, Size \label{fig:qnrl:exp:ablation:c51:CartPole:total-model-parameter-count}}{\includegraphics[width=0.48\columnwidth]{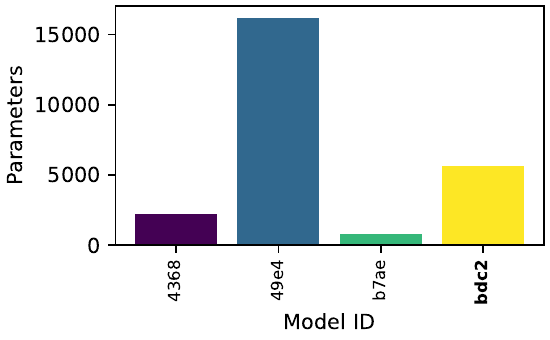}}
\hfil
\subcaptionbox{\texttt{CartPole}, \texttt{QnRL}, Size \label{fig:qnrl:exp:ablation:qnrl:CartPole:total-model-parameter-count}}{\includegraphics[width=0.48\columnwidth]{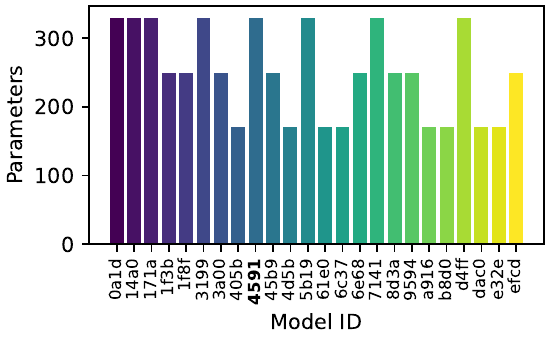}}
\hfil
\subcaptionbox{\texttt{CliffWalking}, \texttt{C51}, Reward \label{fig:qnrl:exp:ablation:c51:CliffWalking:eval-episode-reward}}{\includegraphics[width=0.48\columnwidth]{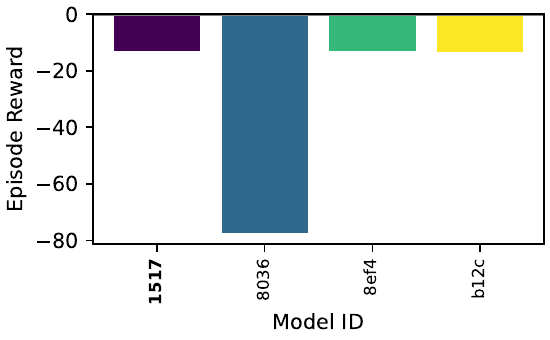}}
\hfil
\subcaptionbox{\texttt{CliffWalking}, \texttt{QnRL}, Reward \label{fig:qnrl:exp:ablation:qnrl:CliffWalking:eval-episode-reward}}{\includegraphics[width=0.48\columnwidth]{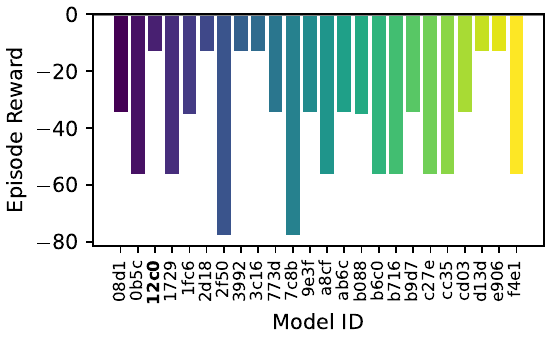}}
\hfil
\subcaptionbox{\texttt{CliffWalking}, \texttt{C51}, Size \label{fig:qnrl:exp:ablation:c51:CliffWalking:total-model-parameter-count}}{\includegraphics[width=0.48\columnwidth]{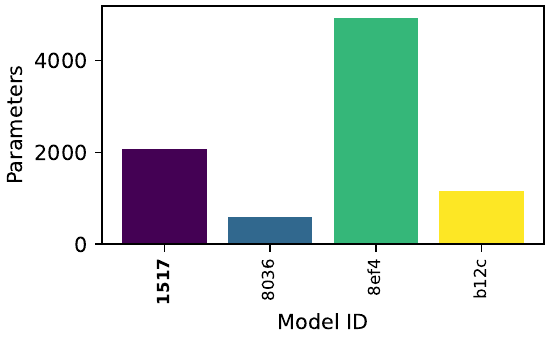}}
\hfil
\subcaptionbox{\texttt{CliffWalking}, \texttt{QnRL}, Size \label{fig:qnrl:exp:ablation:qnrl:CliffWalking:total-model-parameter-count}}{\includegraphics[width=0.48\columnwidth]{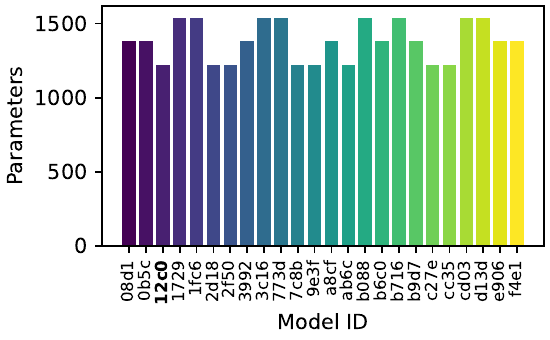}}
\hfil
\subcaptionbox{\texttt{FrozenLake}, \texttt{C51}, Reward \label{fig:qnrl:exp:ablation:c51:FrozenLake:eval-episode-reward}}{\includegraphics[width=0.48\columnwidth]{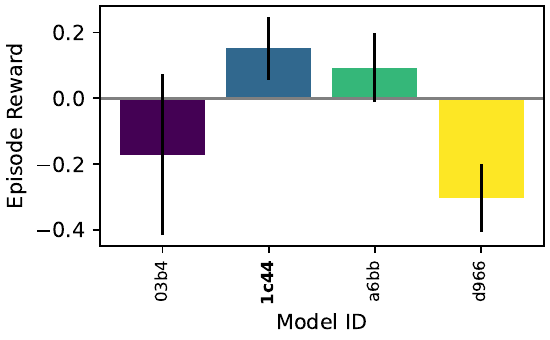}}
\hfil
\subcaptionbox{\texttt{FrozenLake}, \texttt{QnRL}, Reward \label{fig:qnrl:exp:ablation:qnrl:FrozenLake:eval-episode-reward}}{\includegraphics[width=0.48\columnwidth]{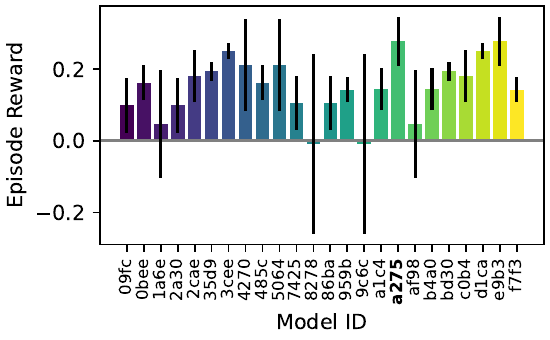}}
\hfil
\subcaptionbox{\texttt{FrozenLake}, \texttt{C51}, Size \label{fig:qnrl:exp:ablation:c51:FrozenLake:total-model-parameter-count}}{\includegraphics[width=0.48\columnwidth]{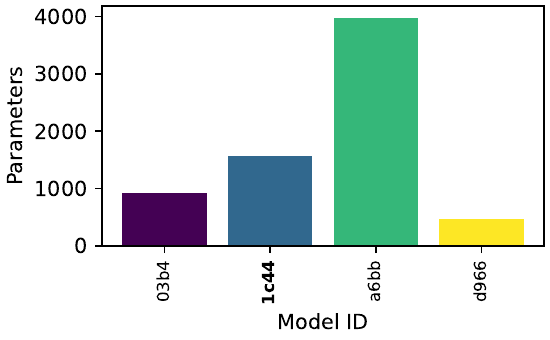}}
\hfil
\subcaptionbox{\texttt{FrozenLake}, \texttt{QnRL}, Size \label{fig:qnrl:exp:ablation:qnrl:FrozenLake:total-model-parameter-count}}{\includegraphics[width=0.48\columnwidth]{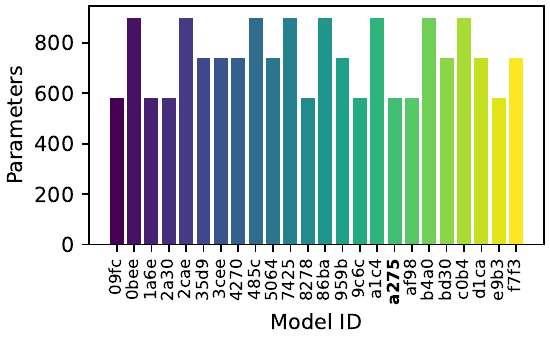}}
\hfil
\subcaptionbox{\texttt{Acrobot}, \texttt{C51}, Reward \label{fig:qnrl:exp:ablation:c51:Acrobot:eval-episode-reward}}{\includegraphics[width=0.48\columnwidth]{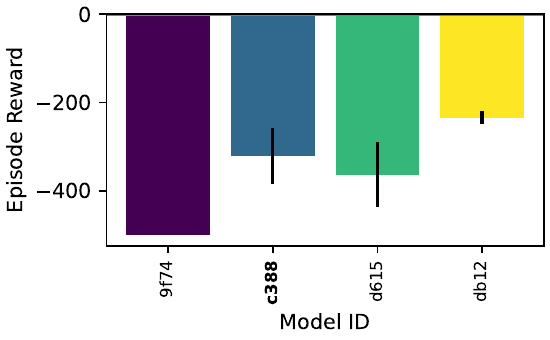}}
\hfil
\subcaptionbox{\texttt{Acrobot}, \texttt{QnRL}, Reward \label{fig:qnrl:exp:ablation:qnrl:Acrobot:eval-episode-reward}}{\includegraphics[width=0.48\columnwidth]{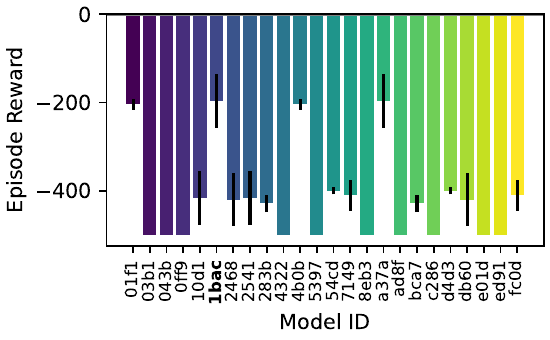}}
\hfil
\subcaptionbox{\texttt{Acrobot}, \texttt{C51}, Size \label{fig:qnrl:exp:ablation:c51:Acrobot:total-model-parameter-count}}{\includegraphics[width=0.48\columnwidth]{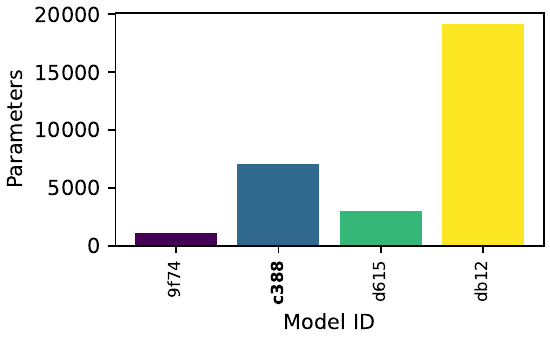}}
\hfil
\subcaptionbox{\texttt{Acrobot}, \texttt{QnRL}, Size \label{fig:qnrl:exp:ablation:qnrl:Acrobot:total-model-parameter-count}}{\includegraphics[width=0.48\columnwidth]{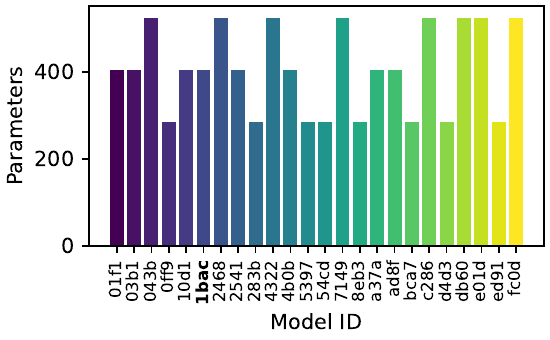}}
\hfil
\subcaptionbox{\texttt{SpaceInvaders}, \texttt{C51}, Reward \label{fig:qnrl:exp:ablation:c51:SpaceInvaders:eval-episode-reward}}{\includegraphics[width=0.48\columnwidth]{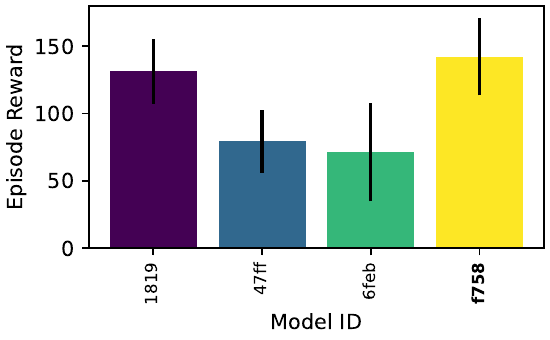}}
\hfil
\subcaptionbox{\texttt{SpaceInvaders}, \texttt{QnRL}, Reward \label{fig:qnrl:exp:ablation:qnrl:SpaceInvaders:eval-episode-reward}}{\includegraphics[width=0.48\columnwidth]{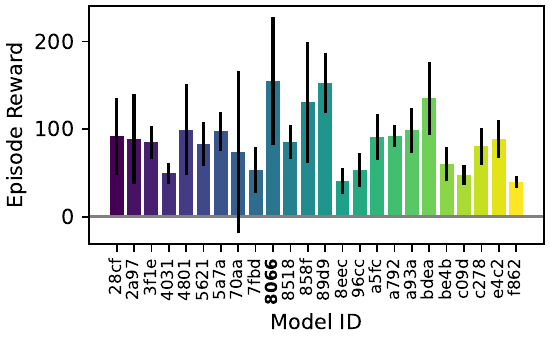}}
\hfil
\subcaptionbox{\texttt{SpaceInvaders}, \texttt{C51}, Size \label{fig:qnrl:exp:ablation:c51:SpaceInvaders:total-model-parameter-count}}{\includegraphics[width=0.48\columnwidth]{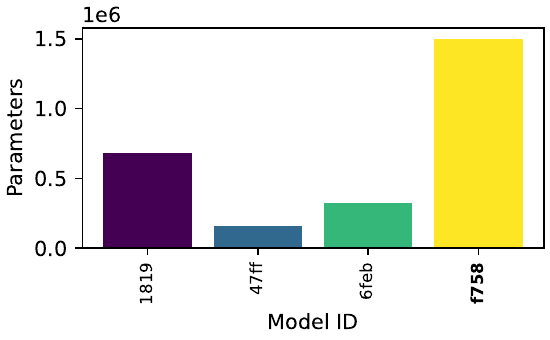}}
\hfil
\subcaptionbox{\texttt{SpaceInvaders}, \texttt{QnRL}, Size \label{fig:qnrl:exp:ablation:qnrl:SpaceInvaders:total-model-parameter-count}}{\includegraphics[width=0.48\columnwidth]{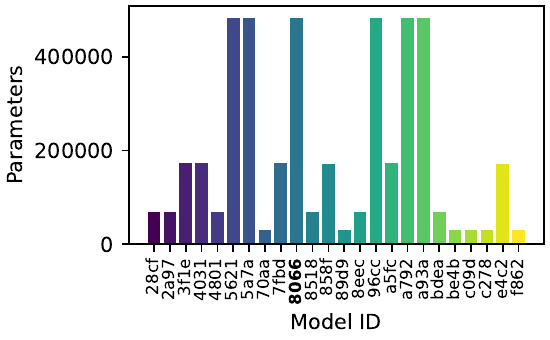}}
\hfil
\subcaptionbox{\texttt{Breakout}, \texttt{C51}, Reward \label{fig:qnrl:exp:ablation:c51:Breakout:eval-episode-reward}}{\includegraphics[width=0.48\columnwidth]{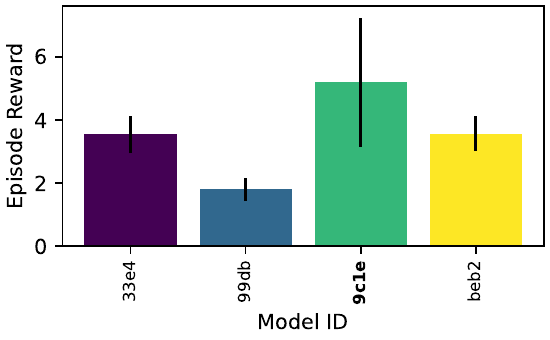}}
\hfil
\subcaptionbox{\texttt{Breakout}, \texttt{QnRL}, Reward \label{fig:qnrl:exp:ablation:qnrl:Breakout:eval-episode-reward}}{\includegraphics[width=0.48\columnwidth]{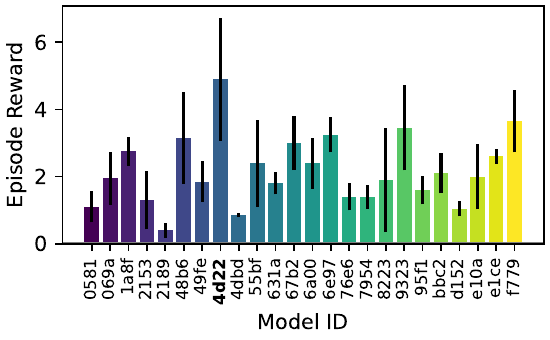}}
\hfil
\subcaptionbox{\texttt{Breakout}, \texttt{C51}, Size \label{fig:qnrl:exp:ablation:c51:Breakout:total-model-parameter-count}}{\includegraphics[width=0.48\columnwidth]{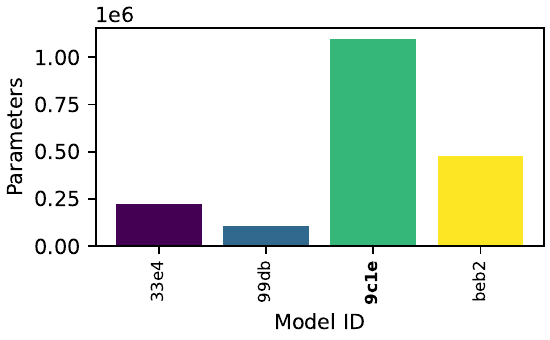}}
\hfil
\subcaptionbox{\texttt{Breakout}, \texttt{QnRL}, Size \label{fig:qnrl:exp:ablation:qnrl:Breakout:total-model-parameter-count}}{\includegraphics[width=0.48\columnwidth]{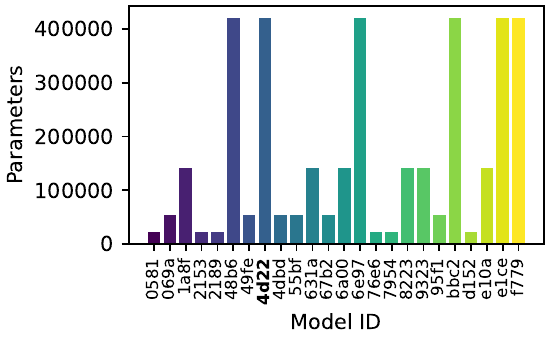}}
\caption[Ablation study for \texttt{QnRL} and \texttt{C51}.]{Ablation study for models (a,c,e,g,i,k,m,o,q,s,u,w)~\texttt{C51} and (b,d,f,h,j,l,n,p,r,t,v,x)~\texttt{QnRL} on environments (a)-(d)~\texttt{CartPole}, (e)-(h)~\texttt{CliffWalking}, (i)-(l)~\texttt{FrozenLake}, (m)-(p)~\texttt{Acrobot}, (q)-(t)~\texttt{SpaceInvaders}, and (u)-(x)~\texttt{Breakout} showing metrics (a,b,e,f,i,j,m,n,q,r,u,v)~evaluation reward and (c,d,g,h,k,l,o,p,s,t,w,x)~model size. Plots are organized by the first four characters of the model configuration hash identifier. Each model configuration was trained on 4 unique seeds, and each training seed was evaluated on 10 unique seeds. The best configuration for each model is highlighted in \textbf{bold}.}
\label{fig:qnrl:exp:ablation}
\end{figure*}

\begin{table*}[t!]
\centering
\footnotesize
\caption[Ablation study for \texttt{C51}.]{Ablation study for \texttt{C51} with varying hidden layer units across the classic control and grid-world environments. Compares evaluation performance results, total model parameters, and all hyperparameters. Selected best model configuration is highlighted in \textbf{bold}. Rows are sorted by model configuration identifier for each respective environment.\label{tab:qnrl:ablation:c51:Gym}}
\begin{tabular}{llllllllll}
\toprule
Environment & Model ID & Mean Eval Reward & Parameters & $\abs{x}$ & $\abs{\setaction}$ & $\abs{\setreturn}$ & $z_{\textrm{min}}$ & $z_{\textrm{max}}$ & Hidden Layers \\
\midrule
\multirow{4}{*}{\texttt{CartPole}} 
& \texttt{4368ff55} & $224.6$ & $2209$ & $4$ & $2$ & $32$ & $-100$ & $100$ & $[30,21]$ \\
& \texttt{49e46a87} & $469.3$ & $16204$ & $4$ & $2$ & $32$ & $-100$ & $100$ & $[120,84]$ \\
& \texttt{b7aea040} & $270.15$ & $792$ & $4$ & $2$ & $32$ & $-100$ & $100$ & $[16,8]$ \\
& \textbf{\texttt{bdc26ac3}} & $387.2250$ & $5614$ & $4$ & $2$ & $32$ & $-100$ & $100$ & $[60,42]$ \\
\cmidrule(lr){1-10}
\multirow{4}{*}{\texttt{CliffWalking}} 
& \textbf{\texttt{15173423}} & $-13$ & $2072$ & $48$ & $4$ & $32$ & $-100$ & $100$ & $[16,8]$ \\
& \texttt{8036b0d0} & $-77.5$ & $590$ & $48$ & $4$ & $32$ & $-100$ & $100$ & $[4,2]$ \\
& \texttt{8ef48998} & $-13$ & $4937$ & $48$ & $4$ & $32$ & $-100$ & $100$ & $[30,21]$ \\
& \texttt{b12ce064} & $-13.5$ & $1151$ & $48$ & $4$ & $32$ & $-100$ & $100$ & $[7,5]$ \\
\cmidrule(lr){1-10}
\multirow{4}{*}{\texttt{FrozenLake}} 
& \texttt{03b4d6ac} & $-0.1708$ & $927$ & $16$ & $4$ & $32$ & $-2$ & $2$ & $[7,5]$ \\
& \textbf{\texttt{1c44782b}} & $0.1517$ & $1560$ & $16$ & $4$ & $32$ & $-2$ & $2$ & $[16,8]$ \\
& \texttt{a6bb645e} & $0.0925$ & $3977$ & $16$ & $4$ & $32$ & $-2$ & $2$ & $[30,21]$ \\
& \texttt{d9663ea7} & $-0.3035$ & $462$ & $16$ & $4$ & $32$ & $-2$ & $2$ & $[4,2]$ \\
\cmidrule(lr){1-10}
\multirow{4}{*}{\texttt{Acrobot}} 
& \texttt{9f7469c2} & $-500$ & $1112$ & $6$ & $3$ & $32$ & $-100$ & $100$ & $[16,8]$ \\
& \textbf{\texttt{c3881283}} & $-320.9$ & $7110$ & $6$ & $3$ & $32$ & $-100$ & $100$ & $[60,42]$ \\
& \texttt{d615efcf} & $-363.35$ & $2973$ & $6$ & $3$ & $32$ & $-100$ & $100$ & $[30,21]$ \\
& \texttt{db125df2} & $-234.075$ & $19164$ & $6$ & $3$ & $32$ & $-100$ & $100$ & $[120,84]$ \\
\bottomrule
\end{tabular}
\end{table*}

\begin{table*}[t!]
\centering
\caption[Ablation study for \texttt{C51} on Atari environments.]{Ablation study for \texttt{C51} with varying convolution features across the Atari environments. Compares evaluation performance results, total model parameters, and all hyperparameters. Selected best model configuration is highlighted in \textbf{bold}. Rows are sorted by model configuration identifier for each respective environment.\label{tab:qnrl:ablation:c51:Atari}}
\resizebox{\linewidth}{!}{
\begin{tabular}{llllllllllll}
\toprule
Environment & Model ID & Mean Eval Reward & Parameters & Observation dim. & $\abs{\setaction}$ & $\abs{\setreturn}$ & $z_{\textrm{min}}$ & $z_{\textrm{max}}$ & Conv. Features & Kernels & Strides \\
\midrule
\multirow{4}{*}{\texttt{SpaceInvaders}} 
& \texttt{1819f2cb} & $131.5$ & $680288$ & $4\times84\times84$ & $6$ & $32$ & $-10$ & $10$ & $[32,64,64]$ & $[8, 4, 3]$ & $[4, 2, 1]$ \\
& \texttt{47ffe1c4} & $79.25$ & $157160$ & $4\times84\times84$ & $6$ & $32$ & $-10$ & $10$ & $[8,16,16]$ & $[8, 4, 3]$ & $[4, 2, 1]$ \\
& \texttt{6feb00c2} & $71.25$ & $322832$ & $4\times84\times84$ & $6$ & $32$ & $-10$ & $10$ & $[16,32,32]$ & $[8, 4, 3]$ & $[4, 2, 1]$ \\
& \textbf{\texttt{f758df04}} & $142.25$ & $1499648$ & $4\times84\times84$ & $6$ & $32$ & $-10$ & $10$ & $[64,128,128]$ & $[8, 4, 3]$ & $[4, 2, 1]$ \\
\cmidrule(lr){1-12}
\multirow{4}{*}{\texttt{Breakout}} 
& \texttt{33e43f68} & $3.55$ & $222416$ & $4\times84\times84$ & $4$ & $32$ & $-10$ & $10$ & $[16,32,32]$ & $[8, 4, 3]$ & $[4, 2, 1]$ \\
& \texttt{99dbcbf7} & $1.05$ & $106920$ & $4\times84\times84$ & $4$ & $32$ & $-10$ & $10$ & $[8,16,16]$ & $[8, 4, 3]$ & $[4, 2, 1]$ \\
& \textbf{\texttt{9c1e6421}} & $5.2$ & $1098176$ & $4\times84\times84$ & $4$ & $32$ & $-10$ & $10$ & $[64,128,128]$ & $[8, 4, 3]$ & $[4, 2, 1]$ \\
& \texttt{beb2a303} & $3.55$ & $479520$ & $4\times84\times84$ & $4$ & $32$ & $-10$ & $10$ & $[32,64,64]$ & $[8, 4, 3]$ & $[4, 2, 1]$ \\
\bottomrule
\end{tabular}
}
\end{table*}
\begin{table*}[t!]
\centering
\tiny
\caption[Ablation study for \texttt{QnRL}.]{Ablation study for \texttt{QnRL} with varying circuit depth $L$, expectation power $n$, moment index $m$, and entanglement style $U_{\textrm{ent}}$ across the classic control and grid-world environments. Compares evaluation performance results, total model parameters, and all hyperparameters. Selected best model configuration is highlighted in \textbf{bold}. Rows are sorted by model configuration identifier for each respective environment.\label{tab:qnrl:ablation:qnrl:Gym}}
\begin{tabular}{llllllllllllllll}
\toprule
Environment & Model ID & Mean Eval Reward & Parameters & $\abs{x}$ & $\abs{\setaction}$ & $\abs{\setreturn}$ & $z_{\textrm{min}}$ & $z_{\textrm{max}}$ & $n$ & $m$ & $L$ & $U_{\textrm{enc}}$ & $U_{\textrm{var}}$ & $U_{\textrm{ent}}$ & $\phi$ \\
\midrule
\multirow{24}{*}{\texttt{CartPole}} 
& \texttt{0a1d6db2} & $336.15$ & $330$ & $4$ & $2$ & $32$ & $-100$ & $100$ & $2$ & $2$ & $7$ & $[RX]$ & $[RZ,RY,RZ]$ & \texttt{offset-CZ} & \texttt{tanh} \\
& \texttt{14a05a3d} & $241.625$ & $330$ & $4$ & $2$ & $32$ & $-100$ & $100$ & $2$ & $1$ & $7$ & $[RX]$ & $[RZ,RY,RZ]$ & \texttt{circ-CZ} & \texttt{tanh} \\
& \texttt{171a91d1} & $336.15$ & $330$ & $4$ & $2$ & $32$ & $-100$ & $100$ & $1$ & $2$ & $7$ & $[RX]$ & $[RZ,RY,RZ]$ & \texttt{offset-CZ} & \texttt{tanh} \\
& \texttt{1f3b03a6} & $168.15$ & $250$ & $4$ & $2$ & $32$ & $-100$ & $100$ & $2$ & $2$ & $5$ & $[RX]$ & $[RZ,RY,RZ]$ & \texttt{offset-CZ} & \texttt{tanh} \\
& \texttt{1f8ffc6f} & $168.15$ & $250$ & $4$ & $2$ & $32$ & $-100$ & $100$ & $1$ & $2$ & $5$ & $[RX]$ & $[RZ,RY,RZ]$ & \texttt{offset-CZ} & \texttt{tanh} \\
& \texttt{31994c7a} & $241.625$ & $330$ & $4$ & $2$ & $32$ & $-100$ & $100$ & $1$ & $1$ & $7$ & $[RX]$ & $[RZ,RY,RZ]$ & \texttt{circ-CZ} & \texttt{tanh} \\
& \texttt{3a00cecd} & $256.475$ & $250$ & $4$ & $2$ & $32$ & $-100$ & $100$ & $1$ & $1$ & $5$ & $[RX]$ & $[RZ,RY,RZ]$ & \texttt{circ-CZ} & \texttt{tanh} \\
& \texttt{405bd6e3} & $43.025$ & $170$ & $4$ & $2$ & $32$ & $-100$ & $100$ & $2$ & $2$ & $3$ & $[RX]$ & $[RZ,RY,RZ]$ & \texttt{offset-CZ} & \texttt{tanh} \\
& \textbf{\texttt{45913de1}} & $452.3$ & $330$ & $4$ & $2$ & $32$ & $-100$ & $100$ & $1$ & $1$ & $7$ & $[RX]$ & $[RZ,RY,RZ]$ & \texttt{offset-CZ} & \texttt{tanh} \\
& \texttt{45b9d666} & $256.475$ & $250$ & $4$ & $2$ & $32$ & $-100$ & $100$ & $2$ & $1$ & $5$ & $[RX]$ & $[RZ,RY,RZ]$ & \texttt{circ-CZ} & \texttt{tanh} \\
& \texttt{4d5bc0c3} & $170.175$ & $170$ & $4$ & $2$ & $32$ & $-100$ & $100$ & $2$ & $2$ & $3$ & $[RX]$ & $[RZ,RY,RZ]$ & \texttt{circ-CZ} & \texttt{tanh} \\
& \texttt{5b19a4fb} & $452.3$ & $330$ & $4$ & $2$ & $32$ & $-100$ & $100$ & $2$ & $1$ & $7$ & $[RX]$ & $[RZ,RY,RZ]$ & \texttt{offset-CZ} & \texttt{tanh} \\
& \texttt{61e03e86} & $170.175$ & $170$ & $4$ & $2$ & $32$ & $-100$ & $100$ & $1$ & $2$ & $3$ & $[RX]$ & $[RZ,RY,RZ]$ & \texttt{circ-CZ} & \texttt{tanh} \\
& \texttt{6c378092} & $156.375$ & $170$ & $4$ & $2$ & $32$ & $-100$ & $100$ & $2$ & $1$ & $3$ & $[RX]$ & $[RZ,RY,RZ]$ & \texttt{circ-CZ} & \texttt{tanh} \\
& \texttt{6e68ed3c} & $215.175$ & $250$ & $4$ & $2$ & $32$ & $-100$ & $100$ & $2$ & $1$ & $5$ & $[RX]$ & $[RZ,RY,RZ]$ & \texttt{offset-CZ} & \texttt{tanh} \\
& \texttt{71415254} & $246.425$ & $330$ & $4$ & $2$ & $32$ & $-100$ & $100$ & $2$ & $2$ & $7$ & $[RX]$ & $[RZ,RY,RZ]$ & \texttt{circ-CZ} & \texttt{tanh} \\
& \texttt{8d3a1eb9} & $239.15$ & $250$ & $4$ & $2$ & $32$ & $-100$ & $100$ & $1$ & $2$ & $5$ & $[RX]$ & $[RZ,RY,RZ]$ & \texttt{circ-CZ} & \texttt{tanh} \\
& \texttt{9594b39b} & $239.15$ & $250$ & $4$ & $2$ & $32$ & $-100$ & $100$ & $2$ & $2$ & $5$ & $[RX]$ & $[RZ,RY,RZ]$ & \texttt{circ-CZ} & \texttt{tanh} \\
& \texttt{a916b178} & $156.375$ & $170$ & $4$ & $2$ & $32$ & $-100$ & $100$ & $1$ & $1$ & $3$ & $[RX]$ & $[RZ,RY,RZ]$ & \texttt{circ-CZ} & \texttt{tanh} \\
& \texttt{b8d03cbf} & $98.05$ & $170$ & $4$ & $2$ & $32$ & $-100$ & $100$ & $1$ & $1$ & $3$ & $[RX]$ & $[RZ,RY,RZ]$ & \texttt{offset-CZ} & \texttt{tanh} \\
& \texttt{d4ff21cb} & $246.425$ & $330$ & $4$ & $2$ & $32$ & $-100$ & $100$ & $1$ & $2$ & $7$ & $[RX]$ & $[RZ,RY,RZ]$ & \texttt{circ-CZ} & \texttt{tanh} \\
& \texttt{dac0fe4a} & $43.025$ & $170$ & $4$ & $2$ & $32$ & $-100$ & $100$ & $1$ & $2$ & $3$ & $[RX]$ & $[RZ,RY,RZ]$ & \texttt{offset-CZ} & \texttt{tanh} \\
& \texttt{e32e8801} & $98.05$ & $170$ & $4$ & $2$ & $32$ & $-100$ & $100$ & $2$ & $1$ & $3$ & $[RX]$ & $[RZ,RY,RZ]$ & \texttt{offset-CZ} & \texttt{tanh} \\
& \texttt{efcd9235} & $215.175$ & $250$ & $4$ & $2$ & $32$ & $-100$ & $100$ & $1$ & $1$ & $5$ & $[RX]$ & $[RZ,RY,RZ]$ & \texttt{offset-CZ} & \texttt{tanh} \\
\cmidrule(lr){1-16}
\multirow{24}{*}{\texttt{CliffWalking}} 
& \texttt{08d163cb} & $-34.5$ & $1380$ & $48$ & $4$ & $32$ & $-100$ & $100$ & $1$ & $1$ & $5$ & $[RX]$ & $[RZ,RY,RZ]$ & \texttt{offset-CZ} & \texttt{tanh} \\
& \texttt{0b5c3b76} & $-56.0$ & $1380$ & $48$ & $4$ & $32$ & $-100$ & $100$ & $1$ & $2$ & $5$ & $[RX]$ & $[RZ,RY,RZ]$ & \texttt{circ-CZ} & \texttt{tanh} \\
& \textbf{\texttt{12c045e6}} & $-13.0$ & $1220$ & $48$ & $4$ & $32$ & $-100$ & $100$ & $1$ & $1$ & $3$ & $[RX]$ & $[RZ,RY,RZ]$ & \texttt{circ-CZ} & \texttt{tanh} \\
& \texttt{1729a977} & $-56.0$ & $1540$ & $48$ & $4$ & $32$ & $-100$ & $100$ & $1$ & $2$ & $7$ & $[RX]$ & $[RZ,RY,RZ]$ & \texttt{circ-CZ} & \texttt{tanh} \\
& \texttt{1fc6943d} & $-35.0$ & $1540$ & $48$ & $4$ & $32$ & $-100$ & $100$ & $4$ & $1$ & $7$ & $[RX]$ & $[RZ,RY,RZ]$ & \texttt{offset-CZ} & \texttt{tanh} \\
& \texttt{2d18e0c2} & $-13.0$ & $1220$ & $48$ & $4$ & $32$ & $-100$ & $100$ & $4$ & $1$ & $3$ & $[RX]$ & $[RZ,RY,RZ]$ & \texttt{circ-CZ} & \texttt{tanh} \\
& \texttt{2f5083d1} & $-77.5$ & $1220$ & $48$ & $4$ & $32$ & $-100$ & $100$ & $1$ & $2$ & $3$ & $[RX]$ & $[RZ,RY,RZ]$ & \texttt{circ-CZ} & \texttt{tanh} \\
& \texttt{3992cb09} & $-13.0$ & $1380$ & $48$ & $4$ & $32$ & $-100$ & $100$ & $1$ & $1$ & $5$ & $[RX]$ & $[RZ,RY,RZ]$ & \texttt{circ-CZ} & \texttt{tanh} \\
& \texttt{3c16ef96} & $-13.0$ & $1540$ & $48$ & $4$ & $32$ & $-100$ & $100$ & $1$ & $2$ & $7$ & $[RX]$ & $[RZ,RY,RZ]$ & \texttt{offset-CZ} & \texttt{tanh} \\
& \texttt{773dd844} & $-34.5$ & $1540$ & $48$ & $4$ & $32$ & $-100$ & $100$ & $1$ & $1$ & $7$ & $[RX]$ & $[RZ,RY,RZ]$ & \texttt{circ-CZ} & \texttt{tanh} \\
& \texttt{7c8bba34} & $-77.5$ & $1220$ & $48$ & $4$ & $32$ & $-100$ & $100$ & $4$ & $2$ & $3$ & $[RX]$ & $[RZ,RY,RZ]$ & \texttt{circ-CZ} & \texttt{tanh} \\
& \texttt{9e3f8627} & $-34.5$ & $1220$ & $48$ & $4$ & $32$ & $-100$ & $100$ & $1$ & $1$ & $3$ & $[RX]$ & $[RZ,RY,RZ]$ & \texttt{offset-CZ} & \texttt{tanh} \\
& \texttt{a8cf904d} & $-56.0$ & $1380$ & $48$ & $4$ & $32$ & $-100$ & $100$ & $4$ & $2$ & $5$ & $[RX]$ & $[RZ,RY,RZ]$ & \texttt{circ-CZ} & \texttt{tanh} \\
& \texttt{ab6c0436} & $-34.5$ & $1220$ & $48$ & $4$ & $32$ & $-100$ & $100$ & $4$ & $1$ & $3$ & $[RX]$ & $[RZ,RY,RZ]$ & \texttt{offset-CZ} & \texttt{tanh} \\
& \texttt{b0883a8f} & $-35.0$ & $1540$ & $48$ & $4$ & $32$ & $-100$ & $100$ & $1$ & $1$ & $7$ & $[RX]$ & $[RZ,RY,RZ]$ & \texttt{offset-CZ} & \texttt{tanh} \\
& \texttt{b6c0c9f7} & $-56.0$ & $1380$ & $48$ & $4$ & $32$ & $-100$ & $100$ & $1$ & $2$ & $5$ & $[RX]$ & $[RZ,RY,RZ]$ & \texttt{offset-CZ} & \texttt{tanh} \\
& \texttt{b716ab69} & $-56.0$ & $1540$ & $48$ & $4$ & $32$ & $-100$ & $100$ & $4$ & $2$ & $7$ & $[RX]$ & $[RZ,RY,RZ]$ & \texttt{circ-CZ} & \texttt{tanh} \\
& \texttt{b9d762f3} & $-34.5$ & $1380$ & $48$ & $4$ & $32$ & $-100$ & $100$ & $4$ & $1$ & $5$ & $[RX]$ & $[RZ,RY,RZ]$ & \texttt{offset-CZ} & \texttt{tanh} \\
& \texttt{c27efbf3} & $-56.0$ & $1220$ & $48$ & $4$ & $32$ & $-100$ & $100$ & $1$ & $2$ & $3$ & $[RX]$ & $[RZ,RY,RZ]$ & \texttt{offset-CZ} & \texttt{tanh} \\
& \texttt{cc3545d4} & $-56.0$ & $1220$ & $48$ & $4$ & $32$ & $-100$ & $100$ & $4$ & $2$ & $3$ & $[RX]$ & $[RZ,RY,RZ]$ & \texttt{offset-CZ} & \texttt{tanh} \\
& \texttt{cd036ad7} & $-34.5$ & $1540$ & $48$ & $4$ & $32$ & $-100$ & $100$ & $4$ & $1$ & $7$ & $[RX]$ & $[RZ,RY,RZ]$ & \texttt{circ-CZ} & \texttt{tanh} \\
& \texttt{d13d8298} & $-13.0$ & $1540$ & $48$ & $4$ & $32$ & $-100$ & $100$ & $4$ & $2$ & $7$ & $[RX]$ & $[RZ,RY,RZ]$ & \texttt{offset-CZ} & \texttt{tanh} \\
& \texttt{e90620d5} & $-13.0$ & $1380$ & $48$ & $4$ & $32$ & $-100$ & $100$ & $4$ & $1$ & $5$ & $[RX]$ & $[RZ,RY,RZ]$ & \texttt{circ-CZ} & \texttt{tanh} \\
& \texttt{f4e1e3f2} & $-56.0$ & $1380$ & $48$ & $4$ & $32$ & $-100$ & $100$ & $4$ & $2$ & $5$ & $[RX]$ & $[RZ,RY,RZ]$ & \texttt{offset-CZ} & \texttt{tanh} \\
\cmidrule(lr){1-16}
\multirow{24}{*}{\texttt{FrozenLake}} 
& \texttt{09fc2cfb} & $0.0990$ & $580$ & $16$ & $4$ & $32$ & $-2$ & $2$ & $4$ & $2$ & $3$ & $[RX]$ & $[RZ,RY,RZ]$ & \texttt{circ-CZ} & \texttt{tanh} \\
& \texttt{0beeeeea} & $0.1623$ & $900$ & $16$ & $4$ & $32$ & $-2$ & $2$ & $1$ & $2$ & $7$ & $[RX]$ & $[RZ,RY,RZ]$ & \texttt{offset-CZ} & \texttt{tanh} \\
& \texttt{1a6e23d3} & $0.0480$ & $580$ & $16$ & $4$ & $32$ & $-2$ & $2$ & $1$ & $1$ & $3$ & $[RX]$ & $[RZ,RY,RZ]$ & \texttt{offset-CZ} & \texttt{tanh} \\
& \texttt{2a30337c} & $0.0990$ & $580$ & $16$ & $4$ & $32$ & $-2$ & $2$ & $1$ & $2$ & $3$ & $[RX]$ & $[RZ,RY,RZ]$ & \texttt{circ-CZ} & \texttt{tanh} \\
& \texttt{2cae8b90} & $0.1815$ & $900$ & $16$ & $4$ & $32$ & $-2$ & $2$ & $4$ & $1$ & $7$ & $[RX]$ & $[RZ,RY,RZ]$ & \texttt{circ-CZ} & \texttt{tanh} \\
& \texttt{35d94c95} & $0.1937$ & $740$ & $16$ & $4$ & $32$ & $-2$ & $2$ & $4$ & $2$ & $5$ & $[RX]$ & $[RZ,RY,RZ]$ & \texttt{offset-CZ} & \texttt{tanh} \\
& \texttt{3cee6e98} & $0.2500$ & $740$ & $16$ & $4$ & $32$ & $-2$ & $2$ & $1$ & $2$ & $5$ & $[RX]$ & $[RZ,RY,RZ]$ & \texttt{circ-CZ} & \texttt{tanh} \\
& \texttt{4270289a} & $0.2120$ & $740$ & $16$ & $4$ & $32$ & $-2$ & $2$ & $4$ & $1$ & $5$ & $[RX]$ & $[RZ,RY,RZ]$ & \texttt{offset-CZ} & \texttt{tanh} \\
& \texttt{485c0bf8} & $0.1622$ & $900$ & $16$ & $4$ & $32$ & $-2$ & $2$ & $4$ & $2$ & $7$ & $[RX]$ & $[RZ,RY,RZ]$ & \texttt{offset-CZ} & \texttt{tanh} \\
& \texttt{5064f085} & $0.2120$ & $740$ & $16$ & $4$ & $32$ & $-2$ & $2$ & $1$ & $1$ & $5$ & $[RX]$ & $[RZ,RY,RZ]$ & \texttt{offset-CZ} & \texttt{tanh} \\
& \texttt{7425f8af} & $0.1055$ & $900$ & $16$ & $4$ & $32$ & $-2$ & $2$ & $1$ & $2$ & $7$ & $[RX]$ & $[RZ,RY,RZ]$ & \texttt{circ-CZ} & \texttt{tanh} \\
& \texttt{82786234} & $-0.0078$ & $580$ & $16$ & $4$ & $32$ & $-2$ & $2$ & $4$ & $1$ & $3$ & $[RX]$ & $[RZ,RY,RZ]$ & \texttt{circ-CZ} & \texttt{tanh} \\
& \texttt{86baa0b5} & $0.1055$ & $900$ & $16$ & $4$ & $32$ & $-2$ & $2$ & $4$ & $2$ & $7$ & $[RX]$ & $[RZ,RY,RZ]$ & \texttt{circ-CZ} & \texttt{tanh} \\
& \texttt{959b1d9b} & $0.1427$ & $740$ & $16$ & $4$ & $32$ & $-2$ & $2$ & $4$ & $1$ & $5$ & $[RX]$ & $[RZ,RY,RZ]$ & \texttt{circ-CZ} & \texttt{tanh} \\
& \texttt{9c6ce5eb} & $-0.0078$ & $580$ & $16$ & $4$ & $32$ & $-2$ & $2$ & $1$ & $1$ & $3$ & $[RX]$ & $[RZ,RY,RZ]$ & \texttt{circ-CZ} & \texttt{tanh} \\
& \texttt{a1c4fa8c} & $0.1443$ & $900$ & $16$ & $4$ & $32$ & $-2$ & $2$ & $1$ & $1$ & $7$ & $[RX]$ & $[RZ,RY,RZ]$ & \texttt{offset-CZ} & \texttt{tanh} \\
& \textbf{\texttt{a275ee45}} & $0.2775$ & $580$ & $16$ & $4$ & $32$ & $-2$ & $2$ & $1$ & $2$ & $3$ & $[RX]$ & $[RZ,RY,RZ]$ & \texttt{offset-CZ} & \texttt{tanh} \\
& \texttt{af98683f} & $0.0480$ & $580$ & $16$ & $4$ & $32$ & $-2$ & $2$ & $4$ & $1$ & $3$ & $[RX]$ & $[RZ,RY,RZ]$ & \texttt{offset-CZ} & \texttt{tanh} \\
& \texttt{b4a0f93c} & $0.1443$ & $900$ & $16$ & $4$ & $32$ & $-2$ & $2$ & $4$ & $1$ & $7$ & $[RX]$ & $[RZ,RY,RZ]$ & \texttt{offset-CZ} & \texttt{tanh} \\
& \texttt{bd30b45a} & $0.1937$ & $740$ & $16$ & $4$ & $32$ & $-2$ & $2$ & $1$ & $2$ & $5$ & $[RX]$ & $[RZ,RY,RZ]$ & \texttt{offset-CZ} & \texttt{tanh} \\
& \texttt{c0b410e0} & $0.1815$ & $900$ & $16$ & $4$ & $32$ & $-2$ & $2$ & $1$ & $1$ & $7$ & $[RX]$ & $[RZ,RY,RZ]$ & \texttt{circ-CZ} & \texttt{tanh} \\
& \texttt{d1ca5afe} & $0.2500$ & $740$ & $16$ & $4$ & $32$ & $-2$ & $2$ & $4$ & $2$ & $5$ & $[RX]$ & $[RZ,RY,RZ]$ & \texttt{circ-CZ} & \texttt{tanh} \\
& \texttt{e9b33425} & $0.2775$ & $580$ & $16$ & $4$ & $32$ & $-2$ & $2$ & $4$ & $2$ & $3$ & $[RX]$ & $[RZ,RY,RZ]$ & \texttt{offset-CZ} & \texttt{tanh} \\
& \texttt{f7f319ca} & $0.1427$ & $740$ & $16$ & $4$ & $32$ & $-2$ & $2$ & $1$ & $1$ & $5$ & $[RX]$ & $[RZ,RY,RZ]$ & \texttt{circ-CZ} & \texttt{tanh} \\
\cmidrule(lr){1-16}
\multirow{24}{*}{\texttt{Acrobot}} 
& \texttt{01f1b2a8} & $-204.2750$ & $405$ & $6$ & $3$ & $32$ & $-100$ & $100$ & $3$ & $1$ & $5$ & $[RX]$ & $[RZ,RY,RZ]$ & \texttt{circ-CZ} & \texttt{tanh} \\
& \texttt{03b176ff} & $-500$ & $405$ & $6$ & $3$ & $32$ & $-100$ & $100$ & $3$ & $2$ & $5$ & $[RX]$ & $[RZ,RY,RZ]$ & \texttt{offset-CZ} & \texttt{tanh} \\
& \texttt{043b80e2} & $-500$ & $525$ & $6$ & $3$ & $32$ & $-100$ & $100$ & $3$ & $1$ & $7$ & $[RX]$ & $[RZ,RY,RZ]$ & \texttt{circ-CZ} & \texttt{tanh} \\
& \texttt{0ff9fb8e} & $-500$ & $285$ & $6$ & $3$ & $32$ & $-100$ & $100$ & $3$ & $2$ & $3$ & $[RX]$ & $[RZ,RY,RZ]$ & \texttt{offset-CZ} & \texttt{tanh} \\
& \texttt{10d163eb} & $-415.7750$ & $405$ & $6$ & $3$ & $32$ & $-100$ & $100$ & $3$ & $1$ & $5$ & $[RX]$ & $[RZ,RY,RZ]$ & \texttt{offset-CZ} & \texttt{tanh} \\
& \textbf{\texttt{1bacc532}} & $-196.4250$ & $405$ & $6$ & $3$ & $32$ & $-100$ & $100$ & $1$ & $2$ & $5$ & $[RX]$ & $[RZ,RY,RZ]$ & \texttt{circ-CZ} & \texttt{tanh} \\
& \texttt{24689683} & $-419.9000$ & $525$ & $6$ & $3$ & $32$ & $-100$ & $100$ & $1$ & $2$ & $7$ & $[RX]$ & $[RZ,RY,RZ]$ & \texttt{offset-CZ} & \texttt{tanh} \\
& \texttt{25417497} & $-415.7750$ & $405$ & $6$ & $3$ & $32$ & $-100$ & $100$ & $1$ & $1$ & $5$ & $[RX]$ & $[RZ,RY,RZ]$ & \texttt{offset-CZ} & \texttt{tanh} \\
& \texttt{283bce61} & $-428.6750$ & $285$ & $6$ & $3$ & $32$ & $-100$ & $100$ & $3$ & $2$ & $3$ & $[RX]$ & $[RZ,RY,RZ]$ & \texttt{circ-CZ} & \texttt{tanh} \\
& \texttt{43229ab1} & $-500$ & $525$ & $6$ & $3$ & $32$ & $-100$ & $100$ & $1$ & $1$ & $7$ & $[RX]$ & $[RZ,RY,RZ]$ & \texttt{circ-CZ} & \texttt{tanh} \\
& \texttt{4b0b88be} & $-204.2750$ & $405$ & $6$ & $3$ & $32$ & $-100$ & $100$ & $1$ & $1$ & $5$ & $[RX]$ & $[RZ,RY,RZ]$ & \texttt{circ-CZ} & \texttt{tanh} \\
& \texttt{5397ad81} & $-500$ & $285$ & $6$ & $3$ & $32$ & $-100$ & $100$ & $3$ & $1$ & $3$ & $[RX]$ & $[RZ,RY,RZ]$ & \texttt{offset-CZ} & \texttt{tanh} \\
& \texttt{54cda79e} & $-400.0750$ & $285$ & $6$ & $3$ & $32$ & $-100$ & $100$ & $1$ & $1$ & $3$ & $[RX]$ & $[RZ,RY,RZ]$ & \texttt{circ-CZ} & \texttt{tanh} \\
& \texttt{714952d8} & $-410.2500$ & $525$ & $6$ & $3$ & $32$ & $-100$ & $100$ & $3$ & $2$ & $7$ & $[RX]$ & $[RZ,RY,RZ]$ & \texttt{circ-CZ} & \texttt{tanh} \\
& \texttt{8eb3be2d} & $-500$ & $285$ & $6$ & $3$ & $32$ & $-100$ & $100$ & $1$ & $1$ & $3$ & $[RX]$ & $[RZ,RY,RZ]$ & \texttt{offset-CZ} & \texttt{tanh} \\
& \texttt{a37af56a} & $-196.4250$ & $405$ & $6$ & $3$ & $32$ & $-100$ & $100$ & $3$ & $2$ & $5$ & $[RX]$ & $[RZ,RY,RZ]$ & \texttt{circ-CZ} & \texttt{tanh} \\
& \texttt{ad8ff535} & $-500$ & $405$ & $6$ & $3$ & $32$ & $-100$ & $100$ & $1$ & $2$ & $5$ & $[RX]$ & $[RZ,RY,RZ]$ & \texttt{offset-CZ} & \texttt{tanh} \\
& \texttt{bca7063e} & $-428.6750$ & $285$ & $6$ & $3$ & $32$ & $-100$ & $100$ & $1$ & $2$ & $3$ & $[RX]$ & $[RZ,RY,RZ]$ & \texttt{circ-CZ} & \texttt{tanh} \\
& \texttt{c2866eeb} & $-500$ & $525$ & $6$ & $3$ & $32$ & $-100$ & $100$ & $1$ & $1$ & $7$ & $[RX]$ & $[RZ,RY,RZ]$ & \texttt{offset-CZ} & \texttt{tanh} \\
& \texttt{d4d301d4} & $-400.0750$ & $285$ & $6$ & $3$ & $32$ & $-100$ & $100$ & $3$ & $1$ & $3$ & $[RX]$ & $[RZ,RY,RZ]$ & \texttt{circ-CZ} & \texttt{tanh} \\
& \texttt{db609a2a} & $-419.9000$ & $525$ & $6$ & $3$ & $32$ & $-100$ & $100$ & $3$ & $2$ & $7$ & $[RX]$ & $[RZ,RY,RZ]$ & \texttt{offset-CZ} & \texttt{tanh} \\
& \texttt{e01d8771} & $-500$ & $525$ & $6$ & $3$ & $32$ & $-100$ & $100$ & $3$ & $1$ & $7$ & $[RX]$ & $[RZ,RY,RZ]$ & \texttt{offset-CZ} & \texttt{tanh} \\
& \texttt{ed91a81d} & $-500$ & $285$ & $6$ & $3$ & $32$ & $-100$ & $100$ & $1$ & $2$ & $3$ & $[RX]$ & $[RZ,RY,RZ]$ & \texttt{offset-CZ} & \texttt{tanh} \\
& \texttt{fc0dd296} & $-410.2500$ & $525$ & $6$ & $3$ & $32$ & $-100$ & $100$ & $1$ & $2$ & $7$ & $[RX]$ & $[RZ,RY,RZ]$ & \texttt{circ-CZ} & \texttt{tanh} \\
\bottomrule
\end{tabular}
\end{table*}

\begin{table*}[t!]
\centering
\caption[Ablation study for \texttt{QnRL} on the Atari environments.]{Ablation study for \texttt{QnRL} with varying circuit depth $L$, entanglement style $U_{\textrm{ent}}$, and convolutional features across the Atari environments. Compares evaluation performance results, total model parameters, and all hyperparameters. Selected best model configuration is highlighted in \textbf{bold}. Rows are sorted by model configuration identifier for each respective environment.\label{tab:qnrl:ablation:qnrl:Atari}}
\resizebox{\linewidth}{!}{
\begin{tabular}{lllllllllllllllllll}
\toprule
Environment & Model ID & Mean Eval Reward & Parameters & Observation dim. & $\abs{\setaction}$ & $\abs{\setreturn}$ & $z_{\textrm{min}}$ & $z_{\textrm{max}}$ & $n$ & $m$ & $L$ & $U_{\textrm{enc}}$ & $U_{\textrm{var}}$ & $U_{\textrm{ent}}$ & $\phi$ & Conv. Features & Kernels & Strides \\
\midrule
\multirow{24}{*}{\texttt{SpaceInvaders}} 
& \texttt{28cf3201} & $91.5$ & $69494$ & $4\times84\times84$ & $4$ & $32$ & $-10$ & $10$ & $1$ & $1$ & $7$ & $[RX]$ & $[RZ,RY,RZ]$ & \texttt{offset-CZ} & \texttt{tanh} & $[16,32,32]$ & $[8, 4, 3]$ & $[4, 2, 1]$ \\
& \texttt{2a976698} & $88.75$ & $69494$ & $4\times84\times84$ & $4$ & $32$ & $-10$ & $10$ & $1$ & $1$ & $7$ & $[RX]$ & $[RZ,RY,RZ]$ & \texttt{circ-CZ} & \texttt{tanh} & $[16,32,32]$ & $[8, 4, 3]$ & $[4, 2, 1]$ \\
& \texttt{3f1e580c} & $84.75$ & $172694$ & $4\times84\times84$ & $4$ & $32$ & $-10$ & $10$ & $1$ & $1$ & $5$ & $[RX]$ & $[RZ,RY,RZ]$ & \texttt{offset-CZ} & \texttt{tanh} & $[32,64,64]$ & $[8, 4, 3]$ & $[4, 2, 1]$ \\
& \texttt{4031053e} & $49.5$ & $172934$ & $4\times84\times84$ & $4$ & $32$ & $-10$ & $10$ & $1$ & $1$ & $7$ & $[RX]$ & $[RZ,RY,RZ]$ & \texttt{offset-CZ} & \texttt{tanh} & $[32,64,64]$ & $[8, 4, 3]$ & $[4, 2, 1]$ \\
& \texttt{4801d71d} & $99.25$ & $69254$ & $4\times84\times84$ & $4$ & $32$ & $-10$ & $10$ & $1$ & $1$ & $5$ & $[RX]$ & $[RZ,RY,RZ]$ & \texttt{circ-CZ} & \texttt{tanh} & $[16,32,32]$ & $[8, 4, 3]$ & $[4, 2, 1]$ \\
& \texttt{56217398} & $83.25$ & $484262$ & $4\times84\times84$ & $4$ & $32$ & $-10$ & $10$ & $1$ & $1$ & $7$ & $[RX]$ & $[RZ,RY,RZ]$ & \texttt{circ-CZ} & \texttt{tanh} & $[64,128,128]$ & $[8, 4, 3]$ & $[4, 2, 1]$ \\
& \texttt{5a7a5c51} & $97.25$ & $484022$ & $4\times84\times84$ & $4$ & $32$ & $-10$ & $10$ & $1$ & $1$ & $5$ & $[RX]$ & $[RZ,RY,RZ]$ & \texttt{circ-CZ} & \texttt{tanh} & $[64,128,128]$ & $[8, 4, 3]$ & $[4, 2, 1]$ \\
& \texttt{70aae234} & $73.5$ & $30830$ & $4\times84\times84$ & $4$ & $32$ & $-10$ & $10$ & $1$ & $1$ & $7$ & $[RX]$ & $[RZ,RY,RZ]$ & \texttt{offset-CZ} & \texttt{tanh} & $[8,16,16]$ & $[8, 4, 3]$ & $[4, 2, 1]$ \\
& \texttt{7fbdd5b1} & $53.5$ & $172934$ & $4\times84\times84$ & $4$ & $32$ & $-10$ & $10$ & $1$ & $1$ & $7$ & $[RX]$ & $[RZ,RY,RZ]$ & \texttt{circ-CZ} & \texttt{tanh} & $[32,64,64]$ & $[8, 4, 3]$ & $[4, 2, 1]$ \\
& \textbf{\texttt{8066810a}} & $154.5$ & $484262$ & $4\times84\times84$ & $4$ & $32$ & $-10$ & $10$ & $1$ & $1$ & $7$ & $[RX]$ & $[RZ,RY,RZ]$ & \texttt{offset-CZ} & \texttt{tanh} & $[64,128,128]$ & $[8, 4, 3]$ & $[4, 2, 1]$ \\
& \texttt{851875f5} & $85$ & $69014$ & $4\times84\times84$ & $4$ & $32$ & $-10$ & $10$ & $1$ & $1$ & $3$ & $[RX]$ & $[RZ,RY,RZ]$ & \texttt{offset-CZ} & \texttt{tanh} & $[16,32,32]$ & $[8, 4, 3]$ & $[4, 2, 1]$ \\
& \texttt{858f9b36} & $130.5$ & $172454$ & $4\times84\times84$ & $4$ & $32$ & $-10$ & $10$ & $1$ & $1$ & $3$ & $[RX]$ & $[RZ,RY,RZ]$ & \texttt{circ-CZ} & \texttt{tanh} & $[32,64,64]$ & $[8, 4, 3]$ & $[4, 2, 1]$ \\
& \texttt{89d9fdcc} & $152.5$ & $30350$ & $4\times84\times84$ & $4$ & $32$ & $-10$ & $10$ & $1$ & $1$ & $3$ & $[RX]$ & $[RZ,RY,RZ]$ & \texttt{circ-CZ} & \texttt{tanh} & $[8,16,16]$ & $[8, 4, 3]$ & $[4, 2, 1]$ \\
& \texttt{8eece259} & $40.75$ & $69014$ & $4\times84\times84$ & $4$ & $32$ & $-10$ & $10$ & $1$ & $1$ & $3$ & $[RX]$ & $[RZ,RY,RZ]$ & \texttt{circ-CZ} & \texttt{tanh} & $[16,32,32]$ & $[8, 4, 3]$ & $[4, 2, 1]$ \\
& \texttt{96ccd925} & $53.75$ & $483782$ & $4\times84\times84$ & $4$ & $32$ & $-10$ & $10$ & $1$ & $1$ & $3$ & $[RX]$ & $[RZ,RY,RZ]$ & \texttt{offset-CZ} & \texttt{tanh} & $[64,128,128]$ & $[8, 4, 3]$ & $[4, 2, 1]$ \\
& \texttt{a5fc5354} & $91$ & $172694$ & $4\times84\times84$ & $4$ & $32$ & $-10$ & $10$ & $1$ & $1$ & $5$ & $[RX]$ & $[RZ,RY,RZ]$ & \texttt{circ-CZ} & \texttt{tanh} & $[32,64,64]$ & $[8, 4, 3]$ & $[4, 2, 1]$ \\
& \texttt{a79207f7} & $92$ & $484022$ & $4\times84\times84$ & $4$ & $32$ & $-10$ & $10$ & $1$ & $1$ & $5$ & $[RX]$ & $[RZ,RY,RZ]$ & \texttt{offset-CZ} & \texttt{tanh} & $[64,128,128]$ & $[8, 4, 3]$ & $[4, 2, 1]$ \\
& \texttt{a93af179} & $98.5$ & $483782$ & $4\times84\times84$ & $4$ & $32$ & $-10$ & $10$ & $1$ & $1$ & $3$ & $[RX]$ & $[RZ,RY,RZ]$ & \texttt{circ-CZ} & \texttt{tanh} & $[64,128,128]$ & $[8, 4, 3]$ & $[4, 2, 1]$ \\
& \texttt{bdeace87} & $135.25$ & $69254$ & $4\times84\times84$ & $4$ & $32$ & $-10$ & $10$ & $1$ & $1$ & $5$ & $[RX]$ & $[RZ,RY,RZ]$ & \texttt{offset-CZ} & \texttt{tanh} & $[16,32,32]$ & $[8, 4, 3]$ & $[4, 2, 1]$ \\
& \texttt{be4b3f0b} & $60$ & $30350$ & $4\times84\times84$ & $4$ & $32$ & $-10$ & $10$ & $1$ & $1$ & $3$ & $[RX]$ & $[RZ,RY,RZ]$ & \texttt{offset-CZ} & \texttt{tanh} & $[8,16,16]$ & $[8, 4, 3]$ & $[4, 2, 1]$ \\
& \texttt{c09dc459} & $47.5$ & $30590$ & $4\times84\times84$ & $4$ & $32$ & $-10$ & $10$ & $1$ & $1$ & $5$ & $[RX]$ & $[RZ,RY,RZ]$ & \texttt{offset-CZ} & \texttt{tanh} & $[8,16,16]$ & $[8, 4, 3]$ & $[4, 2, 1]$ \\
& \texttt{c2786360} & $80.25$ & $30830$ & $4\times84\times84$ & $4$ & $32$ & $-10$ & $10$ & $1$ & $1$ & $7$ & $[RX]$ & $[RZ,RY,RZ]$ & \texttt{circ-CZ} & \texttt{tanh} & $[8,16,16]$ & $[8, 4, 3]$ & $[4, 2, 1]$ \\
& \texttt{e4c22ded} & $89$ & $172454$ & $4\times84\times84$ & $4$ & $32$ & $-10$ & $10$ & $1$ & $1$ & $3$ & $[RX]$ & $[RZ,RY,RZ]$ & \texttt{offset-CZ} & \texttt{tanh} & $[32,64,64]$ & $[8, 4, 3]$ & $[4, 2, 1]$ \\
& \texttt{f86255b5} & $39.5$ & $30590$ & $4\times84\times84$ & $4$ & $32$ & $-10$ & $10$ & $1$ & $1$ & $5$ & $[RX]$ & $[RZ,RY,RZ]$ & \texttt{circ-CZ} & \texttt{tanh} & $[8,16,16]$ & $[8, 4, 3]$ & $[4, 2, 1]$ \\
\cmidrule(lr){1-19}
\multirow{24}{*}{\texttt{Breakout}} 
& \texttt{05815f17} & $1.1$ & $22700$ & $4\times84\times84$ & $4$ & $32$ & $-10$ & $10$ & $1$ & $1$ & $7$ & $[RX]$ & $[RZ,RY,RZ]$ & \texttt{circ-CZ} & \texttt{tanh} & $[8,16,16]$ & $[8, 4, 3]$ & $[4, 2, 1]$ \\
& \texttt{069ad0e5} & $1.95$ & $53204$ & $4\times84\times84$ & $4$ & $32$ & $-10$ & $10$ & $1$ & $1$ & $3$ & $[RX]$ & $[RZ,RY,RZ]$ & \texttt{circ-CZ} & \texttt{tanh} & $[16,32,32]$ & $[8, 4, 3]$ & $[4, 2, 1]$ \\
& \texttt{1a8fc558} & $2.75$ & $141284$ & $4\times84\times84$ & $4$ & $32$ & $-10$ & $10$ & $1$ & $1$ & $7$ & $[RX]$ & $[RZ,RY,RZ]$ & \texttt{circ-CZ} & \texttt{tanh} & $[32,64,64]$ & $[8, 4, 3]$ & $[4, 2, 1]$ \\
& \texttt{2153e7e5} & $1.3$ & $22380$ & $4\times84\times84$ & $4$ & $32$ & $-10$ & $10$ & $1$ & $1$ & $3$ & $[RX]$ & $[RZ,RY,RZ]$ & \texttt{offset-CZ} & \texttt{tanh} & $[8,16,16]$ & $[8, 4, 3]$ & $[4, 2, 1]$ \\
& \texttt{21892ef5} & $0.4$ & $22380$ & $4\times84\times84$ & $4$ & $32$ & $-10$ & $10$ & $1$ & $1$ & $3$ & $[RX]$ & $[RZ,RY,RZ]$ & \texttt{circ-CZ} & \texttt{tanh} & $[8,16,16]$ & $[8, 4, 3]$ & $[4, 2, 1]$ \\
& \texttt{48b698e2} & $3.15$ & $421092$ & $4\times84\times84$ & $4$ & $32$ & $-10$ & $10$ & $1$ & $1$ & $5$ & $[RX]$ & $[RZ,RY,RZ]$ & \texttt{circ-CZ} & \texttt{tanh} & $[64,128,128]$ & $[8, 4, 3]$ & $[4, 2, 1]$ \\
& \texttt{49fee01c} & $1.85$ & $53204$ & $4\times84\times84$ & $4$ & $32$ & $-10$ & $10$ & $1$ & $1$ & $3$ & $[RX]$ & $[RZ,RY,RZ]$ & \texttt{offset-CZ} & \texttt{tanh} & $[16,32,32]$ & $[8, 4, 3]$ & $[4, 2, 1]$ \\
& \textbf{\texttt{4d2222f7}} & $4.9$ & $421092$ & $4\times84\times84$ & $4$ & $32$ & $-10$ & $10$ & $1$ & $1$ & $5$ & $[RX]$ & $[RZ,RY,RZ]$ & \texttt{offset-CZ} & \texttt{tanh} & $[64,128,128]$ & $[8, 4, 3]$ & $[4, 2, 1]$ \\
& \texttt{4dbdc680} & $0.85$ & $53364$ & $4\times84\times84$ & $4$ & $32$ & $-10$ & $10$ & $1$ & $1$ & $5$ & $[RX]$ & $[RZ,RY,RZ]$ & \texttt{circ-CZ} & \texttt{tanh} & $[16,32,32]$ & $[8, 4, 3]$ & $[4, 2, 1]$ \\
& \texttt{55bfe7f3} & $2.4$ & $53364$ & $4\times84\times84$ & $4$ & $32$ & $-10$ & $10$ & $1$ & $1$ & $5$ & $[RX]$ & $[RZ,RY,RZ]$ & \texttt{offset-CZ} & \texttt{tanh} & $[16,32,32]$ & $[8, 4, 3]$ & $[4, 2, 1]$ \\
& \texttt{631a6432} & $1.8$ & $140964$ & $4\times84\times84$ & $4$ & $32$ & $-10$ & $10$ & $1$ & $1$ & $3$ & $[RX]$ & $[RZ,RY,RZ]$ & \texttt{offset-CZ} & \texttt{tanh} & $[32,64,64]$ & $[8, 4, 3]$ & $[4, 2, 1]$ \\
& \texttt{67b28b7e} & $3$ & $53524$ & $4\times84\times84$ & $4$ & $32$ & $-10$ & $10$ & $1$ & $1$ & $7$ & $[RX]$ & $[RZ,RY,RZ]$ & \texttt{offset-CZ} & \texttt{tanh} & $[16,32,32]$ & $[8, 4, 3]$ & $[4, 2, 1]$ \\
& \texttt{6a003579} & $2.4$ & $141124$ & $4\times84\times84$ & $4$ & $32$ & $-10$ & $10$ & $1$ & $1$ & $5$ & $[RX]$ & $[RZ,RY,RZ]$ & \texttt{circ-CZ} & \texttt{tanh} & $[32,64,64]$ & $[8, 4, 3]$ & $[4, 2, 1]$ \\
& \texttt{6e970d6b} & $3.25$ & $421252$ & $4\times84\times84$ & $4$ & $32$ & $-10$ & $10$ & $1$ & $1$ & $7$ & $[RX]$ & $[RZ,RY,RZ]$ & \texttt{circ-CZ} & \texttt{tanh} & $[64,128,128]$ & $[8, 4, 3]$ & $[4, 2, 1]$ \\
& \texttt{76e665b4} & $1.4$ & $22540$ & $4\times84\times84$ & $4$ & $32$ & $-10$ & $10$ & $1$ & $1$ & $5$ & $[RX]$ & $[RZ,RY,RZ]$ & \texttt{offset-CZ} & \texttt{tanh} & $[8,16,16]$ & $[8, 4, 3]$ & $[4, 2, 1]$ \\
& \texttt{79548fdc} & $1.4$ & $22540$ & $4\times84\times84$ & $4$ & $32$ & $-10$ & $10$ & $1$ & $1$ & $5$ & $[RX]$ & $[RZ,RY,RZ]$ & \texttt{circ-CZ} & \texttt{tanh} & $[8,16,16]$ & $[8, 4, 3]$ & $[4, 2, 1]$ \\
& \texttt{82230a78} & $1.90$ & $140964$ & $4\times84\times84$ & $4$ & $32$ & $-10$ & $10$ & $1$ & $1$ & $3$ & $[RX]$ & $[RZ,RY,RZ]$ & \texttt{circ-CZ} & \texttt{tanh} & $[32,64,64]$ & $[8, 4, 3]$ & $[4, 2, 1]$ \\
& \texttt{932345d1} & $3.45$ & $141284$ & $4\times84\times84$ & $4$ & $32$ & $-10$ & $10$ & $1$ & $1$ & $7$ & $[RX]$ & $[RZ,RY,RZ]$ & \texttt{offset-CZ} & \texttt{tanh} & $[32,64,64]$ & $[8, 4, 3]$ & $[4, 2, 1]$ \\
& \texttt{95f17e2b} & $1.6$ & $53524$ & $4\times84\times84$ & $4$ & $32$ & $-10$ & $10$ & $1$ & $1$ & $7$ & $[RX]$ & $[RZ,RY,RZ]$ & \texttt{circ-CZ} & \texttt{tanh} & $[16,32,32]$ & $[8, 4, 3]$ & $[4, 2, 1]$ \\
& \texttt{bbc2fbfe} & $2.1$ & $420932$ & $4\times84\times84$ & $4$ & $32$ & $-10$ & $10$ & $1$ & $1$ & $3$ & $[RX]$ & $[RZ,RY,RZ]$ & \texttt{circ-CZ} & \texttt{tanh} & $[64,128,128]$ & $[8, 4, 3]$ & $[4, 2, 1]$ \\
& \texttt{d1528e08} & $1.05$ & $22700$ & $4\times84\times84$ & $4$ & $32$ & $-10$ & $10$ & $1$ & $1$ & $7$ & $[RX]$ & $[RZ,RY,RZ]$ & \texttt{offset-CZ} & \texttt{tanh} & $[8,16,16]$ & $[8, 4, 3]$ & $[4, 2, 1]$ \\
& \texttt{e10a9b8e} & $2$ & $141124$ & $4\times84\times84$ & $4$ & $32$ & $-10$ & $10$ & $1$ & $1$ & $5$ & $[RX]$ & $[RZ,RY,RZ]$ & \texttt{offset-CZ} & \texttt{tanh} & $[32,64,64]$ & $[8, 4, 3]$ & $[4, 2, 1]$ \\
& \texttt{e1ce14b5} & $2.6$ & $421252$ & $4\times84\times84$ & $4$ & $32$ & $-10$ & $10$ & $1$ & $1$ & $7$ & $[RX]$ & $[RZ,RY,RZ]$ & \texttt{offset-CZ} & \texttt{tanh} & $[64,128,128]$ & $[8, 4, 3]$ & $[4, 2, 1]$ \\
& \texttt{f779270b} & $3.65$ & $420932$ & $4\times84\times84$ & $4$ & $32$ & $-10$ & $10$ & $1$ & $1$ & $3$ & $[RX]$ & $[RZ,RY,RZ]$ & \texttt{offset-CZ} & \texttt{tanh} & $[64,128,128]$ & $[8, 4, 3]$ & $[4, 2, 1]$ \\
\bottomrule
\end{tabular}
}
\end{table*}

The last set of experiments are an ablation study across all the aforementioned environments, from which our ``best'' models are derived. The empirical results for our ablation study are shown in \cref{fig:qnrl:exp:ablation} and \cref{tab:qnrl:ablation:c51:Gym,tab:qnrl:ablation:c51:Atari,tab:qnrl:ablation:qnrl:Gym,tab:qnrl:ablation:qnrl:Atari}.
Looking at the evaluation performance for \texttt{C51} in \cref{tab:qnrl:ablation:c51:Gym,tab:qnrl:ablation:c51:Atari} and \cref{fig:qnrl:exp:ablation:c51:CartPole:eval-episode-reward,fig:qnrl:exp:ablation:c51:CliffWalking:eval-episode-reward,fig:qnrl:exp:ablation:c51:FrozenLake:eval-episode-reward,fig:qnrl:exp:ablation:c51:Acrobot:eval-episode-reward,fig:qnrl:exp:ablation:c51:SpaceInvaders:eval-episode-reward,fig:qnrl:exp:ablation:c51:Breakout:eval-episode-reward} and comparing it with the total model size in \cref{fig:qnrl:exp:ablation:c51:CartPole:total-model-parameter-count,fig:qnrl:exp:ablation:c51:CliffWalking:total-model-parameter-count,fig:qnrl:exp:ablation:c51:FrozenLake:total-model-parameter-count,fig:qnrl:exp:ablation:c51:Acrobot:total-model-parameter-count,fig:qnrl:exp:ablation:c51:SpaceInvaders:total-model-parameter-count,fig:qnrl:exp:ablation:c51:Breakout:total-model-parameter-count} we can see that our selection of the best models is fair because of either the significant reduction in evaluation score for reducing the hidden layer dimensions, or the dramatic increase in total parameter count by increasing the hidden layer dimensions. In particular, we see that in \texttt{CartPole} while the \texttt{C51} baseline can indeed perform slightly better than our \texttt{QnRL} by increasing the number of hidden dimensions, this results in over triple the model size and as such is not a fair comparison.
Looking at the evaluation performance for \texttt{QnRL} in \cref{tab:qnrl:ablation:qnrl:Gym,tab:qnrl:ablation:qnrl:Atari} and \cref{fig:qnrl:exp:ablation:qnrl:CartPole:eval-episode-reward,fig:qnrl:exp:ablation:qnrl:CliffWalking:eval-episode-reward,fig:qnrl:exp:ablation:qnrl:FrozenLake:eval-episode-reward,fig:qnrl:exp:ablation:qnrl:Acrobot:eval-episode-reward,fig:qnrl:exp:ablation:qnrl:SpaceInvaders:eval-episode-reward,fig:qnrl:exp:ablation:qnrl:Breakout:eval-episode-reward} and comparing it with the total model size in \cref{fig:qnrl:exp:ablation:qnrl:CartPole:total-model-parameter-count,fig:qnrl:exp:ablation:qnrl:CliffWalking:total-model-parameter-count,fig:qnrl:exp:ablation:qnrl:FrozenLake:total-model-parameter-count,fig:qnrl:exp:ablation:qnrl:Acrobot:total-model-parameter-count,fig:qnrl:exp:ablation:qnrl:SpaceInvaders:total-model-parameter-count,fig:qnrl:exp:ablation:qnrl:Breakout:total-model-parameter-count} we also see that our selection for the best models based on the combination of circuit layer depth $L$, expectation power $n$, moment index $m$, and entanglement style $U_{\textrm{ent}}$ is fair, resulting in the highest possible evaluation performance with the lowest possible model size and quantum circuit complexity.
\section{Conclusion}\label{sec:qnrl:conclusion}

In this paper, we have proposed \ac{qnrl}, a novel quantum framework leveraging the fundamental properties of quantum superposition and quantum entanglement for modeling and manipulating stochastic distributions natively on quantum hardware. We have also proposed \ac{quak}, a novel algorithm for simultaneously computing and comparing the $n$-th power of the $m$-th moments of multiple superimposed complex distributions entirely in quantum Hilbert space using nonlinear trace non-increasing Kraus operators. 
We have proven that through quantum superposition and quantum entanglement \ac{quak} parallelizes both the computation and comparison of distributional moments more efficiently, with fewer computational resources, and without sampling the distribution into classical space.
We have demonstrated that \ac{qnrl} can be used in \ac{drl} settings to learn quantum distributions of multiple random returns conditioned on given input observations, and that \ac{quak} can be used to form a purely-quantum policy from these quantum return distributions, thus distilling the action selection process down to measurement sampling of the quantum circuit.
Simply put, in \ac{qnrl} we learn conditional return distributions as quantum states, and we manipulate them within their Hilbert space to form an action distribution in situ. In other words, there are no intermediary steps to derive the policy, because the quantum return distribution is the policy.
We have shown that these quantum policies derived in situ from the quantum return distributions achieve up to $82.9\%$ higher evaluation scores, and with up to $94.3\%$ fewer parameters on average compared to the classical baseline. Further, we have demonstrated that the learned quantum distribution more accurately estimates the expected return for unseen observations and under varying stochastic conditions.
Lastly, we have shown that our \ac{qnrl} learns a mapping between input observations and their corresponding return distributions across the full range of the quantum Hilbert space. This realizes complex relationships between the observation and returns as amplitude and phase components of the learned distribution, which are not captured by purely classical nor classically-sampled quantum distributional approaches.

\appendices

\zcsetup{countertype={section=appendix}} 

\renewcommand\thefigure{\thesection.\arabic{figure}} 
\renewcommand\thetable{\thesection.\arabic{table}} 
\renewcommand\theequation{\thesection.\arabic{equation}} 
\renewcommand\thealgocf{\thesection.\arabic{algocf}}
\renewcommand\thesubsection{\thesection.\Alph{subsection}} 

\section{Amplitude kickback proof}\label{app:qnrl:proof:quak:moment}
\setcounter{figure}{0}
\setcounter{table}{0}
\setcounter{equation}{0}
\setcounter{algocf}{0}

\begin{proof}[Combined proof of \cref{thm:qnrl:quak:moment,thm:qnrl:quak:policy}]
A comparison of the $n$-th power of the $m$-th moments of the return distribution can be done entirely in quantum space by first applying a cascade of return-atom preparations conditioned on the action space, followed by a reset over the return-atoms space, producing a weighted sum of $m$-th moment expectations over the return-atoms raised to the power of the number of cascades $n$.
Here we note that it will be convenient to use the density matrix form for representing and manipulating quantum states, as will become evident during the proof.

We begin with an initial state of all zeros across the two combined spaces $\hilbert_{\setaction} \otimes \hilbert_{\setreturn}$ in density matrix form as follows:
\begin{equation}
\rho_{{\textrm{in}}} = \ket{0^{\otimes q_{\mathcal{A}} + q_{\mathcal{Z}}}}\bra{0^{\otimes q_{\mathcal{A}} + q_{\mathcal{Z}}}}.
\end{equation}
We then prepare a quantum state representing the initial distribution of actions according to \cref{eq:qnrl:U_a_x}:
\begin{align}
& \paren{U_{a \given x} \otimes \mathbb{I}^{\otimes q_{\setreturn}}} \rho_{{\textrm{in}}} \paren{U_{a \given x} \otimes \mathbb{I}^{\otimes q_{\setreturn}}}^{\dagger} \nonumber
\\ & = \sum_{a_i, a_j \in \setaction} \sqrt{p(a_i \given x) p(a_j \given x)} \ket{a_i, 0^{\otimes q_{\mathcal{Z}}}} \bra{a_j, 0^{\otimes q_{\mathcal{Z}}}} 
\\ & = \rho_{A \given x}, \label{eq:qnrl:proof:quak:rho_A_x}
\end{align}
which we refer to as the density matrix $\rho_{A \given x}$.
Next, we prepare a quantum state representing the conditional probability distribution of return-atoms for every action according to \cref{eq:qnrl:CU_z_x_a}:
\begin{align}
& CU^{\pi}_{z \given x, a} \rho_{A \given x} \paren{CU^{\pi}_{z \given x, a}}^{\dagger} \nonumber
\\ & = \sum_{\substack{a_i, a_k \in \mathcal{A} \\ z_j, z_l \in \mathcal{Z}}} \sqrt{g_{x}(a_i, z_j) g_{x}(a_k, z_l)} \ket{a_i, z_j} \bra{a_k, z_l}
\\ & = \rho_{Z^{\pi} \given x,A}, \label{eq:qnrl:proof:quak:rho_z_x_a}
\end{align}
which we refer to as the density matrix $\rho_{Z^{\pi} \given x,A}$, where $g_{x}(a, z) = p(a \given x) p(z \given x, a)$ for brevity.

Now, in order to evaluate the expectation of return atoms we need to represent the \emph{values} of the return atom distribution $\set{z}_{z \in \mathcal{Z}}$ in quantum space. In typical \ac{rl} settings, the \emph{raw values} of the returns are unbounded. Due to this explosive upper bound, the raw return values cannot be linearly mapped to a quantum operation.
In our setting, however, we are interested in the \emph{comparison} between expectation of returns amongst the action set. Because the set of return atoms is both discrete and explicitly defined in our setting, the range of the return-atom values is known a priori. In light of this, we can apply a normalization to the raw return-atom values to preserve their relative magnitudes, which allows us to represent them in quantum space as a bounded quantum operator.
Here is where we introduce our proposed \cref{eq:qnrl:U_v_z}, which defines a completely positive, Hermitian, and trace non-increasing operator $U_{v \given z}$ over the return-atoms Hilbert space $\hilbert_{\mathcal{Z}}$ that encodes the return-atom values. Recall that although $U_{v \given z}$ is not unitary, the state normalization will be resolved by applying the reset operation \cref{eq:qnrl:reset} on the return-atom space, as we will see shortly.

Applying $m$ repetitions of $U_{v \given z}$ \cref{eq:qnrl:U_v_z} on \cref{eq:qnrl:proof:quak:rho_z_x_a} produces the density matrix
\begin{align}
& \prod_{m} \bigl[ \mathbb{I}^{\otimes q_{\setaction}} \otimes U_{v \given z} \bigr]  \rho_{Z^{\pi} \given x,A} {\prod_{m} \bigl[ \mathbb{I}^{\otimes q_{\setaction}} \otimes U_{v \given z} \bigr]}^{\dagger} \nonumber
\\ & = \sum_{\substack{a_i, a_k \in \mathcal{A} \\ z_j, z_l \in \mathcal{Z}}} \sqrt{g_{x}(a_i, z_j) v_{z_j}^{m} g_{x}(a_k, z_l) v_{z_l}^{m}} \ket{a_i, z_j} \bra{a_k, z_l},
\\ & = \rho_{m,1}, \label{eq:qnrl:proof:quak:rho_m}
\end{align}
which we will refer to as $\rho_{m,1}$, where the sequence of operations 
\begin{equation}\label{eq:qnrl:proof:quak:Q}
    Q = \prod_{m} \bparen*{\mathbb{I}^{\otimes q_{\setaction}} \otimes U_{v \given z}} CU^{\pi}_{z \given x, a}
\end{equation}
have been applied only once on the initial action distribution state $\rho_{A \given x}$. We will use this short-hand form for brevity in the following derivations.

Now, to correct the normalization on the return-atoms space, we apply a reset operation according to \cref{eq:qnrl:reset} over the space of the return atoms $\hilbert_{\mathcal{Z}}$ to collapse their qubits.
Importantly, in this setting we do not care about the measurement result, only the collapse is desired.
So, applying this to our \ac{drl} setting, we can resume from $\rho_{m,1}$ \cref{eq:qnrl:proof:quak:rho_m} and reset the return-atom qubits in the following manner. We define the set of projectors $\set{\Pi_{z} = \mathbb{I}^{\otimes q_{\setaction}} \otimes \ket{z}\bra{z}}_{z\in \mathcal{Z}}$, which covers all possible return outcomes while doing nothing to the action space, and then apply these projectors as a \emph{partial} measurement on $\rho_{m,1}$ according to \cref{eq:qnrl:measurement} as follows. First, we compute the operator-sum representation (numerator of \cref{eq:qnrl:measurement})
\begin{align}
& \bfunc{\mathcal{E}_{\mathcal{Z}}}{\rho_{m,1}} \nonumber
\\&= \sum_{z \in \setreturn} \Pi_z \rho_{m,1} \Pi^{\dagger}_z
\\&= \sum_{z \in \setreturn} \paren*{\mathbb{I}^{\otimes q_{\setaction}} \otimes \ket{z}\bra{z}} \rho_{m,1} \paren*{\mathbb{I}^{\otimes q_{\setaction}} \otimes \ket{z}\bra{z}}
\\&= \sum_{\substack{a_i,a_k \in \setaction \\ z \in \setreturn}} \sqrt{g_{x}(a_i, z) g_{x}(a_k, z) v_{z}^{2m}} \ket{a_i, z} \bra{a_k, z},
\\& \textrm{(Collapses where} ~ z_j=z_l=z \textrm{)}\nonumber
\end{align}
followed by the normalization factor (denominator of \cref{eq:qnrl:measurement})
\begin{align}
& \bfunc{\textrm{tr}}{\bfunc{\mathcal{E}_{\mathcal{Z}}}{\rho_1}} \nonumber
\\&= \sum_{\substack{a \in \setaction \\ z \in \setreturn}} g_{x}(a,z) v_{z}^{m}
\\& \textrm{(Diagonal sum collapses were} ~ a_i=a_k=a \textrm{)} \nonumber
\\&= \sum_{a \in \setaction} p(a \given x) \sum_{z \in \setreturn} p(z \given x, a) v_{z}^{m}
\\&= \sum_{a \in \setaction} p(a \given x) \expect{(Z^{\pi})^{m} \given x,a},
\\& \textrm{(By definition of the moment of a random variable)} \nonumber
\end{align}
which produces the normalized mixed density matrix
\begin{align}
& \bfunc{M_{\mathcal{Z}}}{\rho_{m,1}} \nonumber
\\&= \frac{\bfunc{\mathcal{E}_{\mathcal{Z}}}{\rho_{m,1}}}{\bfunc{\textrm{tr}}{\bfunc{\mathcal{E}_{\mathcal{Z}}}{\rho_{m,1}}}}
\\&= \Bigl(\frac{1}{\sum_{a \in \setaction} p(a \given x) \underbrace{\expect{(Z^{\pi})^{m} \given x,a}}_{\textrm{$m$-th moment}}}\Bigr) \nonumber
\\&\quad \cdot \sum_{\substack{a_i,a_k \in \setaction \\ z \in \setreturn}} \sqrt{g_{x}(a_i, z) g_{x}(a_k, z) v_{z}^{2m}} \ket{a_i, z} \bra{a_k, z}
\\& = {\rho'}_{m,1}.
\end{align}
Notice that here, after applying $m$ repetitions of \cref{eq:qnrl:U_v_z} and performing a partial measurement on $\hilbert_{\mathcal{Z}}$, is where the $m$-th moment expectation of the random return $Z^{\pi}$ conditioned on the current state and action space $\mathcal{A}$, i.e., $\expect{(Z^{\pi})^{m} \given x,a}$, first appears.
We can now see a direct relationship between the choice of $m$ and the desired moment of the return distribution.
Next, we can find the reduced state of the action space after the partial measurement by taking the partial trace over the return space $\hilbert_{\setreturn}$ as follows:
\begin{align}
& \bfunc{\textrm{tr}_{\mathcal{Z}}}{{\rho'}_{m,1}} \nonumber
\\&= \sum_{z \in \setreturn} (\mathbb{I}^{\otimes q_{\setaction}} \otimes \bra{z}) {\rho'}_{m,1} (\mathbb{I}^{\otimes q_{\setaction}} \otimes \ket{z})
\\&= \Bigl(\frac{1}{\sum_{a \in \setaction} p(a \given x) \expect{(Z^{\pi})^{m} \given x,a}}\Bigr) \nonumber
\\&\quad \cdot \sum_{\substack{a_i,a_k \in \setaction \\ z \in \setreturn}} \sqrt{g_{x}(a_i, z) g_{x}(a_k, z) v_{z}^{2m}} \ket{a_i} \bra{a_k}
\end{align}
The final system state after applying the reset operation from \cref{eq:qnrl:reset} is therefore:
\begin{align}
& \bfunc{\textrm{Reset}_{\mathcal{Z}}}{\rho_{m,1}} \nonumber
\\&= \bfunc{\textrm{tr}_{\mathcal{Z}}}{{\rho'}_{m,1}} \otimes \ket{0^{\otimes q_{\mathcal{Z}}}}\bra{0^{\otimes q_{\mathcal{Z}}}}
\\&= \Bigl(\frac{1}{\sum_{a \in \setaction} p(a \given x) \expect{(Z^{\pi})^{m} \given x,a}}\Bigr) \nonumber
\\&\quad \cdot \sum_{\substack{a_i,a_k \in \setaction \\ z \in \setreturn}} \sqrt{g_{x}(a_i, z) g_{x}(a_k, z) v_{z}^{2m}} \ket{a_i,0^{\otimes q_{\mathcal{Z}}}} \bra{a_k,0^{\otimes q_{\mathcal{Z}}}}
\\&= \tilde{\rho}_{m,1},
\end{align}
which we refer to as $\tilde{\rho}_{m,1}$.
Notice that at this point the return space $\hilbert_{\setreturn}$ has been reduced to the all-zero state $\ket{0^{\otimes q_{\mathcal{Z}}}}\bra{0^{\otimes q_{\mathcal{Z}}}}$, which has a very similar form to $\rho_{A \given x}$ in \cref{eq:qnrl:proof:quak:rho_A_x}. This means after applying the reset the system state is placed back into a form after applying superposition of actions, and is thus receptive to another iteration of sequential return operations $Q$ from \cref{eq:qnrl:proof:quak:Q}.

If at this point we were to stop applying operations and measure the action qubits using projectors $\set{\Pi_{a} = \ket{a}\bra{a}}_{a \in \mathcal{A}}$, ignoring the return space $\hilbert_{\setreturn}$ because it is already zero, we would get the following mixed state in density matrix form
\begin{align}
& \bfunc{\mathcal{E}_{\mathcal{A}}}{\tilde{\rho}_{m,1}} \nonumber
\\&= \sum_{a \in \setaction} \Pi_a \bfunc{\textrm{tr}_{\mathcal{Z}}}{{\rho'}_{m,1}} \Pi^{\dagger}_a
\\&= \sum_{a \in \setaction} \Bigl(\frac{\sum_{z \in \setreturn} g_{x}(a, z) v_{z}^{m}}{\sum_{a' \in \setaction} p(a' \given x) \expect{(Z^{\pi})^{m} \given x,a'}}\Bigr) \ket{a} \bra{a}
\\&~~~~~\textrm{(Trace summations collapse where} ~ a_i = a_k = a \textrm{)}\nonumber
\\&= \sum_{a \in \setaction} \Bigl(\frac{p(a \given x) \sum_{z \in \setreturn} p(z \given x, a) v_{z}^{m}}{\sum_{a' \in \setaction} p(a' \given x) \expect{(Z^{\pi})^{m} \given x,a'}}\Bigr) \ket{a}\bra{a}
\\&= \sum_{a \in \setaction} \Bigl(\frac{p(a \given x) \expect{(Z^{\pi})^{m} \given x,a}}{\sum_{a' \in \setaction} p(a' \given x) \expect{(Z^{\pi})^{m} \given x,a'}}\Bigr) \ket{a}\bra{a},
\\& \textrm{(By definition of the moment of a random variable)} \nonumber
\end{align}
which is diagonal with normalization factor
\begin{align}
\bfunc{\textrm{tr}}{\bfunc{\mathcal{E}_{\mathcal{A}}}{\tilde{\rho}_{m,1}}} 
&=  \frac{\sum_{a \in \setaction} p(a \given x) \expect{(Z^{\pi})^{m} \given x,a}}{\sum_{a' \in \setaction} p(a' \given x) \expect{(Z^{\pi})^{m} \given x,a'}} 
\\& = 1,
\end{align}
producing the normalized state:
\begin{align}
& \bfunc{M_{\mathcal{A}}}{\tilde{\rho}_{m,1}} \nonumber
\\& = \frac{\bfunc{\mathcal{E}_{\mathcal{A}}}{\tilde{\rho}_{m,1}}}{\bfunc{\textrm{tr}}{\bfunc{\mathcal{E}_{\mathcal{A}}}{\tilde{\rho}_{m,1}}}}
\\& = \sum_{a \in \setaction} \Bigl(\frac{p(a \given x) \expect{(Z^{\pi})^{m} \given x,a}}{\sum_{a' \in \setaction} p(a' \given x) \expect{(Z^{\pi})^{m} \given x,a'}}\Bigr) \ket{a}\bra{a}.
\end{align}
This means the probability of measuring an arbitrary action $a_i$ according to \cref{def:qnrl:measurement} is
\begin{equation}\label{eq:qnrl:proof:quak:P_a_i_n1}
\bfunc{P}{\ket{a_i}\bra{a_i}} 
= \frac{p(a_i \given x) \overbrace{\expect{(Z^{\pi})^{m} \given x,a_i}}^{\textrm{exponent $n=1$}}}{\sum_{a' \in \setaction} p(a' \given x) \expect{(Z^{\pi})^{m} \given x,a'}}
\end{equation}
which is the \emph{normalized weighted $m$-th moment expectation of return-atoms for the given action $a_i$ relative to all other actions $\mathcal{A} \setminus \{a_i\}$}. 
Hence, we have shown that we can generate a normalized distribution over the $m$-th moments of the return-atom distribution by selecting the number of sequential applications of \cref{eq:qnrl:U_v_z} in the quantum system, as denoted by complete operator $\textrm{Reset}_{\setreturn} Q$, where $Q$ is from \cref{eq:qnrl:proof:quak:Q}.
We also call out the exponent of the expectation in \cref{eq:qnrl:proof:quak:P_a_i_n1}, which in this case is $n=1$. There is a relationship between the number of times $\textrm{Reset}_{\setreturn} Q$ has been applied on the input state and the exponent of the expectation, which will become increasingly clear in the following derivation.


If we do not measure the action space, and instead apply another iteration of cascaded return operations \cref{eq:qnrl:proof:quak:Q}, we get the following state in density matrix form

\begin{align}
& Q \tilde{\rho}_{m,1} Q^{\dagger} \nonumber
\\&= \Bigl(\frac{1}{\sum_{a \in \setaction} p(a \given x) \expect{(Z^{\pi})^{m} \given x,a}}\Bigr) \nonumber Q \Biggl[
\\&\quad \cdot \sum_{\substack{a_i,a_k \in \setaction \\ z_q \in \setreturn}} \sqrt{g_{x}(a_i, z_q) g_{x}(a_k, z_q) v_{z_q}^{2m}} \nonumber
\\&\quad \cdot \ket{a_i,0^{\otimes q_{\mathcal{Z}}}} \bra{a_k,0^{\otimes q_{\mathcal{Z}}}} \Biggr] Q^{\dagger}
\\& \textrm{(Pull out constant norm. factor, and relabel}~ z \to z_q \textrm{)}\nonumber
\\&= \Bigl(\frac{1}{\sum_{a \in \setaction} p(a \given x) \expect{(Z^{\pi})^{m} \given x,a}}\Bigr) \Biggl[ \nonumber
\\&\quad \cdot \sum_{\substack{a_i,a_k \in \setaction \\ z_q,z_w \in \setreturn}} \sqrt{g_{x}(a_i, z_q) g_{x}(a_k, z_q) v_{z_q}^{2m} p(z_w \given x,a_i) v_{z_w}^{m}} \nonumber
\\&\quad \cdot \ket{a_i,z_w} \bra{a_k,0^{\otimes q_{\mathcal{Z}}}} \Biggr] Q^{\dagger}
\\& \textrm{(Apply $Q$ on LHS, collapse where $a_i=a$ from \cref{eq:qnrl:CU_z_x_a}}, \nonumber
\\& ~~\textrm{and create new index $z_w=z$ from \cref{eq:qnrl:U_v_z})} \nonumber
\\&= \Bigl(\frac{1}{\sum_{a \in \setaction} p(a \given x) \expect{(Z^{\pi})^{m} \given x,a}}\Bigr) \nonumber
\\& \cdot \sum_{\substack{a_i,a_k \in \setaction \\ z_q,z_w,z_t \in \setreturn}} \Biggl[ g_{x}(a_i, z_q) g_{x}(a_k, z_q) v_{z_q}^{2m} \nonumber
\\& \cdot p(z_w \given x,a_i) v_{z_w}^{m} p(z_t \given x,a_k) v_{z_t}^{m} \Biggr]^{\frac{1}{2}} \ket{a_i,z_w} \bra{a_k,z_t}
\\& \textrm{(Apply $Q^{\dagger}$ on RHS, collapse where $a_k=a$ from \cref{eq:qnrl:CU_z_x_a}}, \nonumber
\\& ~~\textrm{and create new index $z_t=z$ from \cref{eq:qnrl:U_v_z})} \nonumber
\\&= \rho_{m,2}, \label{eq:qnrl:proof:quak:rho_m2}
\end{align}
which we will refer to as $\rho_{m,2}$.
We then apply the reset operation to the return space by first applying return projectors 
\begin{align}
&\bfunc{\mathcal{E}_{\mathcal{Z}}}{\rho_{m,2}} \nonumber
\\&= \sum_{z} \Pi_z \rho_{m,2} \Pi^{\dagger}_z
\\&= \Bigl(\frac{1}{\sum_{a \in \setaction} p(a \given x) \expect{(Z^{\pi})^{m} \given x,a}}\Bigr) \nonumber
\\& \cdot \sum_{\substack{a_i,a_k \in \setaction \\ z_q,z \in \setreturn}} \Biggl[ g_{x}(a_i, z_q) g_{x}(a_k, z_q) v_{z_q}^{2m} \nonumber
\\& \cdot p(z \given x,a_i) p(z \given x,a_k) v_{z}^{2m} \Biggr]^{\frac{1}{2}} \ket{a_i,z} \bra{a_k,z}
\\& \textrm{(Collapse where $z_w = z_t = z$)} \nonumber
\end{align}
which has the normalization factor
\begin{align}
& \bfunc{\textrm{tr}}{\bfunc{\mathcal{E}_{\mathcal{Z}}}{\rho_2}} \nonumber
\\&= \frac{\sum_{a,z_q,z \in \setreturn} g_{x}(a, z_q) v_{z_q}^{m} p(z \given x,a) v_{z}^{m}}{\sum_{a \in \setaction} p(a \given x) \expect{(Z^{\pi})^{m} \given x,a}}
\\& \textrm{(Trace summation collapse where $a_i = a_k = a$)} \nonumber
\\&= \frac{\sum_{a \in \setaction} p(a \given x) \sum_{z_q,z \in \setreturn} p(z_q \given x,a) v_{z_q}^{m} p(z \given x,a) v_{z}^{m}}{\sum_{a \in \setaction} p(a \given x) \expect{(Z^{\pi})^{m} \given x,a}}
\\&= \frac{\sum_{a \in \setaction} p(a \given x) \bigl[\sum_{z \in \setreturn} p(z \given x,a) v_{z}^{m}\bigr]^{2}}{\sum_{a \in \setaction} p(a \given x) \expect{(Z^{\pi})^{m} \given x,a}}
\\& \textrm{(Both return summations are identical, i.e., $z = z_q$)} \nonumber
\\&= \frac{\sum_{a \in \setaction} p(a \given x) \expect{(Z^{\pi})^{m} \given x,a}^{2}}{\sum_{a \in \setaction} p(a \given x) \expect{(Z^{\pi})^{m} \given x,a}}.
\\&\textrm{(By definition of the moment of a random variable)}\nonumber
\end{align}
Thus the normalized density matrix according to \cref{eq:qnrl:measurement} is
\begin{align}
&\bfunc{M_{\mathcal{Z}}}{\rho_{m,2}}
\\&= \frac{\bfunc{\mathcal{E}_{\mathcal{Z}}}{\rho_{m,2}}}{\bfunc{\textrm{tr}}{\bfunc{\mathcal{E}_{\mathcal{Z}}}{\rho_{m,2}}}}
\\&= \Bigl(\frac{1}{\sum_{a \in \setaction} p(a \given x) \expect{(Z^{\pi})^{m} \given x,a}^{2}}\Bigr) \nonumber
\\& \cdot \sum_{\substack{a_i,a_k \in \setaction \\ z_q,z \in \setreturn}} \Biggl[ g_{x}(a_i, z_q) g_{x}(a_k, z_q) v_{z_q}^{2m} \nonumber
\\& \cdot p(z \given x,a_i) p(z \given x,a_k) v_{z}^{2m} \Biggr]^{\frac{1}{2}} \ket{a_i,z} \bra{a_k,z}
\\&= {\rho'}_{m,2}
\end{align}
Tracing out the return-atoms system gives
\begin{align}
& \bfunc{\textrm{tr}_{\mathcal{Z}}}{{\rho'}_{m,2}} \nonumber
\\&= \sum_{z} (\mathbb{I}\otimes\bra{z}) {\rho'}_{m,2} (\mathbb{I}\otimes\ket{z})
\\&= \Bigl(\frac{1}{\sum_{a \in \setaction} p(a \given x) \expect{(Z^{\pi})^{m} \given x,a}^{2}}\Bigr) \nonumber
\\& \cdot \sum_{\substack{a_i,a_k \in \setaction \\ z_q,z \in \setreturn}} \Biggl[ g_{x}(a_i, z_q) g_{x}(a_k, z_q) v_{z_q}^{2m} \nonumber
\\& \cdot p(z \given x,a_i) p(z \given x,a_k) v_{z}^{2m} \Biggr]^{\frac{1}{2}} \ket{a_i} \bra{a_k}
\end{align}
and the final system state after applying the reset operation according to \cref{eq:qnrl:reset} is therefore
\begin{align}
&\bfunc{\textrm{Reset}_{\mathcal{Z}}}{\rho_{m,2}} \nonumber
\\&= \bfunc{\textrm{tr}_{\mathcal{Z}}}{{\rho'}_{m,2}} \otimes \ket{0^{\otimes q_{\mathcal{Z}}}}\bra{0^{\otimes q_{\mathcal{Z}}}}
\\&= \Bigl(\frac{1}{\sum_{a \in \setaction} p(a \given x) \expect{(Z^{\pi})^{m} \given x,a}^{2}}\Bigr) \nonumber
\\& \cdot \sum_{\substack{a_i,a_k \in \setaction \\ z_q,z \in \setreturn}} \Biggl[ g_{x}(a_i, z_q) g_{x}(a_k, z_q) v_{z_q}^{2m} \nonumber
\\& \cdot p(z \given x,a_i) p(z \given x,a_k) v_{z}^{2m} \Biggr]^{\frac{1}{2}} \ket{a_i,0^{\otimes q_{\mathcal{Z}}}} \bra{a_k,0^{\otimes q_{\mathcal{Z}}}}
\\&= \tilde{\rho}_{m,2}.
\end{align}

Importantly, notice that a squared version of the expectation over returns appears in the normalization term. From this we readily see a correlation between the power index of the expectation and the number of times the return operations have been applied (conveniently denoted in the subscript of the system state ${\rho'}_{m,2}$). This is even more prevalent if we apply a measurement on the action space as follows
\begin{align}
&\bfunc{\mathcal{E}_{\mathcal{A}}}{\tilde{\rho}_{m,2}}
\\&= \sum_{a} \ket{a}\bra{a} \bfunc{\textrm{tr}_{\mathcal{Z}}}{{\rho'}_{m,2}} \ket{a}\bra{a}
\\&= \Bigl(\frac{1}{\sum_{a \in \setaction} p(a \given x) \expect{(Z^{\pi})^{m} \given x,a}^{2}}\Bigr) \nonumber
\\& \cdot \sum_{\substack{a \in \setaction \\ z_q,z \in \setreturn}} g_{x}(a, z_q) v_{z_q}^{m} p(z \given x,a) v_{z}^{m} \ket{a} \bra{a}
\\& \textrm{(Sums collapse when $a_i = a_k = a$)} \nonumber
\\&= \Bigl(\frac{1}{\sum_{a \in \setaction} p(a \given x) \expect{(Z^{\pi})^{m} \given x,a}^{2}}\Bigr) \nonumber
\\& \cdot \sum_{a \in \setaction} p(a \given x) \sum_{z_q,z \in \setreturn} p(z_q \given x,a) v_{z_q}^{m} p(z \given x,a) v_{z}^{m} \ket{a} \bra{a}
\\&= \Bigl(\frac{1}{\sum_{a \in \setaction} p(a \given x) \expect{(Z^{\pi})^{m} \given x,a}^{2}}\Bigr) \nonumber
\\& \cdot \sum_{a \in \setaction} p(a \given x) \Bigl[ \sum_{z \in \setreturn} p(z \given x,a) v_{z}^{m} \Bigr]^{2} \ket{a} \bra{a}
\\& \textrm{(Both return summations are identical, i.e., $z = z_q$)} \nonumber
\\&= \Bigl(\frac{1}{\sum_{a \in \setaction} p(a \given x) \expect{(Z^{\pi})^{m} \given x,a}^{2}}\Bigr) \nonumber
\\& \cdot \sum_{a \in \setaction} p(a \given x) \expect{(Z^{\pi})^{m} \given x,a}^{2} \ket{a} \bra{a},
\\& \textrm{(By definition of the moment of a random variable)}
\end{align}
which we can readily see is both diagonal and normalized, i.e., $\bfunc{\textrm{tr}}{\bfunc{\mathcal{E}_{\mathcal{A}}}{\tilde{\rho_2}}} = 1$. Hence, according to \cref{eq:qnrl:measurement} we have the final measured state
\begin{align}
& \bfunc{M_{\mathcal{A}}}{\tilde{\rho}_{m,2}} \nonumber
\\& = \frac{\bfunc{\mathcal{E}_{\mathcal{A}}}{\tilde{\rho}_{m,2}}}{\bfunc{\textrm{tr}}{\bfunc{\mathcal{E}_{\mathcal{A}}}{\tilde{\rho}_{m,2}}}}
\\& = \sum_{a \in \setaction} \Bigl(\frac{p(a \given x) \expect{(Z^{\pi})^{m} \given x,a}^{2}}{\sum_{a' \in \setaction} p(a' \given x) \expect{(Z^{\pi})^{m} \given x,a'}^{2}}\Bigr) \ket{a}\bra{a}.
\end{align}
and the probability of measuring an arbitrary action $a_i$ according to \cref{eq:qnrl:measurement_outcome_prob} is therefore:
\begin{equation}
\bfunc{P}{\ket{a_i}\bra{a_i}} = \frac{p(a_i) \overbrace{\expect{Z^{m} \given a_i}^{2}}^{\textrm{exponent $n=2$}}}{\sum_{a' \in \setaction} p(a') \expect{Z^{m} \given a'}^2}.
\end{equation}
Notice that the exponent of the expectation, in this case $n=2$, matches the number of times the operation $\textrm{Reset}_{\setreturn} Q$ has been applied to the input state, as indicated by the subscript of $\rho_{m,2}$. 

Therefore, we have shown that we can control the exponent of the expectation in the resulting probability distribution through the number of cascades of $\textrm{Reset}_{\setreturn} Q$.
Formally, we define this cascaded quantum measurement operation over the joint action-return space $\hilbert_{\mathcal{A}} \otimes \hilbert_{\mathcal{Z}}$ as
\begin{equation}\label{eq:qnrl:proof:quak:result:0}
\begin{split}
& \textrm{QuAK} =
\\ & \underbrace{\prod_{n} \vphantom{\prod_{m}}}_{\textrm{Power}} \Bigl[ \underbrace{\textrm{Reset}_{\setreturn} \vphantom{\prod_{m}}}_{\textrm{Kickback}} \underbrace{\prod_{m} \bigl[ \mathbb{I} \otimes U_{v \given z} \bigr]}_{\textrm{Moment}} \underbrace{CU^{\pi}_{z \given x,a} \vphantom{\prod_{m}}}_{\textrm{Return dist.}} \Bigr] \bigl(\underbrace{U_{a \given x} \otimes \mathbb{I} \vphantom{\prod_{m}}}_{\textrm{Action dist.}}\bigr),
\end{split}
\end{equation}
which is the quantum state map
\begin{equation}\label{eq:qnrl:proof:quak:result:1}
\begin{split}
\textrm{QuAK}
& : \ket{0^{\otimes q_{\mathcal{A}} + q_{\mathcal{Z}}}}\bra{0^{\otimes q_{\mathcal{A}} + q_{\mathcal{Z}}}} 
\\ & \to \tilde{\rho}_{m,n} \otimes \ket{0^{\otimes q_{\setreturn}}} \bra{0^{\otimes q_{\setreturn}}},
\end{split}
\end{equation}
such that the action distribution is formulated as the density matrix
\begin{equation}\label{eq:qnrl:proof:quak:result:2}
\tilde{\rho}_{m,n}
= \frac{\sum_{a \in \mathcal{A}} p(a \given x) \bfunc{\mathbb{E}}{(Z^{\pi})^{m} \given x,a}^{n} \ket{a}\bra{a}}{\sum_{a' \in \mathcal{A}} p(a' \given x) \bfunc{\mathbb{E}}{(Z^{\pi})^{m} \given x,a'}^{n}},
\end{equation}
and the probability of measuring action $a \in \setaction$ is
\begin{equation}\label{eq:qnrl:proof:quak:result:3}
\begin{split}
    P[\ket{a}\bra{a}]
    & = \bfunc{\textrm{tr}}{\bfunc{M_{a \in \setaction}}{\tilde{\rho}_{m,n}}}
    \\ & = \frac{p(a \given x) \bfunc{\mathbb{E}}{(Z^{\pi})^{m} \given x,a}^{n}}{\sum_{a' \in \mathcal{A}} p(a' \given x) \bfunc{\mathbb{E}}{(Z^{\pi})^{m} \given x,a'}^{n}}.
\end{split}
\end{equation}
This pushes an initial all-zero state to a superposition of action states, where the probability of measuring any given action $a \in \setaction$ is the \emph{normalized weighted $m$-th moment of returns, raised to the power of $n$, relative to all other actions $\mathcal{A} \setminus \{a\}$}.
These proven definitions \cref{eq:qnrl:proof:quak:result:0,eq:qnrl:proof:quak:result:1,eq:qnrl:proof:quak:result:2} match our proposed \cref{eq:qnrl:quak:0,eq:qnrl:quak:1,eq:qnrl:rho_A_x_Z} in \cref{thm:qnrl:quak:moment}, and \cref{eq:qnrl:proof:quak:result:3} matches our proposed \cref{eq:qnrl:prob_a_x_Z} in \cref{thm:qnrl:quak:policy}.
Hence, we have shown that the computation and comparison of the $n$-power of the $m$-th moment of the return distributions conditioned on an environment state $x \in \setstate$ and the action set $\setaction$ can be parallelized entirely in the combined quantum space $\hilbert_{\mathcal{A}} \otimes \hilbert_{\mathcal{Z}}$ using \cref{eq:qnrl:quak:0,eq:qnrl:prob_a_x_Z}, without the need for downstream classical computation.
\end{proof}

\section{Broader {Q\MakeLowercase{u}AK} usage}

\Cref{thm:qnrl:quak:moment} is pertinent to applications in both \ac{drl}, and also broader statistical analysis settings, because it allows for comparison of the moments of many random variables within quantum space, that is to say, the entire operation is completely quantum, without the need for encoding/decoding into the classical space mid-way through the algorithm, which is both expensive and error prone due to stochastic noise. Specifically, the premier result of \cref{thm:qnrl:quak:moment} is its application in a more general class of nested distributional settings, which are \emph{distributions over distributions}. To be precise, our \ac{drl} setting can indeed be more generally described as a \emph{nested distribution of actions $\mathcal{A}$ over a distribution of returns $\mathcal{Z}$}.
As shown in \cref{thm:qnrl:quak:policy}, the probability of measuring any given state in the space $\hilbert_{\mathcal{A}}$ is directly proportional to the $m$-th moment expectation over the space $\hilbert_{\mathcal{Z}}$ conditioned on $\hilbert_{\mathcal{A}}$ according to the value mapping in \cref{eq:qnrl:U_v_z}.
Moreover, we can boost the likelihood of measuring actions with higher expectations by intelligently selecting the power index $n$ to ``spread'' the moment distribution, i.e., the probability larger moments get raised while others get lowered. This allows interplay between the relative strength of action probabilities and the number of shots of the quantum circuit; that is to say, higher values of $n$ can be chosen to lower the number of repeated shots required, and vice versa.

\section{Example Scenarios for {Q\MakeLowercase{u}AK}}

Here we discuss several example scenarios where our proposed \ac{quak} algorithm may be useful.
First, consider an \ac{rl} scenario where finding the action corresponding to the maximum average value of the return-atom distribution is desired, i.e., ${\arg\max}_{a \in \setaction}{\set*{\expect{Z^{\pi} \given x, a}}}$. In this setting, we do not care about the actual expected return value, only the relative distance between expected returns. Applying \cref{eq:qnrl:quak:0} with $n=1$ and $m=1$ will yield a distribution across every 1st moment, i.e., the \emph{mean}, of random returns. If the expectations across the actions are similar in value, then the likelihood of measuring any given action will be similar, that is, with a maximum value near $1/\abs{\mathcal{A}}$. As we push $n \to \infty$, the relative (normalized) distance between the expectations becomes greater, with a maximum value near $1$. Hence, the action associated with the maximum expected return can be found by choosing a sufficiently high $n$ such that the probability of measuring that action is boosted (approaching $1$) while others are lowered (approaching $0$) by only running a single shot of the quantum circuit.

Consider another example where instead we are interested in the similarity in expected returns, i.e., which actions have returns that are \emph{close to each other in relative magnitude}. This could be beneficial in situations where exploration in the environment is desired, in contrast to a greedy policy (i.e., always takes the perceived ``best'' action). In this setting, one could choose a very low value for $n$ and instead run many shots of the circuit to generate a probability distribution over the action space
\begin{equation*}
\set*{\frac{p(a_i \given x)\expect{Z^{\pi} \given x, a_i}}{\sum_{a \in \mathcal{A}} p(a \given x)\expect{Z^{\pi} \given x, a}}}_{a_i \in \mathcal{A}},
\end{equation*}
which can then be sampled by the agent to take a ``likely'' action.

As a final example, consider the comparison of the 2nd moments, i.e., the \emph{variance}, of the returns distribution. Selecting $m=2$ with low $n$ one can sample the circuit many times to generate the distribution
\begin{equation*}
\set*{\frac{p(a_i \given x)\expect{(Z^{\pi})^{2} \given x, a_i}}{\sum_{a \in \mathcal{A}} p(a \given x)\expect{(Z^{\pi})^{2} \given x, a}}}_{a_i \in \mathcal{A}},
\end{equation*}
which forms a policy based on the ``spread'' of the distribution support, and is useful in tasks where the consistency and precision of a random variable is of particular interest.

\section{Environment Specifications}\label{app:qnrl:env}
\setcounter{figure}{0}
\setcounter{table}{0}
\setcounter{equation}{0}
\setcounter{algocf}{0}

In this section, we provide a description of each environment used to train and evaluate our models.

\subsection{\texttt{CartPole}}

\begin{figure}[t!]
    \centering
    \fcolorbox{gray!50}{white}{\includegraphics[width=0.75\columnwidth]{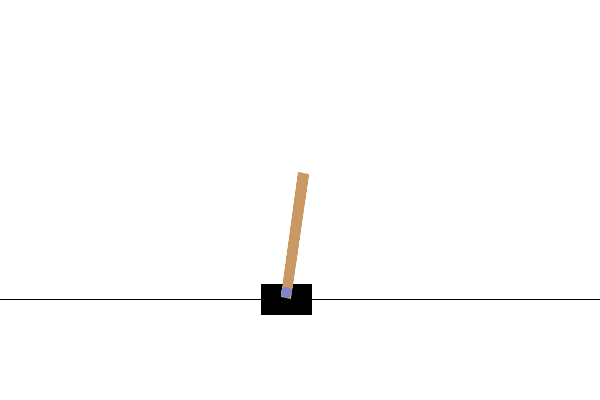}}
    \caption{Example of \texttt{CartPole} environment.\label{fig:qnrl:env:CartPole}}
\end{figure}

\begin{table}[t!]
    \caption{Specifications for \texttt{CartPole} environment.\label{tab:qnrl:env-spec:CartPole}}
    \centering
    \resizebox{\columnwidth}{!}{%
    \begin{tabular}{ll}
        \toprule
        Parameter & Value \\
        \midrule
        Episode Time limit ($T$) & 500 \\
        Observation at time $t$ & $\bm{x}_{t} \in \mathbb{R}^{4}$ \\
        Action at time $t$ & $a_{t} \in \set{\textrm{left},\textrm{right}}$\\
        Reward at time $t$ & $r_{t} = \begin{cases}
            +1, & \textrm{if pole is balanced}, \\
            0, & \textrm{otherwise}
        \end{cases}$\\
        Distribution atoms & $\abs{\setreturn} = 32$ \\
        Distribution range & $[z_{\textrm{min}}, z_{\textrm{max}}] = [-100, 100]$ \\
        \bottomrule
    \end{tabular}
    }
\end{table}

We train using the \texttt{CartPole} environment as proposed in \cite{Barto1983NeuronlikeAdaptiveElements}. This environment is an interesting case study for distributional learning because the observations are continuous real values, and the impact of chosen actions require several time steps to fully realize their consequences. An example of the \texttt{CartPole} environment is shown in \cref{fig:qnrl:env:CartPole}.
Agent observations are a vector $\bm{x} \in \mathbb{R}^{4}$ with 4 real-valued features. The features are: 
\begin{enumerate*}[label=\arabic*)]
    \item Cart position with range $[-4.8, 4.8]$,
    \item Cart velocity with range $(-\infty, \infty)$,
    \item Pole angle in \emph{radians} with range $[-0.418, 0.418]$, and
    \item Pole angular velocity with range $(-\infty, \infty)$.
\end{enumerate*}
The pole is considered \emph{balanced} if the pole angle feature stays within the range $[-.2095, .2095]$ radians, and the cart position feature stays within the range $[-2.4, 2.4]$.
For our quantum models we perform the static transformation 
\begin{equation}
    f(\bm{x}_{i}) = \bm{x}_{i} / \bm{v}_{i}, ~\textrm{where}~ \bm{v}_{i} = [2.4, 2.5, 0.21, 2.5],
\end{equation}
on the observation to convert its features into the effective range of $[-2,2]$.
An agent interacts with the environment by taking actions in the space $\setaction = \set{\textrm{left},\textrm{right}}$, which correspond to pushing their cart to the left and right respectively.
Each time step an agent is successful in keeping their pole balanced they receive a $+1$ episode reward. The episode terminates when an observation falls outside of the balanced range.
The agent's goal is to maximize their expected total episode reward (i.e., the number of time steps they are able to keep the pole balanced). The maximum reward for this environment is $T$, which is the time limit of environment configuration.
Our models represent distributions with support size $\abs{\setreturn} = 32$ and range $[z_{\textrm{min}}, z_{\textrm{max}}] = [-100, 100]$.
The details of the environment are summarized in \cref{tab:qnrl:env-spec:CartPole}.

We evaluate the agents using the \emph{total reward} metric, which aggregates all rewards over a single episode
\begin{equation}
    TR = \sum_{t=0}^{T-1} r_{t}, \label{eq:qnrl:env-spec:total-reward}
\end{equation}
where $t\in[0,T-1]$ is the episode time index, $T$ is the episode time limit, and $r_{t}$ is the agent reward at time $t$.

\subsection{\texttt{CliffWalking}}

\begin{figure}[t!]
    \centering
    \includegraphics[width=\columnwidth]{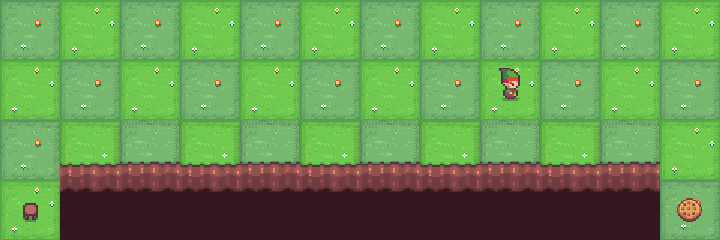}
    \caption{Example of \texttt{CliffWalking} environment.\label{fig:qnrl:env:CliffWalking}}
\end{figure}

\begin{table}[t!]
    \caption{Specifications for \texttt{CliffWalking} environment.\label{tab:qnrl:env-spec:CliffWalking}}
    \centering
    \resizebox{\columnwidth}{!}{%
    \begin{tabular}{ll}
        \toprule
        Parameter & Value \\
        \midrule
        Episode Time limit ($T$) & 99 \\
        Observation at time $t$ & $\bm{x}_{t} \in \set{0,1}^{48}$ (i.e, $4 \times 12$ grid) \\
        Action at time $t$ & $a_{t} \in \set{\textrm{up},\textrm{right},\textrm{down},\textrm{left}}$ \\
        Reward at time $t$ & $r_{t} = \begin{cases}
            -100, & \textrm{agent steps into the cliff}, \\
            -1, & \textrm{otherwise}
        \end{cases}$\\
        Distribution atoms & $\abs{\setreturn} = 32$ \\
        Distribution range & $[z_{\textrm{min}}, z_{\textrm{max}}] = [-100, 100]$ \\
        \bottomrule
    \end{tabular}
    }
\end{table}

We train using the \texttt{CliffWalking} environment as proposed in \cite{Sutton2018ReinforcementLearningIntroduction}. This environment is an interesting case study for building distributional models because of the discrete observation space, the large negative reward received for falling off of the cliff, and the resetting of the agent's position after falling. An example of the \texttt{CliffWalking} environment is shown in \cref{fig:qnrl:env:CliffWalking}.
Agent observations are a one-hot encoded vector $\bm{x}_{t} \in \set{0,1}^{48}$ representing the $\abs{\setstate} = 48$ discrete states of a $4 \times 12$ grid environment. A cliff runs along the bottom border of the grid, between cells $[3, 1] \to [3,10]$, and the goal is located at grid cell $[3,11]$. If the agent moves to a grid cell that is marked as a cliff, it is returned to the starting position and the episode continues until either the time limit is exceeded or the goal cell is reached. An agent interacts with the environment by taking actions in the space $\setaction = \set{\textrm{up},\textrm{right},\textrm{down},\textrm{left}}$, which correspond to moving to an adjacent grid cell along the designated cardinal direction. If an agent is adjacent to the perimeter of the grid and selects an action moving towards the perimeter this is registered as a valid action and the agent's position does not change. The agent receives a $-100$ reward for stepping into a grid cell marked as a cliff, and a $-1$ reward for all other time steps. The agent's goal is to maximize their expected total episode reward, which includes reaching the goal grid cell in as few time steps as possible without falling off the cliff. The maximum reward possible for this environment is $-13$, which corresponds to the minimal number of $13$ time steps required to reach the goal cell. 
Our models represent distributions with support size $\abs{\setreturn} = 32$ and range $[z_{\textrm{min}}, z_{\textrm{max}}] = [-100, 100]$.
We evaluate agents using the same total episode reward metric as outlined in \cref{eq:qnrl:env-spec:total-reward}.
The details of the environment are summarized in \cref{tab:qnrl:env-spec:CliffWalking}.

\subsection{\texttt{FrozenLake}}

\begin{figure}[t!]
    \centering
    \includegraphics[width=0.5\columnwidth]{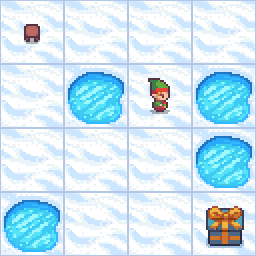}
    \caption{Example of \texttt{FrozenLake} environment.\label{fig:qnrl:env:FrozenLake}}
\end{figure}

\begin{table}[t!]
    \caption{Specifications for \texttt{FrozenLake} environment.\label{tab:qnrl:env-spec:FrozenLake}}
    \centering
    \tiny
    \resizebox{\columnwidth}{!}{%
    \begin{tabular}{ll}
        \toprule
        Parameter & Value \\
        \midrule
        Episode Time limit ($T$) & 100 \\
        Observation at time $t$ & $\bm{x}_{t} \in \set{0,1}^{16}$ (i.e., $4 \times 4$ grid) \\
        Action at time $t$ & $a_{t} \in \set{\textrm{left},\textrm{down},\textrm{right},\textrm{up}}$ \\
        Reward at time $t$ & $r_{t} = \begin{cases}
            +1, & \textrm{agent reaches the goal}, \\
            -0.2, & \textrm{agent reaches a hole}, \\
            -0.01, & \textrm{otherwise}, \\
        \end{cases}$\\
        Distribution atoms & $\abs{\setreturn} = 32$ \\
        Distribution range & $[z_{\textrm{min}}, z_{\textrm{max}}] = [-2, 2]$ \\
        \bottomrule
    \end{tabular}
    }
\end{table}

We train using a variation of the \texttt{FrozenLake} environment proposed in \cite{Brockman2016OpenAIGym}. This environment is an interesting case study for distributional learning because of the discrete observation space and small reward penalty for falling into holes throughout the environment. This requires agents to predict immediate negative rewards for falling into adjacent hazards, while also considering long-term rewards for planning to reach the goal cell. An example of the \texttt{FrozenLake} environment is shown in \cref{fig:qnrl:env:FrozenLake}.
Agent observations are a one-hot encoded vector $\bm{x}_{t} \in \set{0,1}^{16}$ representing the $\abs{\setstate} = 16$ discrete states of a $4 \times 4$ grid environment. The environment is a semi-frozen landscape with some ``hole'' cells marked as hazards, and a single goal cell located at the bottom-right cell $[3,3]$. Agents interact in the environment by taking actions in the space $\setaction = \set{\textrm{left},\textrm{down},\textrm{right},\textrm{up}}$, which correspond to moving to an adjacent grid cell in the designed cardinal direction. If an agent is adjacent to the perimeter of the grid and selects an action moving towards the perimeter this is registered as a valid action and the agent's position does not change. We use a custom reward schedule for this environment, whereby an receives a $+1$ reward for reaching the goal, a $-0.2$ reward for falling into a hole, and a $-0.01$ reward for all other situations. The latter small negative reward penalizes the number of time steps required to reach the goal. Importantly, the environment does not reset when an agent falls into a hole, instead continuing from the hole position as normal. The maximum episode reward for this environment and reward schedule is $0.95$, whereby the agent requires a minimum of 6 time steps to reach the goal position without falling into a hole.
Our models represent distributions with support size $\abs{\setreturn} = 32$ and range $[z_{\textrm{min}}, z_{\textrm{max}}] = [-2, 2]$.
We evaluate agents using the same total episode reward metric as outlined in \cref{eq:qnrl:env-spec:total-reward}.
The details of the environment are summarized in \cref{tab:qnrl:env-spec:FrozenLake}.

\subsection{\texttt{Acrobot}}

\begin{figure}[t!]
    \centering
    \fcolorbox{gray!50}{white}{\includegraphics[width=0.75\columnwidth]{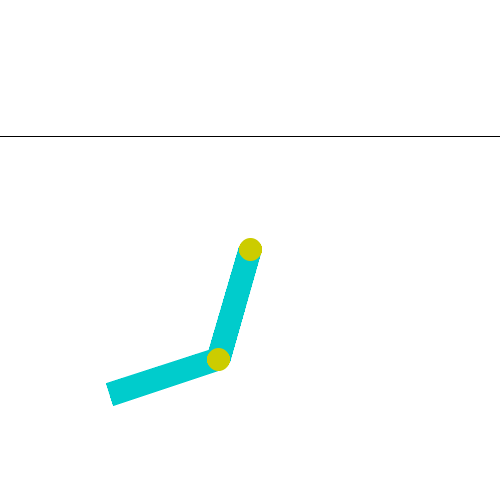}}
    \caption{Example of \texttt{Acrobot} environment.\label{fig:qnrl:env:Acrobot}}
\end{figure}

\begin{table}[t!]
    \caption{Specifications for \texttt{Acrobot} environment.\label{tab:qnrl:env-spec:Acrobot}}
    \centering
    \resizebox{\columnwidth}{!}{%
    \begin{tabular}{ll}
        \toprule
        Parameter & Value \\
        \midrule
        Episode Time limit ($T$) & 500 \\
        Observation at time $t$ & $\bm{x}_{t} \in \mathbb{R}^{6}$ \\
        Action at time $t$ & $a_{t} \in \set{\textrm{-1 torque},\textrm{0 torque},\textrm{+1 torque}}$ \\
        Reward at time $t$ & $r_{t} = \begin{cases}
            0, & \textrm{agent achieves target height}, \\
            -1, & \textrm{otherwise}, \\
        \end{cases}$\\
        Distribution atoms & $\abs{\setreturn} = 32$ \\
        Distribution range & $[z_{\textrm{min}}, z_{\textrm{max}}] = [-100, 100]$ \\
        \bottomrule
    \end{tabular}
    }
\end{table}

We train using the \texttt{Acrobot} environment as proposed in \cite{Sutton2018ReinforcementLearningIntroduction}. This environment serves as an interesting case study for distributional learning for its continuous real-valued observation space, complexity in managing the relationship between an actuated and non-actuated joints, and the long-term planning required for cascading small actions to effect large changes in the environment. Specifically, the arm requires many correlated actions be taken in succession to place it in an upright position, otherwise it will remain downward for the entirety of the episode. Hence, agents must navigate exploring a large continuous space to find an optimal policy. An example of the \texttt{Acrobot} environment is shown in \cref{fig:qnrl:env:Acrobot}.
Agent observations are a vector $\bm{x} \in \mathbb{R}^{6}$ with 6 real-valued features. The features are: 
\begin{enumerate*}[label=\arabic*)]
    \item Cosine of $\theta_{1}$ with range $[-1, 1]$,
    \item Since of $\theta_{1}$ with range $[-1, 1]$,
    \item Cosine of $\theta_{2}$ with range $[-1, 1]$,
    \item Sine of $\theta_{2}$ with range $[-1, 1]$,
    \item Angular velocity of $\theta_{1}$ with range $[-4 \pi, 4 \pi]$, and
    \item Angular velocity of $\theta_{2}$ with range $[-9 \pi, 9 \pi]$,
\end{enumerate*}
where $\theta_{1}$ is the angle of the first joint, and $\theta_{2}$ is relative to the angle of the first link. An angle of $\theta_{1} = 0$ indicates the first link is pointing downwards, and $\theta_{2} = 0$ indicates the same angle between the two links.
For our quantum models we perform the static transformation 
\begin{equation}
    f(\bm{x}_{i}) = \bm{x}_{i} / \bm{v}_{i}, ~\textrm{where}~ \bm{v}_{i} = [1, 1, 1, 1, 4 \pi, 9 \pi],
\end{equation}
on the observation to convert its features into the range of $[-1,1]$.
The goal state is when the free end of the arm reaches a target height derived as $-\cos{(\theta_{1})} - \cos{(\theta_{2} + \theta_{1})} > 1.0$.
The agent interacts within the environment by taking actions representing the torque applied on the actuated joint between the two links, which is the set $\setaction = \set{\textrm{-1 torque},\textrm{0 torque},\textrm{+1 torque}}$ respectively.
Every time step the agent receives a $-1$ reward, penalizing the duration of the episode, unless the goal is reached which terminates the episode with a $0$ reward.
The maximum reward threshold for this environment is $-100$, which indicates the agent reached the target height in 100 time steps.
Our models represent distributions with support size $\abs{\setreturn} = 32$ and range $[z_{\textrm{min}}, z_{\textrm{max}}] = [-100, 100]$.
We evaluate agents using the same total episode reward metric as outlined in \cref{eq:qnrl:env-spec:total-reward}.
The details of the environment are summarized in \cref{tab:qnrl:env-spec:Acrobot}.

\subsection{\texttt{SpaceInvaders}}

\begin{figure}[t!]
    \centering
    \includegraphics[width=0.6\columnwidth]{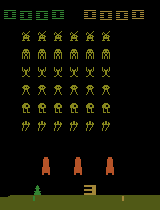}
    \caption{Example of \texttt{SpaceInvaders} environment.\label{fig:qnrl:env:SpaceInvaders}}
\end{figure}

\begin{table}[t!]
    \caption{Specifications for \texttt{SpaceInvaders} environment.\label{tab:qnrl:env-spec:SpaceInvaders}}
    \centering
    \resizebox{\columnwidth}{!}{%
    \begin{tabular}{ll}
        \toprule
        Parameter & Value \\
        \midrule
        Episode Time limit ($T$) & $10,000$ \\
        Frame Skip & $4$ \\
        Frame Stack & $4$ \\
        Screen Size & $84 \times 84$ \\
        Scale Observation & True \\
        Greyscale Observation & True \\
        Repeat Action Probability & $0\%$ \\
        Noop Reset & $30$ \\
        Terminate Episode on Life Loss & True \\
        Observation at time $t$ & $\bm{x}_{t} \in [0,1]^{4 \times 84 \times 84}$ \\
        Action at time $t$ & $a_{t} \in \set{\textrm{NOOP},\textrm{fire},\textrm{right},\textrm{left},\textrm{right fire},\textrm{left fire}}$ \\
        Reward at time $t$ & See \cite{Bellemare2013ArcadeLearningEnvironment,Machado2018RevisitingArcadeLearning} for score structure. \\
        Distribution atoms & $\abs{\setreturn} = 32$ \\
        Distribution range & $[z_{\textrm{min}}, z_{\textrm{max}}] = [-10, 10]$ \\
        \bottomrule
    \end{tabular}
    }
\end{table}

We train using the \texttt{SpaceInvaders} environment as implemented in the \ac{ale} \cite{Bellemare2013ArcadeLearningEnvironment,Machado2018RevisitingArcadeLearning}. This environment is a valuable case study for evaluating how distributional models learn from image data, how they handle a more sparse reward structure, and with randomness injected into the action selection phase via $\epsilon$-greedy.
An example of the \texttt{SpaceInvaders} environment is shown in \cref{fig:qnrl:env:SpaceInvaders}.
We use the widely adopted Atari preprocessing scheme of \cite{Machado2018RevisitingArcadeLearning} with a screen size of $84 \times 84$, greyscale image conversion, observation scaling to range $[0,255] \mapsto [0,1]$, 4-frame skipping with 4-frame stacking, a maximum of $30$ random \texttt{NOOP} actions at the start of each episode, and episode termination on life loss.
Agent observations are thus a matrix $\bm{x} \in [0,1]^{4 \times 84 \times 84}$ consisting of 4 frame-stacked greyscale images of shape $84 \times 84$.
Both the classical and quantum models employ a \ac{cnn} to learn an embedding for the image data, which is then mapped by a fully-connected layer into the dimension necessary for either the downstream classical distributional network or quantum circuit respectively.
The agent interacts within the environment by taking actions to move the player and fire the player weapon, which is the set $\mathcal{A} = \set{\textrm{NOOP},\textrm{fire},\textrm{right},\textrm{left},\textrm{right fire},\textrm{left fire}}$, where $a_{0}=\textrm{NOOP}$ is a special ``do nothing'' operation resulting in no player action taken.
The agent gains points by destroying the space invaders with varying assigned score values. We refer to \cite{Bellemare2013ArcadeLearningEnvironment,Machado2018RevisitingArcadeLearning} for a detailed explanation of the reward structure.
Our models represent distributions with support size $\abs{\setreturn} = 32$ and range $[z_{\textrm{min}}, z_{\textrm{max}}] = [-10, 10]$.
We evaluate agents using the same total episode reward metric as outlined in \cref{eq:qnrl:env-spec:total-reward}, which aggregates the total score received by the agent from destroying the space invaders.
The details of the environment are summarized in \cref{tab:qnrl:env-spec:SpaceInvaders}.

\subsection{\texttt{Breakout}}

\begin{figure}[t!]
    \centering
    \includegraphics[width=0.6\columnwidth]{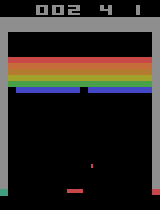}
    \caption{Example of \texttt{Breakout} environment.\label{fig:qnrl:env:Breakout}}
\end{figure}

\begin{table}[t!]
    \caption{Specifications for \texttt{Breakout} environment.\label{tab:qnrl:env-spec:Breakout}}
    \centering
    \resizebox{\columnwidth}{!}{%
    \begin{tabular}{ll}
        \toprule
        Parameter & Value \\
        \midrule
        Episode Time limit ($T$) & $10,000$ \\
        Frame Skip & $4$ \\
        Frame Stack & $4$ \\
        Screen Size & $84 \times 84$ \\
        Scale Observation & True \\
        Greyscale Observation & True \\
        Repeat Action Probability & $0\%$ \\
        Noop Reset & $30$ \\
        Terminate Episode on Life Loss & True \\
        Observation at time $t$ & $\bm{x}_{t} \in [0,1]^{4 \times 84 \times 84}$ \\
        Action at time $t$ & $a_{t} \in \set{\textrm{NOOP},\textrm{fire},\textrm{right},\textrm{left}}$ \\
        Reward at time $t$ & See \cite{Bellemare2013ArcadeLearningEnvironment,Machado2018RevisitingArcadeLearning} for score structure. \\
        Distribution atoms & $\abs{\setreturn} = 32$ \\
        Distribution range & $[z_{\textrm{min}}, z_{\textrm{max}}] = [-10, 10]$ \\
        \bottomrule
    \end{tabular}
    }
\end{table}

We train using the \texttt{Breakout} environment as implemented in \ac{ale} \cite{Bellemare2013ArcadeLearningEnvironment,Machado2018RevisitingArcadeLearning}. Similar to \texttt{SpaceInvaders}, this environment is also a valuable case study for evaluating how return distributions are learned from image data, with a sparse reward structure, and with action selection randomness.
An example of the \texttt{Breakout} environment is shown in \cref{fig:qnrl:env:Breakout}.
Here we also employ the Atari preprocessing scheme of \cite{Machado2018RevisitingArcadeLearning} using the same configuration as on \texttt{SpaceInvaders}, thus agent observations are a matrix $\bm{x} \in [0,1]^{4 \times 84 \times 84}$ consisting of 4 frame-stacked greyscale images of shape $84 \times 84$.
The agent interacts within the environment by taking actions to move the player and fire the player weapon, which is the set $\mathcal{A} = \set{\textrm{NOOP},\textrm{fire},\textrm{right},\textrm{left}}$, where $a_{0}=\textrm{NOOP}$ results in no player action taken.
The agent gains points by destroying the bricks on the wall with varying score values assigned to the colors of each brick. We refer to \cite{Bellemare2013ArcadeLearningEnvironment,Machado2018RevisitingArcadeLearning} for a detailed explanation of the reward structure.
Our models represent distributions with support size $\abs{\setreturn} = 32$ and range $[z_{\textrm{min}}, z_{\textrm{max}}] = [-10, 10]$.
We evaluate agents using the same total episode reward metric as outlined in \cref{eq:qnrl:env-spec:total-reward}, which aggregates the total score received by destroying bricks in the wall.
The details of the environment are summarized in \cref{tab:qnrl:env-spec:Breakout}.
\section{Hyperparameter Description}\label{app:qnrl:hyper}
\setcounter{figure}{0}
\setcounter{table}{0}
\setcounter{equation}{0}
\setcounter{algocf}{0}

The hyperparameters for each of the models trained in our experiments, as discussed in \cref{sec:qnrl:exp}, are shown in \cref{tab:qnrl:hyp:Gym} for the classic control and grid-world environments, and \cref{tab:qnrl:hyp:Atari} for the Atari environments respectively. For \texttt{C51} we vary the number of hidden dimensions via \texttt{douts} for the classic control and grid-world environments, and the number of \ac{cnn} features for the Atari environments. For our \texttt{QnRL} on the classic control and grid-world environments we vary the moment index $m$, expectation power $n$, circuit depth $L$, and entanglement style $U_{\textrm{ent}}$, and similarly on the Atari environments we vary the circuit depth, entanglement style, and add variation for the input embedding \ac{cnn} features.

\begin{table}[t!]
\centering
\caption{Description of hyperparameters for \texttt{C51} and \texttt{QnRL} used on the classic control and grid-world environments.\label{tab:qnrl:hyp:Gym}}
\resizebox{\columnwidth}{!}{%
\begin{tabular}{llp{0.7\columnwidth}}
\toprule
Model & Parameter & Description \\
\midrule
\multirow{11}{*}{\texttt{C51}}
& Total Train Steps & $100,000$ \\
& $\abs{x}$ & Input observation dimension. \\
& $\abs{\setaction}$ & Number of actions. \\
& $\abs{\setreturn}$ & Number of return atoms. \\
& $z_{\textrm{min}}$ & Minimum value for return-atom value range. \\
& $z_{\textrm{max}}$ & Maximum value for return-atom value range. \\
& \texttt{douts} & List containing number of hidden dimensions for each \ac{nn} layer. \\
& Optimizer & \texttt{Adam} \\
& Learning Rate & $lr = 10^{-3}$ \\
& Training seed & $k^{(t)} \in \set{0, 21, 42, 84}$ \\
& Evaluation seed & $\set{k^{(e)}_{i} = \texttt{jax.random.split}(k^{(e)}_{i-1})}_{i=1}^{10}$, where $k^{(e)}_{0}=1024$ \\
& Loss & \texttt{cross-entropy} \\
& $\epsilon_{\textrm{start}}$ & $1$; Starting $\epsilon$ value. \\
& $\epsilon_{\textrm{end}}$ & $0.05$; Final $\epsilon$ value. \\
& $\epsilon_{\textrm{decay}}$ & $50\%$; $\epsilon$ decay duration fraction of total time steps. \\
\cmidrule(lr){1-3}
\multirow{20}{*}{\texttt{QnRL}}
& Total Train Steps & $100,000$ \\
& $\abs{x}$ & Input observation dimension. \\
& $\abs{\setaction}$ & Number of actions. \\
& $\abs{\setreturn}$ & Number of return atoms. \\
& $z_{\textrm{min}}$ & Minimum value for return-atom value range. \\
& $z_{\textrm{max}}$ & Maximum value for return-atom value range. \\
& $q_{\setaction}$ & Number of action qubits: $q_{\setaction} = \ceil{\log_{2} \abs{\setaction}}$ \\
& $q_{\setreturn}$ & Number of return qubits: $q_{\setreturn} = \ceil{\log_{2} \abs{\setreturn}}$ \\
& $n$ & Exponent of the expectation of returns, i.e., $\mathbb{E}[\cdot]^{n}$. \\
& $m$ & Moment of the return distribution, i.e., $\mathbb{E}[(\cdot)^{m}]$. \\
& $L$ & Number of quantum circuit layers. \\
& $U_{\textrm{enc}}$ & Encoding circuit quantum operators. \\
& $U_{\textrm{var}}$ & Variational circuit quantum operators. \\
& $U_{\textrm{ent}}$ & Entanglement circuit type and coupling operator. \\
& $\phi$ & Activation function after multiplexing classical observations into quantum encoder. \\
& Optimizer & \texttt{AdamW} \\
& Learning Rate & $\set{lr_{\textrm{classical}} = 10^{-3}, lr_{\textrm{quantum}} = 10^{-2}}$ \\
& Training seed & $k^{(t)} \in \set{0, 21, 42, 84}$ \\
& Evaluation seed & $\set{k^{(e)}_{i} = \texttt{jax.random.split}(k^{(e)}_{i-1})}_{i=1}^{10}$, where $k^{(e)}_{0}=1024$ \\
& Loss & \texttt{cross-entropy} \\
& $\epsilon_{\textrm{start}}$ & $1$; Starting $\epsilon$ value. \\
& $\epsilon_{\textrm{end}}$ & $0.05$; Final $\epsilon$ value. \\
& $\epsilon_{\textrm{decay}}$ & $50\%$; $\epsilon$ decay duration fraction of total time steps. \\
\bottomrule
\end{tabular}
}
\end{table}

\begin{table}[t!]
\centering
\caption{Description of hyperparameters for \texttt{C51} and \texttt{QnRL} used on the Atari environments.\label{tab:qnrl:hyp:Atari}}
\resizebox{\columnwidth}{!}{%
\begin{tabular}{llp{0.7\columnwidth}}
\toprule
Model & Parameter & Description \\
\midrule
\multirow{11}{*}{\texttt{C51}}
& Total Train Steps & $200,000$ \\
& $\abs{x}$ & Input observation dimension. \\
& $\abs{\setaction}$ & Number of actions. \\
& $\abs{\setreturn}$ & Number of return atoms. \\
& $z_{\textrm{min}}$ & Minimum value for return-atom value range. \\
& $z_{\textrm{max}}$ & Maximum value for return-atom value range. \\
& \texttt{conv\_douts} & List containing number of feature dimensions for each \ac{cnn} layer. \\
& \texttt{conv\_kernel\_size} & List containing sizes of the convolutional filter for each \ac{cnn} layer. \\
& \texttt{conv\_strides} & List containing number of pixels to move the kernel for each \ac{cnn} layer. \\
& Optimizer & \texttt{Adam} \\
& Learning Rate & $lr = 2.5\times10^{-4}$ \\
& Training seed & $k^{(t)} \in \set{0, 21, 42, 84}$ \\
& Evaluation seed & $\set{k^{(e)}_{i} = \texttt{jax.random.split}(k^{(e)}_{i-1})}_{i=1}^{10}$, where $k^{(e)}_{0}=1024$ \\
& Loss & \texttt{cross-entropy} \\
& $\epsilon_{\textrm{start}}$ & $1$; Starting $\epsilon$ value. \\
& $\epsilon_{\textrm{end}}$ & $0.01$; Final $\epsilon$ value. \\
& $\epsilon_{\textrm{decay}}$ & $50\%$; $\epsilon$ decay duration fraction of total time steps. \\
\cmidrule(lr){1-3}
\multirow{20}{*}{\texttt{QnRL}}
& Total Train Steps & $200,000$ \\
& $\abs{x}$ & Input observation dimension. \\
& $\abs{\setaction}$ & Number of actions. \\
& $\abs{\setreturn}$ & Number of return atoms. \\
& $z_{\textrm{min}}$ & Minimum value for return-atom value range. \\
& $z_{\textrm{max}}$ & Maximum value for return-atom value range. \\
& $q_{\setaction}$ & Number of action qubits: $q_{\setaction} = \ceil{\log_{2} \abs{\setaction}}$ \\
& $q_{\setreturn}$ & Number of return qubits: $q_{\setreturn} = \ceil{\log_{2} \abs{\setreturn}}$ \\
& $n$ & Exponent of the expectation of returns, i.e., $\mathbb{E}[\cdot]^{n}$. \\
& $m$ & Moment of the return distribution, i.e., $\mathbb{E}[(\cdot)^{m}]$. \\
& $L$ & Number of quantum circuit layers. \\
& $U_{\textrm{enc}}$ & Encoding circuit quantum operators. \\
& $U_{\textrm{var}}$ & Variational circuit quantum operators. \\
& $U_{\textrm{ent}}$ & Entanglement circuit type and coupling operator. \\
& $\phi$ & Activation function after multiplexing classical observations into quantum encoder. \\
& \texttt{conv\_douts} & List containing number of feature dimensions for each \ac{cnn} layer. \\
& \texttt{conv\_kernel\_size} & List containing sizes of the convolutional filter for each \ac{cnn} layer. \\
& \texttt{conv\_strides} & List containing number of pixels to move the kernel for each \ac{cnn} layer. \\
& Optimizer & \texttt{AdamW} \\
& Learning Rate & $\set{lr_{\textrm{classical}} = 2.5\times10^{-4}, lr_{\textrm{quantum}} = 10^{-2}}$ \\
& Training seed & $k^{(t)} \in \set{0, 21, 42, 84}$ \\
& Evaluation seed & $\set{k^{(e)}_{i} = \texttt{jax.random.split}(k^{(e)}_{i-1})}_{i=1}^{10}$, where $k^{(e)}_{0}=1024$ \\
& Loss & \texttt{cross-entropy} \\
& $\epsilon_{\textrm{start}}$ & $1$; Starting $\epsilon$ value. \\
& $\epsilon_{\textrm{end}}$ & $0.01$; Final $\epsilon$ value. \\
& $\epsilon_{\textrm{decay}}$ & $50\%$; $\epsilon$ decay duration fraction of total time steps. \\
\bottomrule
\end{tabular}
}
\end{table}







\printbibliography

\begin{IEEEbiography}[{\includegraphics[width=1in,height=1.25in,clip,keepaspectratio]{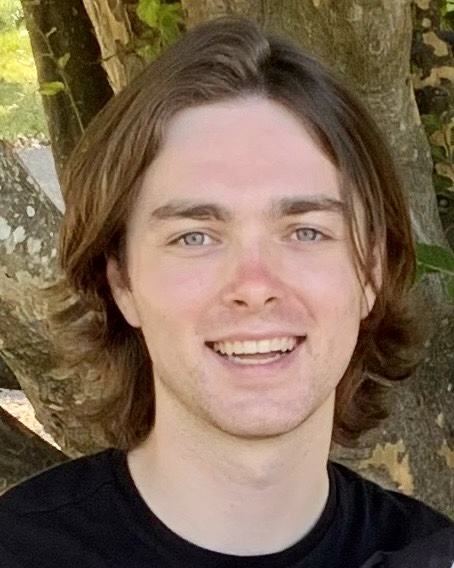}}]{Alexander DeRieux}(Graduate Student Member, IEEE) received the B.S.\@ degree in Electrical Engineering with honors and the B.S.\@ degree in Computer Science with honors from Virginia Tech in 2016, and the M.S.\@ degree in Electrical Engineering in 2022. He is currently a Ph.D.\@ candidate and Bradley Fellow in the Bradley Department of Electrical and Computer Engineering at the Virginia Tech Institute for Advanced Computing. Between 2017 and 2024, he worked as an Electronics Engineer for the U.S.\@ Naval Research Laboratory (NRL) Naval Center for Space Technology (NCST) Space Systems Development Division (SSDD). During this tenure, he researched, designed, and developed space-system technologies in the areas of rocketry, communications, optics, networking, tactical network modeling, surveillance and tracking, software engineering, Positioning, Navigation, and Timing (PNT), and Precise Time and Time Interval (PTTI) theory and applications. His research focuses on designing quantum-native artificial intelligence systems that leverage unique quantum mechanical properties to facilitate learning and operation naturally within the quantum domain.
\end{IEEEbiography}
\begin{IEEEbiography}[{\includegraphics[width=1in,height=1.25in,clip,keepaspectratio]{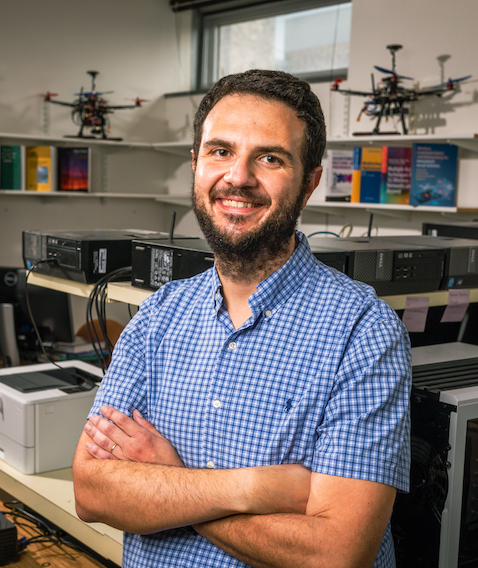}}]{Walid Saad}(Fellow, IEEE) (S’07, M’10, SM’15, F’19) received the Ph.D.\@ degree from the University of Oslo in 2010. He is currently a Professor at the Department of Electrical and Computer Engineering at Virginia Tech, where he leads the Network sciEnce, Wireless, and Security (NEWS) laboratory. He is also the Next-G Wireless research leader at Virginia Tech's Innovation Campus. His research interests include wireless networks (5G/6G/beyond), machine learning, game theory, security, unmanned aerial vehicles, semantic communications, cyber-physical systems, and network science. Dr. Saad is a Fellow of the IEEE. He was the author/co-author of eleven conference best paper awards. He is the recipient of the 2015 and 2022 Fred W. Ellersick Prize from the IEEE Communications Society, of the 2017 IEEE ComSoc Best Young Professional in Academia award, of the 2018 IEEE ComSoc Radio Communications Committee Early Achievement Award, and of the 2019 IEEE ComSoc Communication Theory Technical Committee. He was also a co-author of the 2019 and 2021 IEEE Communications Society Young Author Best Paper. He is the Editor-in-Chief for the IEEE Transactions on Machine Learning in  Communications and Networking.
\end{IEEEbiography}

\vfill 

\EOD

\end{document}